%% file: arxiv-paper.tex
\documentclass[letterpaper,11pt]{article}
\usepackage[round]{natbib}

\usepackage{booktabs} 
\usepackage[margin=1in]{geometry}

\input{aistats-defs}

\title{Scalable Algorithms for Individual Preference Stable Clustering}

\author{Ron Mosenzon\thanks{Toyota Technological Institute at Chicago (TTIC). Email: \href{ron.mosenzon@ttic.edu}{ron.mosenzon@ttic.edu}} \and Ali Vakilian\thanks{Toyota Technological Institute at Chicago (TTIC). Email: \href{mailto:vakilian@ttic.edu}{vakilian@ttic.edu}} }

\date{}

\begin{document}

\maketitle

\begin{abstract}
  In this paper, we study the individual preference (IP) stability, which is an notion capturing individual fairness and stability in clustering. Within this setting, a clustering is $\alpha$-IP stable when each data point's average distance to its cluster is no more than $\alpha$ times its average distance to any other cluster. In this paper, we study the natural local search algorithm for IP stable clustering. Our analysis confirms a $O(\log n)$-IP stability guarantee for this algorithm, where $n$ denotes the number of points in the input. Furthermore, by refining the local search approach, we show it runs in an almost linear time, $\Tilde{O}(nk)$. 
\end{abstract}

\section{Introduction}
In numerous human-centric clustering applications, the data points being clustered often represent agents each having distinct preferences.
In these contexts, ensuring fairness and stability necessitates clustering approaches that prioritize individual satisfaction of each agent over optimizing a global objective.
One such notion is that of \emph{individual preference stability (IP stability)}, which was recently introduced by~\cite{ahmadi2022individual}.
A dataset's clustering is IP stable when each data point finds its own cluster's average distance to be smaller than the average distance to any other cluster. This essentially means that every individual has a preference for its designated cluster, ensuring stability. The notion is closely related to envy-free clustering.
More formally, given a metric space $(X,d)$ on $n$ points and a partition of $X$ into $k$ disjoint non-empty clusters $(C_1,\ldots,C_k)$,
we say that: (1) a point $p \in C_i$ is $\alpha$-envious of a cluster $C_j$ if $p$'s average distance to the points in its own cluster (i.e., $C_i$) is more than $\alpha$ times larger than its average distance to the points in $C_j$; and, (2) the partition is an $\alpha$-IP stable clustering if no point is $\alpha$-envious of any cluster.

Perhaps one of the most natural way to find an $\alpha$-IP stable clustering is the following local search procedure: Start with an arbitrary $k$-clustering of $X$, and, as long as there exists some point $p$ that $\alpha$-envies some cluster $C$, choose such a pair $(p,C)$ and swap the point $p$ from its current cluster to the cluster $C$.
The study of local search for IP stable clustering is conceptually interesting for several reasons.
Firstly, its simplicity suggests a high likelihood of practical effectiveness. Secondly, it offers insights into the natural emergence of IP stable clusterings in real-world scenarios. For instance, in situations where agents might gravitate towards another cluster with more similar members, the dynamics resemble the local search approach, thereby leading to the formation of an IP stable clustering.
An example of this dynamic can be seen in social networks, where individuals often cluster with those who resonate with similar viewpoints or passions.

In this paper, we analyze the natural local search algorithm for this problem, as described in Algorithm \ref{:alg:natural-local-search}, and bound its approximation guarantee for general metric.
\begin{algorithm}[ht]
\caption{Local Search Meta-Algorithm.}\label{:alg:natural-local-search}
\KwData{$(X,d)$, $k \leq n =|X|$, and $\alpha \geq 1$.}
\KwResult{an $\alpha$-IP stable $k$-clustering of $X$.}
\textbf{initialize} an arbitrary $k$-clustering $\mathcal{C}$ of $X$. \\
\While{exists $p$ and $C'$ s.t. $\avg(p,C(p) \setminus \{p\}) > \alpha \cdot \avg(p,C')$}{
\textbf{move} $p$ from $C = C(p)$ to the cluster $C'$\\
}
\Return the current clustering
\end{algorithm}

\subsection{Problem Definition and Background}
Given a metric space $(X,d)$ with $|X| = n$, a $k$-clustering of $(X,d)$ is a partition of $X$ into $k$ disjoint non-empty clusters.
Furthermore, for any non-empty subset $S \subseteq X$ and point $p \in X$, we use $\avg(p,S)$ to denote the average distance of $p$ to the points in $S$,
\begin{align*}
 \avg(p,S) \defeq \frac{1}{|S|}\sum_{p' \in S}d(p,p').   
\end{align*}
\begin{definition}[IP Stability \citep{ahmadi2022individual}]\label{:def:IP-stable-clustering}
Given a metric space $(X,d)$ with $|X| = n$. A clustering $\cC$ of $X$ is said to be $\alpha$-IP stable if for every point $p \in X$ and every cluster $C' \in \cC$ where $C' \neq C(p)$, either $C(p) = \{p\}$ or $\avg(p,C(p) \setminus \{p\}) \leq \alpha \cdot \avg(p,C')$.
\end{definition}

\citet{ahmadi2022individual} showed that, given a metric space $(X,d)$ and a desired number of clusters $k$, it is not always possible to find a $1$-IP stable $k$-clustering of $X$, as such a clustering may not exists. Furthermore, it is NP-hard to decide whether a $1$-IP clustering exists for a given $(X,d)$ and $k$.
This NP-hardness has motivated the study of $\alpha$-IP stability for $\alpha > 1$.~\cite{ahmadi2022individual} proposed an $O(n)$-IP stable clustering for general metrics, by applying the standard $O(n)$-distortion embedding to one-dimensional Euclidean space~\citep[Theorem 2.1]{matouvsek1990bi}. They also gave a bicriteria approximation that discards an $\varepsilon$-fraction of the input and outputs a $O\left(\frac{\log^2 n}{\varepsilon}\right)$-IP stable clustering for the remaining. Very recently,~\cite{aamand2023constant}, via a clever ball carving technique, designed an algorithm that outputs an $O(1)$-IP stable clustering in all metric instance.  

\subsection{Our Results}
Our main contribution is the analysis of the natural and practical local search algorithm for IP-stable clustering.
\paragraph{Simple Local Search.} More precisely, we show that there exists $\alpha = O(\log n)$ such that the natural local search meta-algorithm (Algorithm \ref{:alg:natural-local-search}) always terminates when used with this $\alpha$.
\begin{restatable}{theorem}{NaturalLocalSearchTerminates}\label{:thm:natural-local-search-terminates}
    For every $n \geq 3$, the natural local search algorithm (Algorithm \ref{:alg:natural-local-search}) always terminates with parameter $\alpha \geq 2\log n$ on any metric space with $n$ points, regardless of the value of $k$.
\end{restatable}

Even though~\cref{:thm:natural-local-search-terminates} guarantees that the local search algorithm successfully finds an $O(\log n)$-IP stable clustering, it does not guarantee polynomial time running time. In fact, the worst-case runtime of this simple local search is $n^{\Theta(k)}$. This motivates our second main result, which is a fast algorithm that outputs a $O(\log n)$-IP stable $k$-clustering of any metric space.

\paragraph{Fast Algorithm.}
Building on the standard local search algorithm, we construct a $\Tilde{O}(nk)$-time algorithm that finds a $O(\log n)$-IP stable $k$-clustering.
Specifically, in the practical settings where $k = \mathrm{polylog}(n)$, the running time of this algorithm is near-linear in the number of points.
We remark that the running time is sub-linear in the input size, since the input consists of the pairwise distances of points. Our algorithm improves upon the runtime of~\citep{aamand2023constant} by a factor of $n$.

\begin{restatable}{theorem}{FastAlgorithmExists}\label{:thm:fast-algorithm-exists}
    There exists a Monte-Carlo randomized algorithm that, given a metric space $(X,d)$ with $|X|=n$ and a desired number of clusters $k \leq n$, finds a $O(\log n)$-IP stable $k$-clustering of $(X,d)$ in time $\Tilde{O}(nk)$. The error probability can be made to be an arbitrarily small polynomial at the cost of only a $\polylog(n)$ blowup in the running time.
\end{restatable}
Additionally, if the points in $X$ are located in small number $m \ll n$ of distinct \emph{locations}, then the algorithm from \cref{:thm:fast-algorithm-exists} can be easily modified so that the running time will scale with $m$ rather than with $n$. (i.e. the running time becomes $\Tilde{O}(mk)$, but the $\Tilde{O}$ notation still hides $\mathrm{polylog}(n)$ factors.)
This modification is performed by treating each location as if it is one \quotes{weighted} point, which means that the modified algorithm also has the ``consistency'' property that points in the same location are guaranteed to belong to the same cluster of the output $k$-clustering. We remark that unlike our result, the runtime of the algorithm of~\cite{aamand2023constant} scale with the total weight of points in the input and fails to satisfy the natural consistency property.

\paragraph{Bounding Approximation Factor.}
While the described results so far, as well as the guarantees of prior works, only provide an absolute IP-stability bound,~\cite{aamand2023constant} posed an open question whether it is possible to return a clustering whose IP-stability guarantee is within a constant factor of the the best achievable IP-stability guarantee for the given instance. In particular, if a given metric space $(X,d)$ admits an $\alpha^*$-IP stable clustering with $\alpha^* = o(1)$, it is possible to find a $O(\alpha^*)$-IP stable clustering of the instance. The problem is closely related to a notion of stability for clustering studied in a line of work by~\citep{daniely2012clustering,balcan2013clustering}. However, their guarantee requires a lower bound of $\beta n$ on the size of the minimum cluster in the underlying solution and the runtime of their algorithm exponentially depends on $\frac{n}{\beta} > k$. Here, we answer this question too. 

\begin{theorem} \label{:thm:when-there-is-a-very-good-clustering}
    There exists a deterministic algorithm that, given a metric space $(X,d)$ with $|X|=n$ and a desired number of clusters $k$, runs in polynomial time and returns a $k$-clustering $\cC$ with the guarantee that, if there exists a $\alpha^*$-IP stable $k$-clustering of $X$ for $\alpha^* < 0.001$, then $\cC$ is $(3\alpha^*)$-IP stable.
\end{theorem}

In particular, together with the $O(1)$-IP stable guarantee of the recent work of~\cite{aamand2023constant}, the above theorem implies $O(1)$-approximate IP stable clustering; there exists an algorithm that on every metric instance $(X,d)$ returns an $O(\alpha^*)$-IP stable clustering where $\alpha^*$ is the smallest value for which $(X,d)$ admits a $\alpha^*$-IP stable clustering. Note that $\alpha^*$ can be any value, either $\ge 1$ or $<1$. 

\paragraph{Other IP Stability Functions.}
Moreover, we study the IP-stability with respect to functions other than average.
\cite{aamand2023constant} proposed the general notion of $(\alpha, f)$-IP stability (or $f$-IP stability for brevity) as follows.
\begin{definition}
Given a metric space $(X,d)$ and a function $f:X \times 2^{X} \to \bbR_{\geq 0}$, a clustering $\cC$ of $(X,d)$ is $(\alpha, f)$-IP stable if, for every point $p \in X$, either $C(p)=\{p\}$ or for every cluster $C' \neq C(p)$ in $\cC$, $f(p,C(p)) \leq \alpha \cdot f(p,C')$. 
\end{definition}
We show that our analysis of the natural local search for $\avg$-IP stability, and our fast adaptation of the standard local search, can be extended to other special cases of $f$-IP stability through careful modification.
We use this framework in order to prove that the natural local search for $\median$-IP stability outputs an $O(1)$-approximate $\median$-IP stable clustering, and can be modified to run in polynomial time. \begin{theorem} \label{:thm:constant-median-IP}
For every metric space $(X,d)$, every non-empty subset $S \subseteq X$, and every point $p \in X$, let $\mathrm{median}(p,S)$ denote the $\ceil{|S|/2}{}$th smallest of the $|S|$ values $\{d(p,p')\}_{p' \in S}$.
Then, there exists a deterministic algorithm running in polynomial time that for a desired number of clusters $2 \leq k \leq |X|$, computes a $(O(1), \median)$-IP stable $k$-clustering of $X$.
\end{theorem}

Furthermore, we analyze $\max$-IP stability, which is the case where $f(p,C)\defeq\max_{p' \in C}d(p,p')$, and show that every metric space admits an exact $\max$-IP stable $k$-clustering for every value of $k$. 
\begin{theorem} \label{:thm:existence-of-max-IP}
For every metric space $(X,d)$ and every $2 \leq k \leq |X|$, there exists a $(1, \max)$-IP stable $k$-clustering of $X$.
\end{theorem}
This in particular, improves upon the $3$-approximate guarantee of~\cite{aamand2023constant} for max-IP stable $k$-clustering. Additionally, we note that for a related problem of min-IP stable $k$-clustering, exact polynomial time algorithms are known~\citep{aamand2023constant,laber2023optimization}.

\subsection{Our Techniques}
To prove our main result (\cref{:thm:natural-local-search-terminates}), we employ a potential function argument.
Specifically, we show that if there exists a potential function $\Phi:2^{X} \to \bbR_{\geq 0}$ with the property that, for every $S \subset X$ and every $p \in X \setminus S$,
\begin{equation} \label{:eq:potential-function-main-property}
    \avg(p,S) \leq \Phi(S \cup \{p\}) - \Phi(S) \leq \alpha \cdot \avg(p,S),
\end{equation}
then every step of the natural local search, Algorithm~\ref{:alg:natural-local-search}, reduces the sum of potentials of the clusters in the maintained $k$-clustering.
Furthermore, we show that when $\alpha = O(\log n)$, such a potential function exists.
The potential function $\Phi(C)$ that we choose is defined in terms of the average distances of the points $p \in C$ to the cluster $C$; $\sum_{p\in C} \avg(p,C)$.
A key observation in the proof of~\cref{:eq:potential-function-main-property} is~\cref{:lem:average-distance-lemma} which relates the value of $\avg(p,C)$ to the value of $\Phi(C)$.

\paragraph{Faster Algorithm.}
At its core, the algorithm from \cref{:thm:fast-algorithm-exists} is based on the natural local search, but with some modifications that speed it up significantly.
The idea behind the main adaptation is to identify the cases in which a local update does not make large enough progress with respect to the potential function $\Phi$, and use a different type of local step to handle these cases.
This modification by itself reduces the running time to $\Tilde{O}(n^2 k)$.
Shaving off the last factor of $n$ requires a few more modifications. 
Essentially, we exploit the fact that $\avg(p,C)$ values do not often change drastically, enabling the use of calculated values for $\avg(p,C)$ as estimates in multiple forthcoming steps, instead of refreshing these values after every single step.
Furthermore, once these previously computed $\avg(p,C)$ no longer provide good enough estimates for the current values, we show how to quickly recompute these values, employing a statistical technique known as Importance Sampling.

\paragraph{Approximate IP Stability.}
In contrast to our other results, the algorithm for \cref{:thm:when-there-is-a-very-good-clustering} is not a local search algorithm. Instead, it uses dynamic programming, while exploiting the structure that a highly IP-stable clustering imposes on the metric space.
Indeed, there is little hope of using a local search algorithm similar to Algorithm \ref{:alg:natural-local-search} with $\alpha < 1$, since it may get stuck repeatedly moving the same point $p$ between two clusters that are equidistant from $p$.
We remark that, with a more optimized implementation of the algorithm from \cref{:thm:when-there-is-a-very-good-clustering}, a running time of $O(n^2)$ is achievable.

\paragraph{Median-IP Stability.}
The median-IP stable algorithm of \cref{:thm:constant-median-IP} 
is a simplified version of the algorithm from \cref{:thm:fast-algorithm-exists}.
In fact, in the appendix, we present a framework for constructing such algorithms for general IP stability functions, provided that an appropriate potential function exists.
While \cref{:thm:constant-median-IP} only guarantees polynomial running time, we note that an upper-bound of $\Tilde{O}(n^3)$ can be shown with a tighter analysis of the same algorithm, and that an algorithm with a running time of $\Tilde{O}(n^2 k)$ can be achieved if the stability guarantee is relaxed from $\alpha=O(1)$ to $\alpha=O(\log n)$.

\paragraph{Max-IP Stability.}
The result presented in \cref{:thm:existence-of-max-IP} is derived via a natural local search algorithm. Analogously to Algorithm \ref{:alg:natural-local-search}, there is no guarantee that this local search algorithm runs in polynomial time.

\subsection{Related Work}
\paragraph{Clustering Stability.} 
Designing efficient clustering algorithms that consider various notions of stability is a well-studied problem. Notably, the concept of ``average stability''~\citep{balcan2008discriminative} is particularly relevant to our model.
For a comprehensive survey on stability in clustering, refer to~\citep{awasthi2014center}. 
While existing research on clustering stability~\citep{ackerman2009clusterability,awasthi2012center,bilu2012stable, balcan2013clustering, balcan2016clustering,makarychev2016metric} mainly drives the development of faster algorithms by leveraging stability conditions, IP stable clustering explores whether such stability properties can be identified in generic instances, either exactly or approximately.

\paragraph{Individual Fairness in Clustering.} As introduced by~\citep{ahmadi2022individual}, a primary motivation for IP-stability lies in its connection to individual fairness. Individual fairness was initially formulated in the seminal work of~\cite{dwork2012fairness}, aiming to ensure classification problems provided ``similar predictions'' for ``similar points''. In recent years, several notions of individual fairness for clustering have been proposed~\citep{jung2020,anderson2020distributional,brubach2020pairwise, chakrabarti2022new,kar2023feature}. In a line of research by~\cite{jung2020,mahabadi2020,negahbani2021better,vakilian2022improved}, a notion of fairness has been studied, wherein each point expects to find its closest center at a distance comparable to that from its $(n/k)$-th closest neighbor.~\cite{brubach2020pairwise,brubach2021fairness} explored settings where individuals benefit from being clustered together, aiming to discover a clustering where all individuals gain equitable community benefits. Lastly, \cite{anderson2020distributional} introduced a formulation whereby the output is a distribution over centers, and ``similar'' points must have ``similar'' centers distributions. 

Besides, clustering has also been explored in the context of various other fairness notions, including equitable representation in clusters~\citep{chierichetti2017fair,bera2019fair,bercea2019cost,backurs2019scalable,schmidt2019fair,ahmadian2019clustering,dai2022fair}, equitable representation in centers~\citep{krishnaswamy2011matroid,chen2016matroid,krishnaswamy2018constant,kleindessner2019fair,chiplunkar2020solve,jones2020fair,hotegni2023approximation}, proportional fairness~\citep{chen2019proportionally,micha2020proportionally}, and min-max fairness~\citep{abbasi2020fair,ghadiri2020fair,makarychev2021approximation,chlamtavc2022approximating,ghadiri2022constant}.

IP-stability is also relevant to the notion of {\em envy-freeness}. Within conventional fair division problems, such as cake cutting~\citep{robertson1998cake,procaccia2013cake} and rent division~\citep{edward1999rental,gal2016fairest}, envy-freeness requires that each participant should (weakly) favor their own allocation over that of others~\citep{foley1966resource,varian1973equity}. More recently, the concept was studied for the problem of classification~\citep{zafar2017parity,balcan2019envy}.

\paragraph{Hedonic Games in Clustering.} 
Hedonic games represent another game-theoretic approach to studying clustering. This approach has been explored in several studies, including those by~\citep{dreze1980hedonic, bogomolnaia2002stability, elkind2020price,stoica2018strategic}.
Within these frameworks, players strategically form coalitions, aiming to optimize their respective utilities. Our work diverges from this paradigm as we do not treat data points as selfish players. 

\subsection{Paper Outline}
In~\cref{:sec:analysis-of-local-search}, we analyze the correctness of the natural local search (\cref{:thm:natural-local-search-terminates}).
In \cref{:sec:basic-merge-and-split-algorithm}, we present a slower but simpler version of the algorithm from \cref{:thm:fast-algorithm-exists}, achieving a running time of $\Tilde{O}(n^2 k)$.
This simpler variant highlights the most crucial modification that distinguishes the final $\Tilde{O}(nk)$ time algorithm from the natural local search algorithm, without being obscured by the numerous details inherent in the full algorithm.
In the appendix, we show how to shave off the last factor $n$ from the running time, resulting in the fast algorithm from \cref{:thm:fast-algorithm-exists}.
All missing proofs are also in the appendix.

\subsection{Preliminaries}
In all algorithms in this paper, the input metric space $(X,d)$ is given in the form of query access to the distance metric $d$, where each query takes $\Theta(1)$ time.
We always use $n = |X|$ to denote the number of points in the input.
\begin{definition}[average distance]
Given a metric space $(X,d)$, a non-empty subset $S \subseteq X$, and a point $p \in X$, we let $\avg(p,S)$ denote the average distance of $p$ to the points in $S$, $\avg(p,S) \defeq \frac{1}{|S|}\sum_{p' \in S}d(p,p')$.
\end{definition}

The following lemma is a crucial component in our analysis of the natural local search algorithm and the subsequent fast modifications of it.
This lemma is derived in an elementary manner, utilizing the triangle inequality.
\begin{lemma}\label{:lem:average-distance-lemma}
    In every metric space $(X,d)$, for every non-empty subset $S \subseteq X$, and every point $p \in X$,
    \[
     \frac{1}{|S|}\sum_{p' \in S} \avg(p',S) \leq 2 \cdot \avg(p,S)
    \]
\end{lemma}

First, we prove the following observation, and then we show that it implies the lemma.
\begin{observation}\label{:obs:triangle-inequality-for-average-distance}
    For every metric space $(X,d)$, every point $p \in X$, and every two non-empty subsets $S_1,S_2 \subseteq X$,
    \[
     \frac{1}{|S_1|\cdot|S_2|}\sum_{p_1 \in S_1}\sum_{p_2 \in S_2} d(p_1,p_2) \leq \avg(p,S_1) + \avg(p,S_2).
    \]
\end{observation}

\begin{proof}
For every metric space $(X,d)$, point $p \in X$, and non-empty sets $S_1,S_2 \subseteq X$, by the triangle inequality,
\begin{align*}
    \frac{1}{|S_1|\cdot|S_2|}\sum_{p_1 \in S_1}\sum_{p_2 \in S_2} d(p_1,p_2)
    &\leq \frac{1}{|S_1|\cdot|S_2|}\sum_{p_1 \in S_1}\sum_{p_2 \in S_2} (d(p_1,p) + d(p,p_2))\\
    &= \frac{1}{|S_1|} \sum_{p_1 \in S_1} d(p_1,p) + \frac{1}{|S_2|} \sum_{p_2 \in S_2} d(p,p_2) \\
    &= \avg(p,S_1) + \avg(p,S_2).
\end{align*}
\end{proof}
\begin{proof}[Proof of \cref{:lem:average-distance-lemma}]
By~\cref{:obs:triangle-inequality-for-average-distance},
\[
 \frac{1}{|S|^2} \sum_{p_1 \in S}\sum_{p_2 \in S} d(p_1,p_2) \leq \avg(p,S) + \avg(p,S).
\]
By the definition of $\avg(p_1,S)$ and by the above inequality,
\[
 \frac{1}{|S|}\sum_{p_1 \in S} \avg(p_1,S) = \frac{1}{|S|^2} \sum_{p_1 \in S}\sum_{p_2 \in S} d(p_1,p_2) \leq \avg(p,S) + \avg(p,S).
\]
\end{proof}

\section{Analysis of the Natural Local Search} \label{:sec:analysis-of-local-search}
The goal of this section is to prove \cref{:thm:natural-local-search-terminates}. 
The main technique is a potential function argument.
More precisely, given a metric space $(X,d)$ and a parameter $k$, we define a potential function on the set of all possible $k$-clusterings. Next, we show that each step of the natural local search algorithm reduces the potential of the clustering maintained by the local search. To achieve this, we define a potential function on clusters $\Phi:2^{X} \to \bbR_{\geq 0}$, and say that the potential of a clustering $\cC$ is the sum of the potentials of its clusters $\Phi(\cC) \defeq \sum_{C \in \cC}\Phi(C)$.
Our main important observation is the following:
\begin{observation}\label{:obs:property-of-potential-function}
    Given $\alpha \geq 1$ and a metric space $(X,d)$, if a potential function $\Phi:(2^{X}\setminus\{\emptyset\}) \to \bbR_{\geq 0}$ satisfies that, for every non-empty set $S \subset X$ and every point $p \notin S$,
    \begin{align*}
     \avg(p,S) \leq \Phi(S \cup \{p\}) - \Phi(S) \leq \alpha \cdot \avg(p,S),        
    \end{align*}
    then each step of the natural local search algorithm with parameter $\alpha$ reduces the value of $\sum_{C \in \cC} \Phi(C)$, where $\cC$ is the clustering maintained by the algorithm.
\end{observation}

\begin{proof}
The property of the potential function implies that removing a point $p$ from a cluster $C \neq \{p\}$ decreases the potential of that cluster by at least $\avg(p,C \setminus \{p\})$.
This property also implies that adding a point $p$ to a cluster $C'$ increases its potential by at most $\alpha \cdot \avg(p,C')$.
Note that in each step of the local search algorithm a point $p$ is removed from a cluster $C\neq \{p\}$ and is added to a cluster $C' \neq C$ such that $\avg(p,C \setminus \{p\}) > \alpha \cdot \avg(p,C')$.
Hence, in each step, the decrease in the potential of the respective cluster $C$ is larger than the increase in the potential of the respective cluster $C'$, and there is no change in the potential of any other clusters.
This implies that each local update reduces the total potential of the clustering, as desired.
\end{proof}
The main challenge in the proof of~\cref{:thm:natural-local-search-terminates} is to show the existence of a potential function satisfying the required property of~\cref{:obs:property-of-potential-function} with $\alpha = 2\log n$.

\subsection{Potential Function}\label{:sec:potential-function}

In this subsection, we prove that there exists a potential function $\Phi:(2^{X}\setminus\{\emptyset\}) \to \bbR_{\geq 0}$ that satisfies the property from \cref{:obs:property-of-potential-function} for every $\alpha \geq 2\log n$. 
\begin{theorem} \label{:thm:existence-of-potential-function}
    For every metric space $(X,d)$ with $|X| = n \geq 3$, there exists a potential function $\Phi:(2^{X}\setminus\{\emptyset\}) \to \bbR_{\geq 0}$ on subsets of $X$ with the property that, for every non-empty set $S \subset X$ and every point $p \notin S$,
    \[
     \avg(p,S) \leq \Phi(S \cup \{p\}) - \Phi(S) \leq 2\log n \cdot \avg(p,S).
    \]
\end{theorem}
In order to prove the theorem, we will use the following potential function:
\begin{align*}
 \Phi(S) = \log |S| \cdot \sum_{p \in S} \avg(p,S) = \frac{\log |S|}{|S|} \sum_{p,q \in S} d(p,q).    
\end{align*}
We provide a brief outline of the proof for the existence of a potential function here and the full proof is deferred to~\cref{:sec:existence-of-potential-function}.
\begin{proof}[Proof sketch of~\cref{:thm:existence-of-potential-function}]
We provide upper and lower bounds for $\Phi(S \cup \{p\}) - \Phi(S)$, the amount by which the potential of the set $S$ increases when a point $p$ is added to it.
Let us define the function $f(x,y) = \frac{x\log y}{y}$ and use the notation $x_{S} = \sum_{p_1,p_2 \in S} d(p_1,p_2)$ and $y_{S} = |S|$. Note that $\Phi(S) \equiv f(x_S, y_S)$.
Thus, using this notation, the quantity that we need to analyze is the amount of change in the value of $f(x_S,y_S)$ as the point $p$ is added to $S$. Our goal is to show that this change is at least $\avg(p,S)$ and is at most $2\log n\cdot\avg(p,S)$.
At a high-level, we estimate this change by approximating the function $f$ linearly, utilizing its derivative at the point $(x_S,y_S)$. Note that, as $f$ is linear in its first argument, and the value of $y_S$ only changes by $1$, it provides a ``good'' estimate for the amount of change in $f(x_S,y_S)$.

Moreover, note that by adding a point $p$ to a set $S$, the value of $x_S$ increases by exactly 
\begin{align*}
 x_{S \cup \{p\}} - x_S = \sum_{p_1,p_2 \in S \cup \{p\}} d(p_1,p_2) - \sum_{p_1,p_2 \in S} d(p_1,p_2)
 = d(p,p) + 2\sum_{p' \in S} d(p,p') 
 = 2|S|\cdot \avg(p,S),   
\end{align*}
and the value of $y_S$ increases by exactly $1$.
So,
\begin{align*}
    \Phi(S \cup \{p\}) - \Phi(S)
    \simeq 2|S|\cdot\avg(p,S) \cdot \frac{\partial f}{\partial x}(x_S,y_S) + \frac{\partial f}{\partial y}(x_S,y_S)
\end{align*}
It is straightforward to verify that $\frac{\partial f}{\partial x}(x_S,y_S) = \frac{\log y_S}{y_S}$ and $\frac{\partial f}{\partial y}(x_S,y_S) = \frac{\log e - \log y_s}{y_S^2}x_S$. 
So,
\begin{align*}
    \Phi(S \cup \{p\}) - \Phi(S) 
    &\simeq 2|S|\cdot\avg(p,S) \cdot \frac{\log y_S}{y_S} + \frac{\log e - \log y_S}{y_S^2}x_S \\
    &= 2\cdot\avg(p,S) \cdot \log y_S + \frac{\log e - \log y_S}{y_S^2}x_S,
\end{align*}

Furthermore, using this notation,~\cref{:lem:average-distance-lemma} exactly translates to $\frac{x_S}{y_S^2} \leq 2 \avg(p,S)$. Since $\frac{x_S}{y_S^2} \geq 0$, when $y_S \geq e$,
\begin{align*}
    \Phi(S \cup \{p\}) - \Phi(S) &\ge 2 \cdot \avg(p,S) \cdot \log e, \;\mathrm{and} \\
    \Phi(S \cup \{p\}) - \Phi(S) &\le 2 \cdot \avg(p,S) \cdot \log y_S \le 2\log n \cdot \avg(p,S).
\end{align*}
\end{proof}

Now, we are ready to prove~\cref{:thm:natural-local-search-terminates}.
\begin{proof}[Proof of~\cref{:thm:natural-local-search-terminates}]
Given a metric space $(X,d)$, \cref{:thm:existence-of-potential-function} implies a potential function $\Phi$ that satisfies the condition of \cref{:obs:property-of-potential-function} for $\alpha = 2 \log n$.
Thus, \cref{:obs:property-of-potential-function} implies that $\sum_{C \in \cC} \Phi(C)$ decreases at each step of the local search (Algorithm \ref{:alg:natural-local-search}) when $\alpha \geq 2 \log n$.
Therefore, the state of the clustering $\cC$ can never repeat.
So, since the number of possible $k$-clusterings of $(X,d)$ is finite, the algorithm must terminate.
\end{proof}

\section{Adapting Local Search to Run In Polynomial Time}\label{:sec:basic-merge-and-split-algorithm}
In this section, we present a simplified version of the algorithm from \cref{:thm:fast-algorithm-exists}, with run-time $\Tilde{O}(n^2k)$.
This version presents most of the main ideas in the algorithm, and is much simpler.
In the appendix, we present the fast algorithm, thus proving \cref{:thm:fast-algorithm-exists}.
We first show why the natural local search is slow, and then present an overview of our approach to fix it, achieving the $\Tilde{O}(n^2 k)$-time algorithm.

While \cref{:obs:property-of-potential-function} showed that by setting $\alpha = 2\log n$, after each local search step, moving $p\in C$ to a new cluster $C'$, the total potential of the clustering decreases, it is not hard to see that the same analysis implies that if we instead set $\alpha = \gamma \cdot (2\log n)$, after each local search step, the potential decreases by at least $\left(1-\frac{1}{\gamma}\right) \cdot \avg(p, C\setminus\{p\})$.  

However, the natural local search may still be notably slow since the value of $\avg(p,C \setminus \{p\})$ might be significantly smaller than the total potential of the clustering, denoted by $\Phi(\cC) = \sum_{C \in \cC}\Phi(C)$.
Moreover, there exist instances where, for all $p\in X$ that envy any cluster, $\avg(p, C(p))$ is small. This scenario occurs when a large fraction of $\Phi(\cC)$ is due to $\avg(p', C(p'))$ for points $p'$ that do not envy any other cluster. 

In this section we remedy this issue of the natural local search, thus achieving an $\Tilde{O}(n^2 k)$-time algorithm.
To do this, we show that when the value of $\avg(p, C \setminus \{p\})$ in the current step is significantly smaller than $\Phi(\cC)$, then there is a different ``modififcation'' of the current clustering that significantly reduces the value of $\Phi(\cC)$.
A key observation is that any cluster $C^*$ can efficiently be split into two clusters $C^*_1$ and $C^*_2$ such that $\Phi(C^*_1) + \Phi(C^*_2)$ is smaller than $\Phi(C^*)$.
Indeed, by splitting a cluster $C^*$ into two equally-sized randomly chosen parts, 
\begin{align*}
 \expect{\Phi(C^*_1) + \Phi(C^*_2)}
 &= \frac{\log (|C^*|/2)}{|C^*|/2} \cdot \sum_{p,p' \in C^*}d(p,p')\cdot\Pr[p,p' \in C_1^* \lor p,p' \in C_2^*] \\  
 &\leq \frac{\log (|C^*|/2)}{|C^*|/2} \cdot \sum_{p,p' \in C^*}d(p,p')\cdot\frac{1}{2} \\  
 &= \frac{\log (|C^*|)-1}{|C^*|} \cdot \sum_{p,p' \in C^*}d(p,p') \\  
 &= (1-\frac{1}{\log |C^*|}) \cdot \Phi(C^*).
\end{align*}
Intuitively speaking, if we randomly split a cluster $C^*$ such that $\Phi(C^*) \geq \frac{\Phi(\cC)}{k}$, we can expect $\Phi(\cC)$ to decrease by a factor of approximately $\left(1-\frac{1}{k \log |C^*|}\right)$.
However, splitting a cluster is not a feasible operation on its own since it increases the number of clusters. So, to leverage the splitting operation, we also consider a ``merge'', particularly beneficial if merging two existing clusters increases the total potential by less than the amount decreased by the split.

Indeed, we demonstrate that a ``merge and split'' step exists when the standard local search step scarcely decreases the potential function.
When moving $p\in C$ to a new cluster $C'$ only decreases the potential by a little, the point $p$ is ``close'' to both clusters $C$ and $C'$. This implies that $C$ and $C'$ are also ``close to each other''.
Furthermore, we define a distance measure between two clusters providing an upper bound on the amount by which merging the clusters will increase the total potential $\Phi(\cC)$.

To summarize, the high-level idea is that if the current step of the natural local search only marginally reduces $\Phi(\cC)$, then both $\avg(p,C(p) \setminus \{p\})$ and $\avg(p,C')$ are ``small''. Consequently, we can merge the clusters $C(p)$ and $C'$, incurring only a ``small'' increase in $\Phi(\cC)$. In such cases, $C(p)$ and $C'$ can be merged and a cluster $C^*$, for which $\Phi(C^*) \ge \Phi(\cC)/k$, can instead be split. Then, we show that the increase in the total potential from the merge is offset by the potential's decrease due to the split operation.

The above ideas culminate in an algorithm that can quickly make progress on reducing the potential of its maintained clustering by repeatedly finding a point $p$ that is envious, and performing either the \quotes{swap} step of the natural local search, or a \quotes{merge and split} step.
The algorithm stops when no point is $\alpha$-envious for $\alpha = O(\log n)$, which implies that the current clustering is $\alpha$-IP stable. 
We show that the described algorithm reduces the potential of the clustering by a constant factor in time $\Tilde{O}(n^2 k)$.
One last issue is that the ratio between the potential of the initial clustering and the potential of the final clustering can be very large, resulting in a possibly $\Omega(\log(\frac{\Phi(\cC_{\mathrm{initial}})}{\Phi(\cC_{\mathrm{final}})}))$ iterations of ``reducing the clustering's potential by a constant factor''. 
To get around this issue, we show that initializing with an approximate $k$-center solution guarantees that $\Phi(\cC_{\mathrm{initial}})\le \poly(n) \cdot \Phi(\cC_{\mathrm{final}})$; hence, the total number of iterations in the ``modified'' local search is $O(\log n)$. 

The previous discussion shows that formally analyzing the algorithm requires the next three technical lemmas.
In Section~\ref{:sec:implementing-basic-merge-and-split}, we will present the algorithm and its analysis using the following three lemmas.

\begin{algorithm}[ht]
\caption{$\split$ is the partitioning procedure from \cref{:lem:slow-split}.}\label{:alg:split-procedure}
\KwData{$(X,d)$ where $|X|=n$, and a $(k-1)$-clustering $\cC$.}
\KwResult{a cluster $C^* \in \cC$ and a partition of $C^*$ satisfying the condition described in \cref{:lem:slow-split}.}
\textbf{compute} the potential of the clustering $\cC$; $\Phi(\cC) = \sum_{C \in \cC} \Phi(C)$ \label{:line:split:compute-potential-of-clustering}\\
$C^* \leftarrow \argmax_{\{C \mid C\in \cC, |C| > 1\}} \Phi(C)$ \label{:line:split:choose-cluster-to-split}\\
\Repeat{$\Phi(C^*_1) + \Phi(C^*_2) \leq \Phi(C^*) - \Omega(\frac{\Phi(C^*)}{\log n})$}{\label{:line:split:loop-start}
\textbf{sample} $C^*_1 \subset C^*$ of size $\ceil{|C^*|/2}{}$ chosen uniformly at random \label{:line:split:sample-partition}\\ 
$C^*_2 \leftarrow C^* \setminus C^*_1$\\
}\label{:line:split:repeat-until-condition}
\Return{$C^*$ and $(C^*_1,C^*_2)$}
\end{algorithm}

\begin{lemma}\label{:lem:slow-split}
    There exists a $\Tilde{O}(n^2)$-time Monte-Carlo randomized algorithm that, given a metric space $(X,d)$ with $n$ points and a $(k-1)$-clustering $\cC$, finds a cluster $C^* \in \cC$ and a partition $(C^*_1,C^*_2)$ of $C^*$ such that $\Phi(C^*_1) + \Phi(C^*_2) \leq \Phi(C^*) - \Omega(\frac{\Phi(\cC)}{k\log n})$.
    The error probability can be made as small as $O(\frac{1}{n^c})$, for any predefined constant $c$.
\end{lemma}
\begin{proof}
The $\split$ algorithm is described in Algorithm \ref{:alg:split-procedure}. $\split$ is a Las-Vegas algorithm, but it can be turned into a Monte-Carlo algorithm via standard reductions.

\paragraph{Correctness Analysis.}
Since the only clusters with non-zero potential have size at least $2$, and there are at most $k$ of these, we get that $\Phi(C^*) \geq \frac{1}{k}\sum_{C \in \cC \text{ s.t. } |C| \geq 2} \Phi(C) = \frac{1}{k}\Phi(\cC)$.
Thus, since the returned partition must satisfy the inequality from line \ref{:line:split:repeat-until-condition} of the algorithm, it also satisfies $\Phi(C^*_1) + \Phi(C^*_2) \leq \Phi(C^*) - \Omega(\frac{\Phi(\cC)}{k\log n})$, which proves the correctness of the algorithm.

\paragraph{Time Analysis.}
Firstly, the potential $\Phi(C)$ of each cluster $C$ can be calculated in time $O(|C|^2) \leq O(n|C|)$, so lines \ref{:line:split:compute-potential-of-clustering} and \ref{:line:split:choose-cluster-to-split} of the algorithm can be easily implemented in time $O(n^2)$.
Furthermore, each iteration of the loop can be performed in time $O(n^2)$, so we just need to show that at most $\Tilde{O}(1)$ iterations occur in expectation.
This is shown by the following \cref{:obs:probability-of-good-split}.
\end{proof}

\begin{observation}\label{:obs:probability-of-good-split}
In each execution of line \ref{:line:split:sample-partition} of Algorithm \ref{:alg:split-procedure}, the generated partition has a probability of $\Tilde{\Omega}(1)$ to satisfy the inequality $\Phi(C^*_1) + \Phi(C^*_2) \leq \Phi(C^*) - \frac{\Phi(C^*)}{4\log n}$.
\end{observation}

\begin{proof}
To prove the observation, it is enough to show that, for a given execution of line \ref{:line:split:sample-partition} of the algorithm,
\begin{equation}\label{:eq:split-procedure-expectation-argument}
    \expect{\Phi(C^*_1) + \Phi(C^*_2)} \leq (1 - \frac{1}{2\log n}) \cdot \Phi(C^*).
\end{equation}
To see why this is enough, notice that, by Markov's inequality, \cref{:eq:split-procedure-expectation-argument} implies that
\begin{align*}
    \Pr\left[ \Phi(C^*_1) + \Phi(C^*_2) \geq (1 - \frac{1}{4\log n})\Phi(C^*) \right] \leq \frac{\left(1 - \frac{1}{2\log n}\right)}{\left(1 - \frac{1}{4\log n}\right)} = \frac{4\log n-2}{4\log n - 1}.
\end{align*}
Hence,
\begin{align*}
     \Pr\left[ \Phi(C^*_1) + \Phi(C^*_2) \leq (1 - \frac{1}{4\log n})\Phi(C^*) \right] \geq 
     1- \frac{4\log n - 2}{4\log n -1}
     = \frac{1}{4\log n -1}
     = \Tilde{\Omega}(1).   
\end{align*}

So, our goal for the rest of the proof is to show that, for a given execution of line \ref{:line:split:sample-partition} of the algorithm, \cref{:eq:split-procedure-expectation-argument} holds:
For each point $p \in C^*$, let $C^*_{i(p)}$ denote the side of the partition to which $p$ belongs.
Now, since $\Phi(C^*_1) = \log |C^*_1| \sum_{p \in C^*_1} \avg(p,C^*_1) = \log (\ceil{|C^*|/2}{})\sum_{p \in C^*_1} \avg(p,C^*_1)$ and similarly $\Phi(C^*_2) \leq \log (\ceil{|C^*|/2}{})\sum_{p \in C^*_2} \avg(p,C^*_2)$, we get by linearity of expectation that
\begin{equation} \label{:eq:split-procedure-potentials-of-resulting-clusters}
    \expect{\Phi(C^*_1) + \Phi(C^*_2)} \leq \log (\ceil{|C^*|/2}{}) \sum_{p \in C^*} \expect{\avg(p,C^*_{i(p)})}.
\end{equation}
Now, for a given point $p$, consider the process that samples a partition $(C^*_1,C^*_2)$ as in line \ref{:line:split:sample-partition} of the algorithm, and then samples a point $p_{\mathrm{sample}}$ from $C^*_{i(p)}$ uniformly at random.
Notice that $d(p,p_{\mathrm{sample}})$ can be non-zero only when $p_{\mathrm{sample}}$ is a point from $C^* \setminus \{p\}$.
Furthermore, by a symmetry argument, each point in $C^* \setminus \{p\}$ has the same probability of being chosen as $p_{sample}$ by the aforementioned process.
Thus,
\[
 \expect[(C^*_1,C^*_2)]{\expect[p_{\mathrm{sample}} \sim C^*_{i(p)}]{p_{\mathrm{sample}}}}
 = \frac{\Pr[p_{\mathrm{sample}} \in C^* \setminus \{p\}]}{|C^* \setminus \{p\}|} \sum_{p' \in C^* \setminus \{p\}}d(p,p').
\]
So, since
\[
 \Pr[p_{\mathrm{sample}} \in C^* \setminus \{p\}]
 = \expect[(C^*_1,C^*_2)]{(|C^*_{i(p)}|-1)/|C^*_{i(p)}|} \leq (|C^*|-1)/|C^*|,
\]
we get that
\[
 \expect[(C^*_1,C^*_2)]{\expect[p_{\mathrm{sample}} \sim C^*_{i(p)}]{p_{\mathrm{sample}}}}
 \leq \frac{1}{|C^*|} \sum_{p' \in C^* \setminus \{p\}}d(p,p') = \avg(p,C^*).
\]
However, for each such point $p$ and each possible partition $(C^*_1,C^*_2)$, the value of $\expect[p_{\mathrm{sample}} \sim C^*_{i(p)}]{d(p,p_{\mathrm{sample}})}$ is exactly $\avg(p,C^*_{i(p)})$, so plugging this into the previous inequality gives
\[
 \forall p \in C^*, \quad \expect[(C^*_1,C^*_2)]{\avg(p,C^*_{i(p)})} \leq \avg(p,C^*).
\]
Plugging this inequality into \cref{:eq:split-procedure-potentials-of-resulting-clusters}, we get that
\[
 \expect{\Phi(C^*_1) + \Phi(C^*_2)}
 \leq \log (\ceil{|C^*|/2}{}) \sum_{p \in C^*} \avg(p,C^*)
\]
Furthermore, since $|C^*|\geq 2$, we have that $\ceil{|C^*|/2}{} \leq 2|C^*|/3$ and thus $\log (\ceil{|C^*|/2}{}) \leq \log |C^*| - 1/2$, so we get that
\[
 \expect{\Phi(C^*_1) + \Phi(C^*_2)}
 \leq (\log |C^*|- 1/2) \cdot \sum_{p \in C^*} \avg(p,C^*)
 = \frac{\log |C^*|- 1/2}{\log |C^*|} \cdot \Phi(C^*)
\]
which proves \cref{:eq:split-procedure-expectation-argument}, as we needed.
\end{proof}

\begin{lemma} \label{:lem:merge-cost-bound}
    Given a metric space $(X,d)$ and two disjoint clusters $C,C' \subseteq X$, the increase in the potential $\Phi$ after merging $C$ and $C'$ is at most,
    \begin{align}\label{eq:merge-gain}
        \frac{2\log n}{\max\{|C|,|C'|\}} \sum_{p_1 \in C} \sum_{p_2 \in C'} d(p_1,p_2).
    \end{align}
    Furthermore, for every point $p \in X$, Eq.~\eqref{eq:merge-gain} is at most $2\log n \cdot \min\{|C|,|C'|\} \cdot (\avg(p,C) + \avg(p,C'))$.
\end{lemma}
\begin{proof}
Firstly, for every point $p \in X$, the inequality
\begin{align*}
    \frac{2\log n}{\max\{|C|,|C'|\}} \sum_{p_1 \in C} \sum_{p_2 \in C'} d(p_1,p_2)
    \leq 2\log n \cdot \min\{|C|,|C'|\} \cdot (\avg(p,C) + \avg(p,C'))
\end{align*}
follows easily from \cref{:obs:triangle-inequality-for-average-distance}, so all we need to prove is that the increase in potential due to merging two clusters, $C$ and $C'$, is bounded by
\begin{align*}
    \Phi(C \cup C') - \Phi(C) - \Phi(C')
    \leq \frac{2\log n}{\max\{|C|,|C'|\}} \sum_{p_1 \in C} \sum_{p_2 \in C'} d(p_1,p_2).
\end{align*}
We will divide this proof to two cases, based on whether one of the clusters is of size $1$.

\paragraph{Case $\min\{|C|,|C'|\}=1$.} 
For this case, we assume without loss of generality that $C'$ is the cluster of size $1$, and let $p$ denote the single point in $C'$.
Then, $\Phi(C')=0$, and increase in potential can be rewritten as
\[
\Phi(C \cup C') - \Phi(C) - \Phi(C') = \Phi(C \cup \{p\}) - \Phi(C)
\]
Furthermore, since $\Phi$ is the potential function form \cref{:thm:existence-of-potential-function}, that theorem promises that $\Phi(C \cup \{p\}) - \Phi(C) \leq 2\log_2(n)\cdot\avg(p,C)$.
So,
\begin{align*}
 \Phi(C \cup C') - \Phi(C) - \Phi(C')
 \leq 2\log n \cdot\avg(p,C)
 &= \frac{2\log n}{|C|}\sum_{p_1 \in C}d(p_1,p) \\
 &= \frac{2\log n}{\max\{|C|,|C'|\}} \sum_{p_1 \in C} \sum_{p_2 \in C'} d(p_1,p_2),    
\end{align*}
as needed.

\paragraph{Case $\min\{|C|,|C'|\}\geq2$.}
For this case, we let $f:\bbN_{\geq 2} \to \bbR_{> 0}$ denote the function $f(\ell) = \frac{\log \ell}{\ell}$.
Then, since $|C \cup C'|\geq 4$, since $f$ is decreasing in the interval $[4,\infty)$, and since $f(2),f(3) \geq f(4)$, it must be that $f(|C|),f(|C'|)\geq f(|C \cup C'|)$.
Furthermore, for this function $f$, the potential of any set $S$ with $|S|\geq 2$ is $\Phi(S)=f(|S|)\sum_{p_1 \in S}\sum_{p_2 \in S} d(p_1,p_2)$.
Together, these facts imply that
\[
 \Phi(C \cup C')
 = f(|C \cup C'|) \sum_{p_1 \in C \cup C'}\sum_{p_2 \in C \cup C'} d(p_1,p_2)
 \leq \Phi(C) + \Phi(C') + 2f(|C \cup C'|) \sum_{p_1 \in C}\sum_{p_2 \in C'} d(p_1,p_2);
\]
hence,
\[
 \Phi(C \cup C') - \Phi(C) - \Phi(C')
 \leq 2\frac{\log |C \cup C'|}{|C \cup C'|} \sum_{p_1 \in C}\sum_{p_2 \in C'} d(p_1,p_2)
 \leq \frac{2\log n}{\max\{|C|,|C'|\}} \sum_{p_1 \in C} \sum_{p_2 \in C'} d(p_1,p_2).
\]
\end{proof}

\begin{lemma} \label{:lem:k-center-initialization}
    Consider a metric space $(X,d)$ with $n$ points and a $k$-clustering $\cC$. If $\cC$ is a $O(1)$-approximate solution for $k$-center on $X$, then it $\poly(n)$-approximately minimizes the potential function $\Phi$; more precisely, $\Phi(\cC) \le \poly(n) \cdot \Phi(\cC^*)$, where $\cC^*$ is a $k$-clustering minimizing $\Phi$.
\end{lemma}
\begin{proof}
To prove this lemma, we just need to show that, for every $k$-clustering $\cC$ of a metric space $(X,d)$, the $k$-centers value of $\cC$ is a $\poly(n)$ approximation for the potential of $\cC$.
Let $\diameter(\cC) \defeq \max_{C \in \cC} \diameter(C) = \max_{C \in \cC} \max_{p,p' \in C}d(p,p')$.
It is a known fact that $Diam(\cC)$ is a $2$-approximation for the $k$-centers value of $\cC$, so it is enough if we prove that $\diameter(\cC)$ is a $\poly(n)$ approximation for the potential of $\cC$.
Indeed, since $\Phi(\cC) = \sum_{C \in \cC} \frac{\log |C|}{|C|}\sum_{p,p' \in C}d(p,p')$, it holds that $\Phi(\cC) \geq \frac{1}{n} \diameter(\cC)$ and that $\Phi(\cC) \leq kn^2\log n  \cdot \diameter(\cC)$.\footnote{While a tighter bound on $\Phi(\cC)$ relative to the diameter is attainable, the current bound is sufficient for our purposes.}
\end{proof}

Note that there exists a $2$-approximation algorithm for $k$-center, e.g.,~\citep{gonzalez1985clustering}, that runs in time $O(nk)$. So, the initialization runs in $O(nk)$ time. 
\subsection{Algorithm Description and its Analysis} \label{:sec:implementing-basic-merge-and-split}
We show that Algorithm \ref{:alg:basic-merge-and-split} runs in time $\Tilde{O}(n^2 k)$. This is a slower but simpler variant of the algorithm from \cref{:thm:fast-algorithm-exists}.

In our analysis, we will proceed under the assumption that the split procedure of \cref{:lem:slow-split} does not encounter any errors.
This assumption is justified since the error probability of the split procedure can be made to be an arbitrarily small polynomial at the cost of only a $\polylog(n)$ blowup in the running time.

\begin{algorithm}[ht]
\caption{Adapting the local search algorithm for a runtime of $\Tilde{O}(n^2 k)$.}\label{:alg:basic-merge-and-split}
\KwData{$(X,d)$, $n = |X|$, $k \leq n$, $\alpha=4\log n$}
\KwResult{an $\alpha$-IP stable $k$-clustering of $(X,d)$}
{\bf set} $\cC$ be the clustering returned by the greedy algorithm of $k$-center~\citep{gonzalez1985clustering} on $X$\\
\For{every $p \in X$ and $C \in \cC$\label{:line:initializing-averages-start}}{
    {\bf compute} $\avg(p,C)$\label{:line:initializing-averages-end}
}
\While{exist $p\in X$ and $C' \in \cC$ s.t. $C(p)\neq\{p\}$ and $\avg(p,C(p)\setminus\{p\}) > \alpha \cdot \avg(p,C')$}{
    \uIf{$\avg(p,C(p)\setminus\{p\}) \geq \Omega(\frac{\Phi(\cC)}{k\log(n)})/(5n\log n)$ \label{:line:if-statement}}{
        {\bf move} $p$ from $C=C(p)$ to $C'$\label{:line:swap-step-start}\\
        \For{every $p'\in X$ \label{:line:swap-step-update-loop}}{
            {\bf update} $\avg(p', C)$ and $\avg(p',C')$\label{:line:swap-step-end}
        }
    }
    \Else{
        $C'' \leftarrow C \cup C'\quad\rhd$ merge clusters $C,C'$\label{:line:merge-and-split-step-start}\\
        $\cC \leftarrow (\cC\setminus \{C, C'\}) \cup \{C''\}$ \\
        {\sc Split}($\cC$) $\quad\rhd$ a cluster $C^* \in \cC$ is split into $C^*_1$ and $C^*_2$ by {\sc Split} (as in \cref{:lem:slow-split}) \label{:line:split}\\
        $\cC \leftarrow (\cC\setminus \{C^*\}) \cup \{C^*_1, C^*_2\}$ \\
        \For{every $p'\in X$ and $C\in \{C'', C^*_1, C^*_2\}$}{
            {\bf compute} $\avg(p',C)$ \label{:line:merge-and-split-step-end}
        }
    }
}
\Return $\cC$
\end{algorithm}
In the analysis of this algorithm, we refer to each iteration of the while loop as a \emph{step} of the algorithm. We refer to those iterations where the condition of the if statement in line \ref{:line:if-statement} was met as \emph{swap} steps, and to those iterations where this condition was not met as \emph{merge and split} steps.
It is easy to see that algorithm only terminates once it reaches an $\alpha$-approximate IP stable clustering, so the entire analysis is focused on bounding the runtime by $\Tilde{O}(n^2k)$.
We do this by proving the following two claims.
\begin{claim}\label{:cl:potential-reduction-from-steps}
    Each swap step reduces the clustering's potential by a factor of $\left(1-\Tilde{\Omega}(\frac{1}{nk})\right)$, and each merge and split step reduces the potential by a factor of $\left(1-\Tilde{\Omega}(\frac{1}{k})\right)$.
\end{claim}
\begin{claim}\label{:cl:time-required-for-steps}
    Each swap step runs in time $\Tilde{O}(n)$, and each merge and split step can be implemented in time $\Tilde{O}(n^2)$. This includes the operations required to check the conditions of the while loop and the if statement.
\end{claim}
To see why these two claims imply $\Tilde{O}(n^2 k)$ runtime of the algorithm, notice that, by \cref{:lem:k-center-initialization}, the total reduction in potential of the clustering across all steps is at most a $\poly(n)$ multiplicative factor.
Thus, \cref{:cl:potential-reduction-from-steps} implies that there are at most $\Tilde{O}(nk)$ swap steps and at most $\Tilde{O}(k)$ merge and split steps.
Together with \cref{:cl:time-required-for-steps}, this implies that the algorithm spends at most $\Tilde{O}(n^2 k)$ time across all steps.
Furthermore, the time that the algorithm spends outside of the while loop is at most $O(n^2 k)$, since the greedy algorithm of $k$-center runs in time $O(nk)$ and calculating each $\avg(p,C)$ requires only $O(n)$ time.
\begin{proof}[Proof of \cref{:cl:potential-reduction-from-steps}]
Let $\cC$ be the clustering at the beginning of a given step.
We will now show that the total change in potential during this step is a decrease, by $\Phi(\cC) \cdot \Tilde{\Omega}(\frac{1}{nk})$ if the step is a swap step, and by $\Phi(\cC) \cdot \Tilde{\Omega}(\frac{1}{k})$ if the step is a merge and split step.

For a swap step, the same analysis from the proof of \cref{:obs:property-of-potential-function} shows that removing the point $p$ from the cluster $C$ decreases the potential of $C$ by at least $\avg(p,C\setminus\{p\})$, and that adding the point $p$ to the cluster $C'$ increases the potential of that cluster by at most $\alpha \cdot \avg(p,C')$, where $\alpha = 2\log n$ as in \cref{:thm:existence-of-potential-function}. 
However, unlike in \cref{:sec:analysis-of-local-search}, Algorithm \ref{:alg:basic-merge-and-split} guarantees that $\avg(p,C\setminus\{p\}) > (4\log n) \cdot \avg(p,C')$, so the increase in the potential of the cluster $C'$ is smaller than $\frac{1}{2}\avg(p,C\setminus\{p\})$.
Thus, the total change in the potential of the clustering in this swap step is a decrease by more than $\frac{1}{2}\avg(p,C\setminus\{p\})$.
Furthermore, since this is a swap step, the condition of the if statement in line \ref{:line:if-statement} of the algorithm must have been met, which means that 
\begin{align*}
 \frac{1}{2}\avg(p,C\setminus\{p\}) 
 \geq \Omega(\frac{\Phi(\cC)}{k\log(n)})/(10n\log n) = \Phi(\cC) \cdot \Tilde{\Omega}(\frac{1}{nk}).   
\end{align*}

Next, we analyze the merge and split step. First, note that splitting $C^*$ decreases the potential by $\Omega(\frac{\Phi(\cC)}{k\log(n)})$.
Second, \cref{:lem:merge-cost-bound} guarantees that merging $C$ and $C'$ increases the potential by at most 
\begin{align*}
    (2\log n) \cdot \min\{|C|,|C'|\} \cdot (\avg(p,C) + \avg(p,C'))
    \leq (2\log n) \cdot n \cdot (\avg(p,C) + \avg(p,C')).
\end{align*}
Since $\avg(p,C) \leq \avg(p,C \setminus\{p\})$ and $\avg(p,C') < \avg(p,C \setminus\{p\})/(4\log n)$, this increase in the potential is at most $(\log n) \cdot n \cdot \frac{5}{2}\avg(p,C \setminus\{p\})$.
Furthermore, since this is a merge and split step, the condition of the if statement in line \ref{:line:if-statement} does not hold.
Thus, if we set the value of $\Omega(\frac{\Phi(\cC)}{k\log(n)})$ in the if statement to be the same as the one guaranteed by \cref{:lem:slow-split}, we get that the increase in the potential due to the merge is at most a half of the decrease in the potential due to the split.
So, the overall change in the potential due to the merge and split step is a decrease by at least $\frac{1}{2} \cdot \Omega(\frac{\Phi(\cC)}{k\log(n)}) = \Phi(\cC) \cdot \Tilde{\Omega}(\frac{1}{k})$.
\end{proof}

\begin{proof}[Proof of \cref{:cl:time-required-for-steps}]
We will start by explaining how to check the conditions of the while statement and the if statement in time $\Tilde{O}(n)$.
Then we analyze the running time of implementing lines \ref{:line:swap-step-start}-\ref{:line:swap-step-end} and lines \ref{:line:merge-and-split-step-start}-\ref{:line:merge-and-split-step-end}.
For the if statement, the only non-trivial operation involves calculating $\Phi(\cC)$. 
However, this can be efficiently computed using the equality $\Phi(\cC) = \sum_{p' \in X} \log (|C(p')|)\cdot\avg(p',C(p'))$ and considering that the algorithm already maintains the values of $\avg(p',C(p'))$.
The most challenging part is implementing the condition of the while loop efficiently.
To do so, the algorithm maintains, for each $p \in X$, a min-heap $\mathrm{Heap}_{p}$ that contains the values of $\avg(p,C')$ for all clusters $C' \neq C(p)$.
Then, to check the condition of the while loop, the algorithm uses $\mathrm{Heap}_{p}$ to find $\min_{C' \neq C(p)} \avg(p,C')$ in time $\Tilde{O}(1)$ for each point $p$.
Furthermore, $\avg(p,C(p) \setminus \{p\})$ can be easily computed using the maintained value of $\avg(p,C(p))$ and the equality 
$\avg(p,C(p) \setminus \{p\}) = \frac{|C(p)|}{|C(p) \setminus \{p\})|}\avg(p,C(p))$, 
given that the algorithm is also maintaining the size of the clusters. 
Moreover, once all $\avg(p,C)$ values in lines \ref{:line:initializing-averages-start}-\ref{:line:initializing-averages-end} are computed, initializing these heaps, $\{\mathrm{Heap}_p \;|\;  p\in P\}$, require only $\Tilde{O}(nk)$ time. 

For a swap step, line~\ref{:line:swap-step-start} runs simply in $O(1)$. Moreover, each iteration of the for loop in line~\ref{:line:swap-step-update-loop} can be implemented in time $\Tilde{O}(1)$, using the equalities $\avg(p',C \setminus \{p\}) = \frac{1}{|C \setminus \{p\}|}\left(|C|\avg(p,C)- d(p,p')\right)$ and $\avg(p',C' \cup \{p\}) = \frac{1}{|C' \cup \{p\}|}\left(|C'|\avg(p,C') + d(p,p')\right)$.
Note that the we need to update the min-heaps, which can be done in $\Tilde{O}(1)$ since only two elements of each heap have changed.

For a merge and split step, \cref{:lem:slow-split} guarantees that line \ref{:line:split} runs in time $\Tilde{O}(n^2)$. Furthermore, it is straightforward to verify that the rest of operations can also be implemented in that time.
\end{proof}
This concludes the analysis of Algorithm \ref{:alg:basic-merge-and-split}.

\section{Conclusion}
In our study of the IP-stable clustering problem, we examined the natural local search algorithm. Notably, recent works~\citep{ahmadi2022individual,aamand2023constant} have left open the efficacy of existing natural clustering algorithms concerning the IP-stability metric. We established that by employing a carefully selected update rule, the local search terminates, yielding an $O(\log n)$-IP stable clustering. Moreover, with further refinements, we achieved an algorithm runtime of $\Tilde{O}(nk)$, surpassing the runtime of the existing $O(1)$-IP stable clustering by~\cite{aamand2023constant}.

Additionally, we studied IP-stable clustering using alternative functions, including $\max$ and $\mathrm{median}$, and provided provable guarantees.

Finally, analyzing a global clustering objective for our proposed local search algorithms remains an interesting open question.

\bibliographystyle{abbrvnat}
\bibliography{ip-cluster}
\appendix

\section{Proof of Theorem \texorpdfstring{\ref{:thm:existence-of-potential-function}}{Theorem 6}}\label{:sec:existence-of-potential-function}
\begin{proof}[Proof of~\cref{:thm:existence-of-potential-function}]
As described in \cref{:sec:potential-function}, we define the potential function
$\Phi(S) \defeq \log |S| \cdot \sum_{p \in S} \avg(p,S)$.
Then, our goal is to show that, for every non-empty set $S \subset X$ and every point $p \in X \setminus S$,
\begin{equation} \label{:eq:potential-function-proof:property-of-potential-function}
    \avg(p,S) \leq \Phi(S \cup \{p\}) - \Phi(S) \leq 2\log_2(n) \cdot \avg(p,S).
\end{equation}
We will take care of the special cases $|S|=1$ and $|S|=2$ separately, and then do the more general case $|S| \geq 3$. We note that the constant factor in the value $2\log n$ could be slightly improved by changing the potential function so that all sets of size at least $3$ have a slightly smaller potential, but we choose not to do this for the sake of simplicity.
\begin{itemize}[leftmargin=*]
    \item \textbf{Case $|S| = 1$.}
    In this case, $S = \{p'\}$ for some point $p' \neq p$.
    So, we have that $\avg(p,S) = d(p,p')$, that $\Phi(S)=0$, and that $\Phi(S \cup \{p\}) = \log 2 \cdot (\avg(p', S \cup \{p\}) + \avg(p,S \cup \{p\})) = (\frac{d(p,p')}{2} + \frac{d(p,p')}{2}) = d(p,p')$.
    Thus $\Phi(S \cup \{p\}) - \Phi(S) = d(p,p') = \avg(p,S)$.
    \item \textbf{Case $|S| = 2$.}
    In this case, $S = \{p_1,p_2\}$ for some $p_1,p_2 \in X \setminus \{p\}$.
    So, $\avg(p,S) = \frac{d(p,p_1)+d(p,p_2)}{2}$, and $\Phi(S) = \avg(p_1,S) + \avg(p_2,S) = \frac{d(p_1,p_2)}{2} + \frac{d(p_1,p_2)}{2} = d(p_1,p_2)$.
    Furthermore,
    \[
     \Phi(S \cup \{p\}) = \log 3 \cdot \sum_{p' \in \{p,p_1,p_2\}} \avg(p', \{p,p_1,p_2\}) = \log 3 \cdot \frac{2}{3} \cdot (d(p,p_1) + d(p,p_2) + d(p_1,p_2)).
    \]
    Thus, since $\frac{2}{3}\log 3 \geq 1$, we have
    \[
     \Phi(S \cup \{p\}) - \Phi(S) \geq d(p,p_1) + d(p,p_2) = 2 \avg(p,S) \geq \avg(p,S)
    \]
    and since $\frac{2}{3}\log 3 \leq \frac{11}{10}$, we have $\Phi(S \cup \{p\}) - \Phi(S) \leq \frac{11}{10}(d(p,p_1) + d(p,p_2)) + \frac{1}{10}d(p_1,p_2)$, which, by the triangle inequality means that
    \[
     \Phi(S \cup \{p\}) - \Phi(S) \leq \frac{6}{5}(d(p,p_1) + d(p,p_2)) < \log n\cdot (d(p,p_1) + d(p,p_2)) = 2\log n \cdot \avg(p,S)
    \]
    \item \textbf{Case $|S| \geq 3$.}
    Let
    \begin{align*}
        y_S = |S| \quad \mathrm{and} \quad x_S = \sum_{p_1,p_2 \in S} d(p_1,p_2) \quad \mathrm{and} \quad x_{S \cup \{p\}} = \sum_{p_1,p_2 \in S \cup \{p\}} d(p_1,p_2).
    \end{align*}
    Then 
    \begin{align*}
        \Phi(S) = \frac{\log y_S}{y_S}x_S \quad \mathrm{and} \quad \Phi(S \cup \{p\}) = \frac{\log (y_S + 1)}{y_S + 1}x_{S \cup \{p\}},
    \end{align*}
    and
    \begin{align*}
        x_{S \cup \{p\}} - x_S = \sum_{p_1,p_2 \in S \cup \{p\}} d(p_1,p_2) - \sum_{p_1,p_2 \in S} d(p_1,p_2) = d(p,p) + 2\sum_{p' \in S} d(p,p') = 2y_S\cdot \avg(p,S).
    \end{align*}
    Therefore,
    \begin{equation} \label{:eq:potential-function-proof:exact-difference-in-potential}
    \begin{aligned}
        \Phi(S \cup \{p\}) - \Phi(S)
        &= \frac{\log (y_S + 1)}{y_S + 1} \cdot x_{S \cup \{p\}} - \frac{\log y_S}{y_S} \cdot x_S \\
        &= \left(\frac{\log (y_S + 1)}{y_S + 1} - \frac{\log y_S}{y_S}\right)x_S + \frac{\log (y_S + 1)}{y_S + 1} \cdot 2y_S\cdot \avg(p,S).
    \end{aligned}
    \end{equation}
    Furthermore, since $y_S = |S| \geq 3$ and the function $\frac{\log y}{y}$ is decreasing for $y\ge 3$,
    \begin{align} \label{:eq:potential-function-proof:negativity-of-first-term}
        \left(\frac{\log (y_S + 1)}{y_S + 1} - \frac{\log y_S}{y_S}\right) < 0.
    \end{align}
    \cref{:lem:average-distance-lemma} exactly says that $x_S \leq 2y_S^2\avg(p,S)$. \cref{:eq:potential-function-proof:negativity-of-first-term} lets us plug $x_S \geq 0$ and $x_S \leq 2y_S^2\avg(p,S)$ into \cref{:eq:potential-function-proof:exact-difference-in-potential}, respectively giving
    \begin{equation} \label{:eq:potential-function-proof:proof-of-righthand-inequality}
    \begin{aligned}
        \Phi(S \cup \{p\}) - \Phi(S)
        \leq \frac{\log (y_S + 1)}{y_S + 1} \cdot 2y_S\cdot \avg(p,S) 
        &\leq \frac{\log n}{y_S + 1} \cdot 2y_S\cdot \avg(p,S) \\ 
        &\leq 2 \log n\cdot \avg(p,S)
    \end{aligned}
    \end{equation}
    and
    \begin{equation} \label{:eq:potential-function-proof:proof-of-lefthand-inequality}
    \begin{aligned}
        \Phi(S \cup \{p\}) - \Phi(S) 
        &\geq \left(\frac{\log (y_S + 1)}{y_S + 1} - \frac{\log y_S}{y_S}\right) \cdot 2y_S^2\avg(p,S) + \frac{\log (y_S + 1)}{y_S + 1} \cdot 2y_S\cdot \avg(p,S) \\
        &= \left(\log (y_S + 1) - \log y_S\right) \cdot 2y_S\avg(p,S) \\
        &\geq \frac{\log e}{y_S + 1} \cdot 2y_S\avg(p,S) 
        \geq 2\avg(p,S),
    \end{aligned}
    \end{equation}
    where the second to last inequality follows from the fact that the derivative of the function $\log x$ in greater than $\frac{\log e}{y_S+1}$ in the whole interval $(y_S,y_S + 1)$, and the last inequality follows from the fact that $\frac{y_S}{y_S + 1}\log e \geq \frac{3}{4}\log e \geq 1$.
    \cref{:eq:potential-function-proof:proof-of-lefthand-inequality} and \cref{:eq:potential-function-proof:proof-of-righthand-inequality} prove the left-hand and right-hand inequalities in \cref{:eq:potential-function-proof:property-of-potential-function}, as we needed.
\end{itemize}
\end{proof}

\section{Faster Implementation of Local Search}\label{:sec:faster-implementation}
The goal of this section is to prove \cref{:thm:fast-algorithm-exists}, by designing an adaptation of local search that runs in time $\Tilde{O}(nk)$.
For the sake of simplicity, we will separate the steps of the local search into epochs, where each epoch reduces the potential of the maintained clustering by a constant factor.
In \cref{:sec:fast-alg-from-fast-epoch}, we formally prove that if such an epoch can be implemented fast, then the local search can be implemented fast.
Thus, our main focus is to prove that an epoch can be implement fast, as stated in the following theorem.
We prove this theorem in \cref{:sec:proof-of-fast-epoch}

\begin{theorem} \label{:thm:fast-epoch}
    There exists a Monte-Carlo randomized algorithm {\sc Epoch} (Algorithm \ref{:alg:fast-epoch}) that runs in time $\Tilde{O}(nk)$, and given a metric space $(X,d)$ with $|X|=n$ and a $k$-clustering $\cC$ of $(X,d)$, computes a $k$-clustering $\cC'$ of $(X,d)$ that is either $O(\log n)$-IP stable or has $\sum_{C \in \cC'} \Phi(C) < \frac{3}{4} \cdot \sum_{C \in \cC} \Phi(C)$.
    The error probability can be made as low as $O(\frac{1}{n^c})$, for any predefined constant $c\geq1$.
\end{theorem}

We will now present a sketch of what additional modifications need to be made to the local search strategy from Algorithm \ref{:alg:basic-merge-and-split} in order to make it run in time $\Tilde{O}(nk)$. 
To do this, we focus on a single epoch of the local search, and let $\Phi_{0}$ be the potential of the clustering at the beginning of the epoch.

As long as $\Phi(\cC) \geq \Omega(\Phi_{0})$, a cluster split can reduce the potential of the clustering by $\Tilde{\Omega}(\Phi_0/k)$, so the strategy from Algorithm \ref{:alg:basic-merge-and-split} either performs a ``merge and split'' step that reduces $\Phi(\cC)$ by $\Tilde{\Omega}(\Phi_0/k)$ in time $\Tilde{O}(n^2)$, or performs a ``swap'' step that reduces the potential by $\Tilde{\Omega}(\Phi_0/(nk))$ in time $\Tilde{O}(n)$.
In both these operations, the ratio between the time spent and the decrease in potential is $\Tilde{O}(n^2 k / \Phi_0)$, resulting in the running time of $\Tilde{O}(n^2 k)$ for Algorithm \ref{:alg:basic-merge-and-split}.
Furthermore, once $\Phi(\cC) < \Omega(\Phi_{0})$, we can simply end the current epoch.

In order to reduce the running time by a factor of $n$, we need to reduce the ratio between the time spent and the decrease in potential to be $\Tilde{O}(nk / \Phi_0)$, which means that we need to implement each merge and split step in time $\Tilde{O}(n)$, and implement each swap step in time $\Tilde{O}(1)$.

\paragraph{Merge and Split Step.} In evaluating the runtime of a merge and split step, the primary bottlenecks are the split procedure itself and the time taken to compute the values of $\avg(p,C)$ for each new cluster and for every $p \in X$. Moreover, within the split procedure, the bottleneck is calculating the potential of each potential random split to determine if it is a good split. This computation is essentially about finding values of $\avg(p,C)$, given that $\Phi(\cC) = \log(|C|)\cdot\sum_{p \in C} \avg(p,C)$.
Thus, to implement a merge and split step quickly, we just need to know how to quickly compute $\avg(p,C)$ for a specified cluster $C$ and all points $p \in C$.
It turns out that if we want to compute these values approximately, we can do it via {\em importance sampling}.

\paragraph{Swap Step.} In the swap operation, the bottlenecks are finding the point $p$ that needs to be swapped, and updating $\avg(p',C)$ and $\avg(p',C')$ for all points $p' \in X$.
In order to find a candidate point $p$ for the swap operation, the implementation from \cref{:sec:implementing-basic-merge-and-split} was already supporting a min-heap for each point $p$, which gives the algorithm access to $\min_{C' \in \cC\setminus\{C(p)\}}\avg(p,C')$.
To speed up this process further and get $\Tilde{O}(1)$ runtime, we just need to maintain an additional min-heap that contains the values $\left\{\frac{\min_{C' \in \cC\setminus\{C(p)\}} \avg(p,C')}{\avg(p,C(p)\setminus\{p\})}\right\}_{p \in X}$.
The more challenging part is updating the values $\avg(p',C)$ and $\avg(p',C')$ for all points $p' \in X$.
Since updating these $\Theta(n)$ different values requires $\Omega(n)$ time regardless of how fast we can compute them, we cannot update them after every swap step.
Instead, we show that it is possible to keep track of an upper-bound on the additive difference between the old, outdated, value $\widehat{\avg}(p,C)$ and the true unknown value $\avg(p,C)$, and only recompute these values once $|\widehat{\avg}(p,C) - \avg(p, C)|$ becomes too large.

We further demonstrate that for each $\avg(p',C)$ and $\avg(p',C')$, the increase in this upper-bound due to a swap step is proportionate to the corresponding decrease in $\Phi(\cC)$. As a result, the cumulative increase in all these upper bounds throughout the entire epoch aligns proportionally with $\Phi_0$.

Finally, we need to ensure that the additive error $|\widehat{\avg}(p,C) - \avg(p,C)|$ remains roughly on par with $\Phi_0$ before recalculating $\widehat{\avg}(p,C)$.
To do this, we show that distinguishing between the case $p$ is $\alpha$-envious of $C'$ and the case $p$ is less than $(\alpha/2)$-envious of $C'$ only requires estimates of $\avg(p,C(p))$ and $\avg(p,C')$ that are correct up to an additive term of $\approx \frac{1}{\alpha}(\avg(p,C(p)) + \avg(p,C'))$. So, by \cref{:obs:triangle-inequality-for-average-distance}, we can tolerate an additive error that is roughly equal to the ``distance'' $\frac{1}{|C|\cdot|C'|}\sum_{p_1 \in C}\sum_{p_2 \in C'} d(p_1,p_2)$.
Furthermore, we can ensure that the distance $\frac{1}{|C|\cdot|C'|}\sum_{p_1 \in C}\sum_{p_2 \in C'} d(p_1,p_2)$ is roughly comparable to $\Phi_0$. If not, we can perform a merge and split step that combines the clusters $C$ and $C'$.

\subsection{Preliminaries for the Faster Implementation}\label{:sec:preliminaries-for-fast-implementation}

\paragraph{Remark Regarding High-Probability Guarantees.}
In the analysis of all algorithms in \cref{:sec:faster-implementation}, we assume that the outputs of the probabilistic subroutines always satisfy the high-probability guarantee. 
Since the probability of failure in all of these subroutines can be made to be an arbitrarily small polynomial at the cost of only $\polylog(n)$ factors blowup in the running time, and all algorithms in this section run in polynomial time, this assumption is justified.
Furthermore, we note that it is possible to deterministically check, in time $O(nk)$, whether the final clustering returned by the algorithm from \cref{:thm:fast-algorithm-exists} satisfies the required IP-stability condition, so the algorithm can be made to be a Las-Vegas algorithm with expected running time $\Tilde{O}(nk)$.

\paragraph{Importance Sampling.}
Importance sampling is a statistical technique used to estimate properties of a particular distribution, while only having samples generated from a different distribution than the distribution of interest. It achieves this by reweighting the samples in such a way that their distribution matches the target distribution, thus enhancing the efficiency of Monte Carlo algorithms, as in~\citep{karp1989monte}.\footnote{For more details, see this lecture note: \href{https://tinyurl.com/lect-imp-samp}{https://tinyurl.com/lect-imp-samp}.}
\begin{theorem}[Importance Sampling~\citep{karp1989monte}] \label{:thm:importance-sampling}
Let $x_1,\ldots,x_n \geq 0$ be non-negative values, and let $\hat{x}_1,\ldots,\hat{x}_n \geq 0$ be estimates of $x_1,\ldots,x_n$ such that $x_i \leq \hat{x}_i$ holds for all $i\in [n]$, and $\sum_{i = 1}^{n} \hat{x}_i = O\left(\sum_{i = 1}^{n} x_i\right)$.
Furthermore, let $i_1,\ldots,i_t$ be independent samples from the distribution over $\{1,\ldots,n\}$ that gives probability $\frac{\hat{x}_i}{\sum_{i=1}^{n}\hat{x}_i}$ to each $i$.
Then, $t=O(1/\epsilon^2)$ suffices so that, w.h.p,
\[
 (1-\epsilon)\sum_{i=1}^{n} x_i \leq \frac{\left(\sum_{i=1}^{n}\hat{x}_i\right)}{t}\sum_{j=1}^{t}\frac{x_{i_j}}{\hat{x}_{i_j}} \leq (1+\epsilon) \sum_{i=1}^{n} x_i.
\]
\end{theorem}

By the importance sampling technique, we prove the following lemma. Proof of~\cref{:lem:fast-average-distances} is deferred to~\cref{:sec:fast-average-distances}.

\begin{lemma}\label{:lem:fast-average-distances}
    Given a metric space $(X,d)$, a non-empty cluster $C \subseteq X$, a set $S \subseteq X$, and $0 < \epsilon \leq 1$, $\CalcAverage$ (Algorithm \ref{:alg:fast-average-distances}) runs in time $\Tilde{O}(\frac{|C|+|S|}{\epsilon^2})$, and w.h.p., for each $p \in S$, computes a $(1+\epsilon)$-approximation of $\avg(p,C)$.
    
    More precisely, for each $p \in S$, $\CalcAverage$ returns $\widehat{\avg}(p,C)$ s.t. $\avg(p,C) \leq \widehat{\avg}(p,C) \leq (1+\epsilon)\avg(p,C)$.
\end{lemma}
\begin{corollary}\label{:cor:fast-potential-of-one-cluster}
    There exists an algorithm that runs in time $\Tilde{O}(|C|/\epsilon^2)$, and for any non-empty subset $C \in X$, w.h.p. computes a $(1+\epsilon)$-approximation of $\Phi(C)$, where $0<\epsilon \le 1$.
    
    More precisely, for every non-empty $C\subset X$, the algorithm returns $\hat{\Phi}(C)$ s.t. $\Phi(C) \leq \hat{\Phi}(C) \leq (1+\epsilon)\Phi(C)$.
\end{corollary}
\begin{proof}
The algorithm calculates each estimate $\hat{\Phi}(C)$ by calling $\CalcAverage(C, C, \epsilon)$ and then setting $\hat{\Phi}(C) = \log(|C|)\cdot \sum_{p \in C}\widehat{\avg}(p,C)$, where $\{\widehat{\avg}(p,C)\}_{p \in C}$ are the outputs of $\CalcAverage$.
It is straightforward to see that the algorithm runs in time $\Tilde{O}(|C|/\epsilon^2)$.
Furthermore, since $\Phi(C) = \log(|C|) \cdot \sum_{p \in C} \avg(p,C)$ (see \cref{:sec:potential-function}), and the estimates returned by $\CalcAverage$ satisfy $\avg(p,C) \leq \widehat{\avg}(p,C) \leq (1+\epsilon)\avg(p,C)$,
\[
 \Phi(C) \leq \log(|C|)\cdot \sum_{p \in C}\widehat{\avg}(p,C) \leq (1+\epsilon)\Phi(C).
\]
Hence, $\Phi(C) \leq \hat{\Phi}(C) \leq (1+\epsilon)\Phi(C)$.
\end{proof}
\begin{corollary}\label{:cor:fast-potential-of-clustering}
    Given a metric space $(X,d)$ with $|X|=n$, a clustering $\cC$, and $0 < \epsilon \leq 1$, $\CalcPotential$ computes a $(1+\epsilon)$-approximation of $\Phi(\cC) = \sum_{C \in \cC}\Phi(C)$ in time $\Tilde{O}(n/\epsilon^2)$ w.h.p.
    
    More precisely, given a clustering $\cC$, $\CalcPotential$ returns an estimate $\hat{\Phi}(\cC)$ s.t. $\Phi(\cC) \leq \hat{\Phi}(\cC) \leq (1+\epsilon)\Phi(\cC)$ holds w.h.p.
\end{corollary}
\begin{proof}
Simply follows from the guarantee of \cref{:cor:fast-potential-of-one-cluster} for each cluster $C \in \cC$.
\end{proof}

We will use the above results throughout all of \cref{:sec:faster-implementation}. For now, we will use them in order to prove the following lemma, which provides a faster version of the $\split$ algorithm from \cref{:lem:slow-split}.
Both the procedure and it's analysis are analogous to the ones from \cref{:lem:slow-split}, with the main difference being that we use \cref{:cor:fast-potential-of-one-cluster} in order to compute (estimates of) potentials of clusters, rather than using the slow trivial way of computing potentials.
The proof of \cref{:lem:fast-split} is found in \cref{:sec:proof-of-fast-split}.

\begin{lemma}\label{:lem:fast-split}
    
    There exists an algorithm $\fastSplit$ running in time $\tilde{O}(n)$ that given a metric space $(X,d)$ with $n$ points and a $(k-1)$-clustering $\cC$, w.h.p., finds a cluster $C^* \in \cC$ and a partition $(C^*_1, C^*_2)$ of $C^*$ such that $\Phi(C^*_1)+\Phi(C^*_2) \leq \Phi(C^*) - \Omega(\frac{\Phi(\cC)}{k \log n})$.
\end{lemma}

Lastly, the following lemma will be useful in order to bound the amount by which values such as $\avg(p',C)$ change as we swap points in and out of a cluster $C$.
The proof of this lemma is provided in \cref{:sec:proof-of-change-in-average}.

\begin{lemma}\label{:lem:change-in-average}
    For every metric space $(X,d)$, every non-empty subset $C \subset X$, every point $p \in X \setminus C$, and every point $p' \in X$,
    \[
     |\avg(p',C\cup\{p\}) - \avg(p',C)| \leq \avg(p,C)/(|C|+1) = \avg(p,C \cup \{p\})/|C|.
    \]
\end{lemma}

\subsubsection{Proof of \texorpdfstring{\cref{:lem:fast-split}}{Lemma 6}} \label{:sec:proof-of-fast-split}
\begin{proof}
The procedure follows the same meta-algorithm as $\split$ (Algorithm \ref{:alg:split-procedure}), except that now potentials of clusters are computed via the algorithm from \cref{:cor:fast-potential-of-one-cluster} with $\epsilon = \frac{1}{100 \log n}$.
Furthermore, the constant in the $\Omega$ notation in the condition of the loop is now $\frac{1}{5}$ instead of $\frac{1}{4}$.

\paragraph{Correctness Analysis.}
When the algorithm returns, the computed estimates must satisfy $\hat{\Phi}(C^*_1)+\hat{\Phi}(C^*_2) \leq \hat{\Phi}(C^*)\left(1 - \frac{1}{5\log n}\right)$.
Therefore, since the estimates are $(1+\epsilon)$-approximations of the true values, with $\epsilon = \frac{1}{100\log n}$,
\begin{align*}
    \Phi(C^*_1)+\Phi(C^*_2)
    \leq \Phi(C^*)\cdot (1+\epsilon)\left(1 - \frac{1}{5\log n}\right)
    \leq \Phi(C^*)\cdot \left(1 - \frac{1}{6\log n}\right)
    = \Phi(C^*) - \Omega(\frac{\Phi(C^*)}{\log n})
\end{align*}

\paragraph{Running Time Analysis.}
This analysis follows the same template as the one of the slow $\split$ algorithm from \cref{:lem:slow-split};
Firstly, since we invoke the algorithm from \cref{:cor:fast-potential-of-one-cluster} with parameter $\epsilon=\frac{1}{\polylog(n)}$, we get that the running time of each such invocation on each cluster $C$ is $\Tilde{O}(|C|)$.
Thus, out implementation of lines \ref{:line:split:compute-potential-of-clustering} and \ref{:line:split:choose-cluster-to-split} of the $\split$ meta-algorithm (Algorithm \ref{:alg:split-procedure}) runs in time $\Tilde{O}(\sum_{C \in \cC}|C|) = \Tilde{O}(n)$, and our implementation for checking the condition of the loop in lines \ref{:line:split:loop-start}-\ref{:line:split:repeat-until-condition} runs in time $\Tilde{O}(|C^*_1|+|C^*_2|)\leq\Tilde{O}(n)$.
Furthermore, it is not hard to see that line \ref{:line:split:sample-partition} can be implemented in time $\Tilde{O}(|C^*|) \leq \Tilde{O}(n)$.
Therefore, to complete the analysis, we just need to show that only $\Tilde{O}(1)$ iterations of the loop occur in expectation.
Specifically, we will show that after each iteration of the loop, there is a probability of at least $\Tilde{\Omega}(1)$ that the condition for terminating the loop holds, i.e. that $\hat{\Phi}(C^*_1) + \hat{\Phi}(C^*_2) \leq \hat{\Phi}(C^*) \cdot (1-\frac{1}{5 \log n})$ holds, where $\hat{\Phi}(\cdot)$ represents our estimate of the real value $\Phi(\cdot)$ that we get by running \cref{:cor:fast-potential-of-one-cluster}:

Indeed, according to \cref{:obs:probability-of-good-split}, each split generated by line \ref{:line:split:sample-partition} has a probability of $\Tilde{\Omega}(1)$ to satisfy $\Phi(C^*_1) + \Phi(C^*_2) \leq \Phi(C^*) \cdot (1-\frac{1}{4\log n})$.
Whenever this event occurs, since the estimates produced by \cref{:cor:fast-potential-of-one-cluster} are $(1+\epsilon)$ approximations of the true values, we have
\begin{align*}
    \hat{\Phi}(C^*_1)+\hat{\Phi}(C^*_2)
    \leq \hat{\Phi}(C^*)\cdot (1+\epsilon)\left(1 - \frac{1}{4\log n}\right)
    \leq \hat{\Phi}(C^*)\cdot \left(1 - \frac{1}{5\log n}\right)
\end{align*}
as we needed.
This concludes the proof that, for each iteration of the loop of the meta-algorithm (Algorithm \ref{:alg:split-procedure}), there is a probability of at $\Tilde{\Omega}(1)$ that the loop terminates at the end of this iteration (regardless of what happened in previous iterations).
This implies that the expected number of iterations is $\Tilde{O}(1)$.

This concludes the running time analysis, and thus concludes the proof of \cref{:lem:fast-split}.
\end{proof}

\subsubsection{Proof of \texorpdfstring{\cref{:lem:change-in-average}}{Lemma 7}} \label{:sec:proof-of-change-in-average}

\begin{proof}[Proof of \cref{:lem:change-in-average}]
Since
\[
 \avg(p',C\cup\{p\}) = \frac{1}{|C \cup \{p\}|}\sum_{p'' \in C \cup \{p\}}d(p',p'')
  = \frac{1}{|C|+1}\sum_{p'' \in C \cup \{p\}}d(p',p'')
\]
and since
\[
 \frac{|C|}{|C|+1}\avg(p',C) = \frac{1}{|C|+1}\sum_{p'' \in C}d(p',p''),
\]
we get that $\avg(p',C\cup\{p\}) =  \frac{1}{|C|+1}d(p',p) + \frac{|C|}{|C|+1}\avg(p',C)$ and thus
\begin{equation}\label{:eq:inserting-point-to-cluster-changes-average-distances}
    \avg(p',C\cup\{p\}) - \avg(p',C) = \frac{1}{|C|+1}\left(d(p',p) - \avg(p',C)\right).
\end{equation}
Furthermore, since $\avg(p',C) = \frac{1}{|C|}\sum_{p'' \in C}d(p',p'')$, we have that
\[
 d(p',p) - \avg(p',C) = \frac{1}{|C|}\left(\sum_{p'' \in C}d(p',p) - d(p',p'')\right).
\]
By the triangle inequality for real numbers, the last equality implies that
\[
 |d(p',p) - \avg(p',C)| \leq \frac{1}{|C|}\left(\sum_{p'' \in C}|d(p',p) - d(p',p'')|\right)
\]
and by the triangle inequality in the metric space $(X,d)$, this implies that
\[
 |d(p',p) - \avg(p',C)| \leq \frac{1}{|C|}\left(\sum_{p'' \in C}|d(p,p'')|\right) = \avg(p,C)
\]
which means that
\[
 |d(p',p) - \avg(p',C)| \leq \frac{1}{|C|}\left(\sum_{p'' \in C}d(p,p'')\right)
 = \avg(p,C).
\]
Plugging the last inequality into \cref{:eq:inserting-point-to-cluster-changes-average-distances}, we get that
\[
 |\avg(p',C\cup\{p\}) - \avg(p',C)|
 = \frac{1}{|C|+1}|d(p',p) - \avg(p',C)|
 \leq \frac{1}{|C|+1}\avg(p,C),
\]
as we needed to prove.
\end{proof}

\subsection{Proof of \texorpdfstring{\cref{:thm:fast-algorithm-exists}}{Theorem 2}} \label{:sec:fast-alg-from-fast-epoch}
Now, we show that \cref{:thm:fast-algorithm-exists} can easily be proved using the algorithm {\sc Epoch} from \cref{:thm:fast-epoch} as a subroutine.

\begin{algorithm}
\caption{\fastLS~runs in time $\Tilde{O}(n k)$.}\label{:alg:fast-merge-and-split}
\KwData{$(X,d)$, $n = |X|$, $k \leq n$, $\alpha=16\log n$}
\KwResult{An $\alpha$-IP stable $k$-clustering of $(X,d)$}
$\cC' \leftarrow $ the output of the greedy algorithm of $k$-center~\citep{gonzalez1985clustering} on $X$, and $\epsilon \leftarrow 1/10$ \\
\Repeat{$\CalcPotential(\cC', \epsilon) \geq \frac{7}{8} \cdot \CalcPotential(\cC, \epsilon)$}{
    $\cC \leftarrow \cC'$, $\cC' \leftarrow ${\sc Epoch}($\cC$)
}
\Return $\cC'$
\end{algorithm}
\begin{proof}[Proof of \cref{:thm:fast-algorithm-exists}]
The implementation of \cref{:thm:fast-algorithm-exists} is described in Algorithm \ref{:alg:fast-merge-and-split}.
Next, we analyze the correctness and running time of the algorithm.

\paragraph{Correctness Analysis.}
The algorithm only returns a clustering $\cC'$ if the condition
\begin{align*}
    ``\CalcPotential(\cC', \epsilon) \geq \frac{7}{8} \cdot \CalcPotential(\cC, \epsilon)"
\end{align*}
is satisfied.
Since $\CalcPotential$ returns a $(1+\epsilon)$-approximation of the potential of a clustering, with one-sided error, this condition implies that $(1+\epsilon)\Phi(\cC') \geq \frac{7}{8}\cdot\Phi(\cC)$.
Since the algorithm always uses $\epsilon=1/10$, the above condition implies $\Phi(\cC') \geq \frac{10}{11}\cdot\frac{7}{8}\Phi(\cC) \geq \frac{3}{4} \cdot \Phi(\cC)$.
Therefore, when the above condition is satisfied, because $\cC'$ is the clustering returned by a call of {\sc Epoch} on $\cC$, \cref{:thm:fast-epoch} implies that $\cC'$ is an $O(\log n)$-IP stable clustering of $(X,d)$.
In summary, a clustering $\cC'$ is only returned if it is $O(\log n)$-IP stable.

\paragraph{Time Analysis.}
Note that the greedy algorithm of $k$-center algorithm runs in time $O(nk)$.
Since, by~\cref{:thm:fast-epoch}, {\sc Epoch} runs in time $\Tilde{O}(nk)$, and by~\cref{:cor:fast-potential-of-clustering}, each call to $\CalcPotential$ runs in time $\Tilde{O}(n)$, we get that each iteration of the loop in Algorithm~\ref{:alg:fast-merge-and-split} takes time $\Tilde{O}(nk)$.
So, it suffices to bound the number of iterations.
By~\cref{:lem:k-center-initialization}, the potential of the initial clustering is at most $\poly(n)$ times larger than the minimum possible potential of a $k$-clustering of $(X,d)$.
Furthermore, whenever the condition of the loop is not satisfied, since $\CalcPotential$ computes a $(1+\epsilon)$-approximation, we must have $\Phi(\cC') < (1+\epsilon) \cdot \frac{7}{8} \cdot \Phi(\cC)=\frac{77}{80}\cdot \Phi(\cC)$.
Thus, each iteration except the last one must decrease the potential of the maintained clustering by a constant factor; hence, that there can be at most $O(\log n)$ iterations.
So, the total runtime of Algorithm \ref{:alg:fast-merge-and-split} is $\Tilde{O}(nk)$.
\end{proof}

\subsection{Proof of \texorpdfstring{\cref{:thm:fast-epoch}}{Theorem 7}} \label{:sec:proof-of-fast-epoch}

In this section, we analyze Algorithm \ref{:alg:fast-epoch}, and show that it satisfies the conditions of \cref{:thm:fast-epoch}.
Initially, we introduce certain notations to facilitate the algorithm's analysis, which can be found in \cref{:def:step-types} and \cref{:def:old-averages}. Following this, we show that the algorithm is well defined, i.e.,  it never tries to read uninitialized values (\cref{:cl:values-defined}).
Then, we state some basic invariants that hold during the run of the algorithm (\cref{:cl:basic-invariants-of-fast-epoch}), and two additional lemmas that are proven using the invariants (\cref{:lem:cluster-distances-lemma} and \cref{:lem:steps-are-effective}). 
The correctness analysis of the algorithm is in \cref{:sec:correctness-analysis-of-fast-epoch}.
Prior to analyzing the algorithm's running time, we need to bound the number of various steps the algorithm executes. This is detailed in \cref{:sec:num-swap-and-merge-and-split-steps} and \cref{:sec:num-compute-steps}.
Then, we do the analysis of the running time in \cref{:sec:runtime-analysis-of-fast-epoch}.

\begin{proof}[Proof of \cref{:thm:fast-epoch}] 
By \cref{:cl:return-IP-stable-clustering} and \cref{:cl:return-clustering-with-less-potential}, when Algorithm \ref{:alg:fast-epoch} returns a clustering $\cC$, it is either $\alpha$-IP stable, or satisfies $\Phi(\cC) < \frac{3}{4} \cdot \Phi(\cC_{\mathrm{input}})$.
Furthermore, by \cref{:cor:final-proof-of-running-time-of-fast-epoch}, the algorithm runs in time $\Tilde{O}(nk)$.
\end{proof}

\begin{algorithm}
\caption{The procedure {\sc Epoch} from \cref{:thm:fast-epoch}. It implements a single epoch in the $\Tilde{O}(nk)$-time adaptation of the local search. The variables denoted by $progress(C)$ for each $C\in \cC$ are introduced solely for the sake of analysis. 
}\label{:alg:fast-epoch}
\KwData{$(X,d)$, $n = |X|$, $k \leq n$, $\alpha=16\log n$, and a $k$-clustering $\cC_{input}$ of $X$} 
\KwResult{A $k$-clustering $\cC$ of $(X,d)$ that is either $\alpha$-IP stable or has $\Phi(\cC) < 3\Phi(\cC_{input})/4$}
$\cC \leftarrow \cC_{input}$, $\epsilon \leftarrow 1/10$\\
$\hat{\Phi}(\cC) \leftarrow \CalcPotential(\cC,\epsilon)$ \label{:line:fast-epoch:computing-initial-estimated-potential}\\
$t^* \leftarrow \Omega(\frac{\hat{\Phi}(\cC)}{k\log(n)})/(4\log n)$ \label{:line:fast-epoch:definition-of-t-start}\\
$\clusterToRecompute \leftarrow \cC \qquad\rhd\text{the set of clusters whose average estimates need to be recomputed.}$\\
\While{$\clusterToRecompute\neq \emptyset$ \textbf{or} exists $p \in X$ and $C' \in \cC \setminus \{C(p)\}$ s.t. $C(p) \neq \{p\}$ and $\frac{|C(p)|}{|C(p)\setminus\{p\}|}\cdot \widehat{\avg}(p,C(p)) > \frac{\alpha}{2}\cdot\widehat{\avg}(p,C')$\label{:line:fast-epoch:main-loop}}{
    \uIf{$\clusterToRecompute = \emptyset$\label{:line:fast-epoch:swap-or-recompute-if-statement}}{
        {\bf choose} $p \in X$ and $C' \in \cC$ as in the condition of the {\bf while} loop \label{:line:fast-epoch:choose-point-to-swap} \label{:line:fast-epoch:before-swap}\\
        {\bf move} $p$ from $C=C(p)$ to $C'$ \label{:line:fast-epoch:actual-swap}\\
        $progressIncrease \leftarrow (\frac{1}{1+\epsilon}\widehat{\avg}(p,C) - error(C))/2$ \label{:line:fast-epoch:set-progress-increase}\\
        \For{$C'' \in \{C,C'\}$}{
            $error(C'') \leftarrow error(C'') + (\widehat{\avg}(p,C'')+error(C''))/|C''|$ \label{:line:fast-epoch:increase-error}\\
            $progress(C'') \leftarrow progress(C'') + progressIncrease$\label{:line:fast-epoch:increase-progress}\\
            $numSwaps(C'') \leftarrow numSwaps(C'') + 1$\label{:line:fast-epoch:increase-numSwaps}\\
            \If{$error(C'') > t^*/(100\alpha|C''|)$ \textbf{or} $numSwaps(C'') > \widehat{size}(C'')/2$\label{:line:fast-epoch:if-statement-error-too-large}}{
                {\bf add} $C''$ to $\clusterToRecompute$ \label{:line:fast-epoch:recompute-because-error-too-large}
            }
        }
    }
    \Else{
        $C \leftarrow pop(\clusterToRecompute)$ \label{:line:fast-epoch:select-cluster-to-recompute}\label{:line:fast-epoch:remove-cluster-from-clusters-to-recompute}\\
        $\{\widehat{\avg}(p,C)\}_{p \in X} \leftarrow \CalcAverage(C,X,\epsilon)$\label{:line:fast-epoch:recompute-averages}\\
        $error(C) \leftarrow 0$ \label{:line:fast-epoch:reset-error}\\
        $progress(C) \leftarrow 0$\label{:line:fast-epoch:reset-progress}\\
        $\widehat{size}(C) \leftarrow |C|$, $numSwaps(C) \leftarrow 0$\\
        \If{$\CalcPotential(\cC,\epsilon)< \frac{1+\epsilon}{2} \cdot \hat{\Phi}(\cC)$}{
            \Return $\cC$ \label{:line:fast-epoch:return-clustering-with-less-potential}
        }
        \If{exists cluster $C' \in \cC \setminus \{C\}$ such that $\frac{\min\{|C|,|C'|\}}{|C'|}\sum_{p \in C'}\widehat{\avg}(p,C) < t^*$\label{:line:fast-epoch:merge-and-split-if-statement}}{
            $C'' \leftarrow C \cup C'\quad\rhd$ merge clusters $C,C'$\\
            $\cC \leftarrow (\cC\setminus \{C, C'\}) \cup \{C''\}$ \label{:line:fast-epoch:changing-clustering-in-merge} \\
            $\clusterToRecompute \leftarrow (\clusterToRecompute \setminus \{C,C'\}) \cup \{C''\}$\label{:line:fast-epoch:changing-clusterstorecompute-in-merge}\\
            $\fastSplit(\cC) \qquad \rhd$ a cluster $C^* \in \cC$ is split into $(C^*_1, C^*_2)$ (see \cref{:lem:fast-split})\\
            $\cC \leftarrow (\cC\setminus \{C^*\}) \cup \{C^*_1, C^*_2\}$ \label{:line:fast-epoch:changing-clustering-in-split}\\
            $\clusterToRecompute \leftarrow (\clusterToRecompute\setminus \{C^*\}) \cup \{C^*_1, C^*_2\}$ \label{:line:fast-epoch:changing-clusterstorecompute-in-split}
        }
    }
}
\Return $\cC$ \label{:line:fast-epoch:return-IP-stable-clustering}
\end{algorithm}

\begin{definition}[Step Types]\label{:def:step-types}
We refer to each iteration of the while loop in line \ref{:line:fast-epoch:main-loop} of Algorithm \ref{:alg:fast-epoch} as a ``step''.
We say that an iteration is a swap step if the condition in line \ref{:line:fast-epoch:swap-or-recompute-if-statement} was satisfied during that iteration, and otherwise we say that it is a \emph{recompute} step on the cluster $C$ selected in line \ref{:line:fast-epoch:select-cluster-to-recompute}.
If the condition in line \ref{:line:fast-epoch:merge-and-split-if-statement} is satisfied during a recompute step, then we furthermore refer to this step as a merge and split step.
\end{definition}

\begin{definition} \label{:def:old-averages}
At any point during the execution of Algorithm \ref{:alg:fast-epoch}, for every point $p$ and cluster $C$, we let $\widetilde{\avg}(p,C)$ denote the value that was the true value of $\avg(p,C)$ during the previous time the algorithm executed line \ref{:line:fast-epoch:recompute-averages} on $C$.
By \cref{:lem:fast-average-distances}, $\widetilde{\avg}(p,C) \leq \widehat{\avg}(p,C) \leq (1+\epsilon)\widetilde{\avg}(p,C)$ holds for every $p \in X, C \in \cC$ at all times after $\widehat{\avg}(p,C)$ is first initialized.
\end{definition}

\begin{claim} \label{:cl:values-defined}
At the beginning of every iteration of the loop in line \ref{:line:fast-epoch:main-loop} of Algorithm \ref{:alg:fast-epoch}, for  every cluster $C \in \cC \setminus ClusterToRecompute$, and every point $p \in X$, the values $error(C)$, $\widehat{size}(C)$, $numSwaps(C)$, $progress(C)$, and $\widehat{\avg}(p,C)$ have already been initialized in a previous recompute step on cluster $C$.
\end{claim}

\begin{claim} \label{:cl:basic-invariants-of-fast-epoch}
At the beginning of every iteration of the loop in line \ref{:line:fast-epoch:main-loop} of Algorithm \ref{:alg:fast-epoch}, for every cluster $C \in \cC \setminus ClustersToRecompute$, the following invariants hold:
\begin{enumerate}
    \item $numSwaps(C) \leq \widehat{size}(C)/2$,
    \item $error(C) \leq t^*/(100\alpha|C|)$,
    \item $\left|\widehat{size}(C) - |C|\right| \leq numSwaps(C)$; and,
    \item for every point $p' \in X$, $|\widetilde{\avg}(p',C) - \avg(p',C)| \leq error(C)$.
\end{enumerate}
\end{claim}

\begin{corollary} \label{:cor:size-didnt-change-much}
At the beginning of every iteration of the loop in line \ref{:line:fast-epoch:main-loop} of Algorithm \ref{:alg:fast-epoch}, for every cluster $C \in \cC \setminus \clusterToRecompute$, the inequalities $|C|/2 \leq \widehat{size}(C) \leq 2|C|$ hold.
\end{corollary}
\begin{proof}[Proof of \cref{:cor:size-didnt-change-much}]
This corollary follows immediately from invariants 1 and 3 of \cref{:cl:basic-invariants-of-fast-epoch}
\end{proof}

\begin{lemma}\label{:lem:cluster-distances-lemma}
At the beginning of every iteration of the loop in line \ref{:line:fast-epoch:main-loop} of Algorithm \ref{:alg:fast-epoch}, for every two different clusters $C,C' \in \cC \setminus \clusterToRecompute$, and every point $p \in X$, the inequalities $\widetilde{\avg}(p,C) + \widetilde{\avg}(p,C') \geq (3/16)\frac{t^*}{\min\{|C|,|C'|\}}$ and $\avg(p,C)+\avg(p,C') \geq (1/8)\frac{t^*}{\min\{|C|,|C'|\}}$ hold.
\end{lemma}

\begin{lemma}\label{:lem:steps-are-effective}
At every swap step made by Algorithm \ref{:alg:fast-epoch}, when the execution is at line \ref{:line:fast-epoch:before-swap}, the inequality $\avg(p,C(p)) \geq \frac{t^*}{10\min\{|C(p)|,|C'|\}}$ holds.
\end{lemma}

\subsubsection{Proof of \texorpdfstring{\cref{:cl:values-defined}}{Claim 3}} \label{:sec:proof-of-cl-values-defined}
We start by stating and proving the following observation.
\begin{observation} \label{:obs:cluster-to-recompute-get-recomputed}
If $C$ belongs to $\cC \setminus \clusterToRecompute$ at the beginning of iteration $\tau$ of the loop in line \ref{:line:fast-epoch:main-loop} of Algorithm \ref{:alg:fast-epoch}, and if this cluster belonged to $\cC \cap \clusterToRecompute$ at the beginning of an earlier iteration $\tau' < \tau$, then there exists some iteration $\tau' \leq \tau'' < \tau$ such that $\tau''$ is a recompute step on cluster $C$.
\end{observation}
\begin{proof}
The only line of the algorithm that might remove a cluster in $\cC$ from $\clusterToRecompute$ is line \ref{:line:fast-epoch:select-cluster-to-recompute}, and clusters that get removed from $\cC$ never get inserted into $\cC$ again.
So, since $C \in \cC \setminus \clusterToRecompute$ holds at the beginning of iteration $\tau$, and $C \in \cC \cap \clusterToRecompute$ held at the beginning of iteration $\tau'$, and since clusters that got removed from $\cC$ never get reinserted into $\cC$, there must be some iteration $\tau' \leq \tau'' < \tau$ such that $C$ was removed from $\clusterToRecompute$ while being in $\cC$ during iteration $\tau''$.
So since the only way for a cluster to get removed from $\clusterToRecompute$ without getting removed from $\cC$ is by line \ref{:line:fast-epoch:select-cluster-to-recompute} of the algorithm, it must be that $\tau'$ was a recompute step on cluster $C$.
\end{proof}
We are now ready to prove \cref{:cl:values-defined}.
\begin{proof}[Proof of \cref{:cl:values-defined}]
Let $\tau$ be the current iteration, and let $\tau_{C}$ be the first iteration such that $C$ was in $\cC$ at the beginning of iteration $\tau_C$.
We will begin by proving that $C$ was in the set $\clusterToRecompute$ at the beginning of iteration $\tau_{C}$:
On the one hand, if $\tau_{C}$ was the first iteration of the algorithm, then $C$ must have been part of the input clustering, which means that $C$ was in the set $\clusterToRecompute$ at the start of iteration $\tau_{C}$.
On the other hand, if $\tau_{C}$ was not the first iteration of the algorithm, then $C$ must have been added to $\cC$ during iteration $\tau_{C}-1$.
In that case, since the only lines that can add a cluster to $\cC$ during an iteration are line \ref{:line:fast-epoch:changing-clustering-in-merge} and line \ref{:line:fast-epoch:changing-clustering-in-split}, and since these lines are immediately followed by lines that insert the same cluster into $\clusterToRecompute$, it must be that $C$ was added to $\clusterToRecompute$ during iteration $\tau_{C}-1$.
Furthermore, it is impossible that $C$ was removed from $\clusterToRecompute$ in iteration $\tau_{C}-1$ after being added to it in that same iteration, because the only way that could happen is if it was added by line \ref{:line:fast-epoch:changing-clustering-in-merge} and then removed by line \ref{:line:fast-epoch:changing-clustering-in-split}, which would imply that $C$ was removed from $\cC$ by line \ref{:line:fast-epoch:changing-clusterstorecompute-in-split} in iteration $\tau_{C}-1$, which contradicts the fact that it belongs to $\cC$ at the beginning of iteration $\tau_{C}$.
To summarize, we proved that $C$ must have belonged to $\clusterToRecompute$ at the beginning of iteration $\tau_{C}$.
Therefore, by \cref{:obs:cluster-to-recompute-get-recomputed}, there must have been an iteration $\tau_{C} \leq \tau' < \tau$ that was a recompute step on $C$.
\end{proof}

\subsubsection{Proof of \texorpdfstring{\cref{:cl:basic-invariants-of-fast-epoch}}{Claim 4}} \label{:sec:proof-of-cl-basic-invariants-of-fast-epoch}
\begin{proof}[Proof of \cref{:cl:basic-invariants-of-fast-epoch}]
We will prove the claim by induction on the iteration number.
The base of the induction holds trivially, since there are no clusters in $\cC \setminus \clusterToRecompute$ at the beginning of the first iteration.
So, from now on, we assume the invariants held at the beginning of iteration $\tau$, and our goal is to prove that they hold at the beginning of iteration $\tau+1$.
We divide the proof into three cases based on whether $\tau$ is a swap step, a merge and split step, or a recompute step that isn't a merge and split step.

\paragraph{Case: $\tau$ is a swap step.}
In this case, no new clusters were added to the set $\cC \setminus \clusterToRecompute$ during iteration $\tau$, and the only clusters for which any invariant may have become invalidated are the clusters $C$ and $C'$ selected by the swap step.
So, we just need to show that the invariants hold for each $C'' \in \{C,C'\}$ such that $C''$ is still in $\cC \setminus \clusterToRecompute$ at the beginning of iteration $\tau+1$:
\begin{itemize}[leftmargin=*]
    \item \textbf{invariants 1 and 2.} Since $C''$ is still in $\cC \setminus \clusterToRecompute$ at the beginning of iteration $\tau+1$, the condition of the if statement in line \ref{:line:fast-epoch:if-statement-error-too-large} of the algorithm must not have been met for this $C''$, which exactly means that $error(C'') \geq t^*/(100\alpha|C''|)$ and $numSwaps(C'') \leq \widehat{size}(C'')/2$, as stated in the invariants.
    \item \textbf{invariant 3.} Since the invariant $\left|\widehat{size}(C'') - |C''|\right| \leq numSwaps(C'')$ held at the beginning of iteration $\tau$, since $\widehat{size}(C'')$ did not change during iteration $\tau$, since $|C''|$ changed by at most $1$ during iteration $\tau$, and since $numSwaps(C'')$ was increased by $1$ during iteration $\tau$, the inequality $\left|\widehat{size}(C'') - |C''|\right| \leq numSwaps(C'')$ must still hold at the beginning of iteration $\tau+1$.
    \item \textbf{invariant 4.} Fix some point $p' \in X$, and we'll show that the invariant holds for this point at the start of iteration $\tau+1$.
    Since the invariant held at the beginning of iteration $\tau$, and since $\widetilde{\avg}(p',C'')$ did not change during iteration $\tau$, it is enough if we upperbound the absolute value of the change to $\avg(p',C'')$ during iteration $\tau$ by the amount that $error(C'')$ was increased during iteration $\tau$:
    For each $\tau' \in \{\tau,\tau+1\}$, let $C''_{\tau'}$ and $error(C''_{\tau'})$ denote the state of the cluster $C''$ and the value of $error(C'')$ at the beginning of iteration $\tau'$.
    Furthermore, since, for every $p'' \in X$, the values of $\widehat{\avg}(p'',C'')$ and $\widetilde{\avg}(p'',C'')$ did not change during iteration $\tau$, we will write these without a subscript to denote their value during the whole of iteration $\tau$.
    So, let $p$ be the point selected by line \ref{:line:fast-epoch:choose-point-to-swap} of the algorithm during iteration $\tau$.
    By \cref{:lem:change-in-average}, we must have
    \[
     |\avg(p',C'_{\tau+1}) - \avg(p',C'_{\tau})| = |\avg(p',C'_{\tau} \cup \{p\}) - \avg(p',C'_{\tau})| \leq \avg(p,C'_{\tau})/|C'_{\tau+1}|
    \]
    and
    \begin{align*}
        |\avg(p',C_{\tau}) - \avg(p',C_{\tau+1})| = |\avg(p',C_{\tau+1} \cup \{p\}) - \avg(p',C_{\tau+1})| 
        &\leq \frac{\avg(p,C_{\tau+1} \cup \{p\})}{|C_{\tau+1}|} \\
        &= \frac{\avg(p,C_{\tau})}{|C_{\tau+1}|},
    \end{align*}
    which means that $|\avg(p',C''_{\tau}) - \avg(p',C''_{\tau+1})| \leq \avg(p,C''_{\tau})/|C''_{\tau+1}|$ holds regardless of which $C'' \in \{C,C'\}$ we are dealing with.
    Since the invariant held at the beginning of iteration $\tau$, we know that $\avg(p,C''_{\tau}) \leq \widetilde{\avg}(p,C'') + error(C''_{\tau}) \leq \widehat{\avg}(p,C'') + error(C''_{\tau})$.
    Together, the last two inequalities imply that $|\avg(p',C''_{\tau}) - \avg(p',C''_{\tau+1})| \leq \left( \widehat{\avg}(p,C'') + error(C''_{\tau}) \right)/|C''_{\tau+1}|$.
    Since $|C''_{\tau+1}|$ is exactly the size of $C''$ after line \ref{:line:fast-epoch:actual-swap} is executed in iteration $\tau$, this means that the absolute value of the change to $\avg(p',C'')$ during iteration $\tau$ is upperbounded by the amount that $error(C'')$ was increased during iteration $\tau$, as we needed to prove.
\end{itemize}

\paragraph{Case: $\tau$ is a merge and split step.}
In this case, during step $\tau$, all the clusters that got removed from $\clusterToRecompute$ also got removed from $\cC$, and all the clusters that got added to $\cC$ also got added to $\clusterToRecompute$.
There, all clusters that belong to $\cC \setminus \clusterToRecompute$ at the beginning of iteration $\tau+1$ were also in that set at the beginning of iteration $\tau$.
Furthermore, since the cluster $C$ selected at line \ref{:line:fast-epoch:select-cluster-to-recompute} of the algorithm is one of the clusters that for removed from $\cC$ at step $\tau$ (line \ref{:line:fast-epoch:changing-clustering-in-merge}), we get that, for every cluster $\overline{C}$ that belongs to $\cC \setminus \clusterToRecompute$ at the beginning of iteration $\tau+1$, the values of $numSwaps(\overline{C})$, $\hat{size}(\overline{C})$, $error(\overline{C})$, $|\overline{C}|$, $\widetilde{\avg}(p',\overline{C})$, $\avg(p',\overline{C})$ were not changed during step $\tau$, so the invariants all remain.

\paragraph{Case: $\tau$ is a recompute step that isn't a merge and split step.}
Let $C$ be the cluster chosen at line \ref{:line:fast-epoch:select-cluster-to-recompute} during step $\tau$.
Then, for all clusters other than $C$, it's clear that none of the invariants became invalidated during step $\tau$.
So, we just need to prove that the invariants hold for cluster $C$ at the beginning of iteration $\tau+1$:
\begin{itemize}[leftmargin=*]
    \item \textbf{invariant 1.} Since iteration $\tau$ sets $numSwaps(C)$ to $0$, this invariant holds.
    \item \textbf{invariant 2.} Since iteration $\tau$ sets $error(C)$ to $0$, this invariant holds.
    \item \textbf{invariant 3.} Since iteration $\tau$ sets $\widehat{size}(C)$ to $|C|$, this invariant holds.
    \item \textbf{invariant 4.} For all $p' \in X$, since the value of $\widetilde{\avg}(p',C)$ at the beginning of iteration $\tau + 1$ is defined as the value of $\avg(p',C)$ during the last recompute step on cluster $C$ that happened before iteration $\tau+1$ (see \cref{:def:old-averages}), and since iteration $\tau$ is a recompute step on cluster $C$, we get $\widetilde{\avg}(p',C)$ at the beginning of iteration $\tau+1$ is equal to the value of $\avg(p',C)$ at iteration $\tau$. Since the value of $\avg(p',C)$ does not change during iteration $\tau$, this implies that $|\widetilde{\avg}(p',C) - \avg(p',C)|=0$ holds at the beginning of iteration $\tau+1$, so the invariant holds at that time.
\end{itemize}
This concludes the proof of \cref{:cl:basic-invariants-of-fast-epoch}.
\end{proof}

\subsubsection{Proof of Lemma \ref{:lem:cluster-distances-lemma}} \label{:sec:proof-of-lem-cluster-distances-lemma}

We begin by proving the following claim and corollary.

\begin{claim}\label{:cl:clusters-far-after-recompute}
At the end of every iteration of the loop in line \ref{:line:fast-epoch:main-loop} of Algorithm \ref{:alg:fast-epoch}, if this iteration was a recompute step for cluster $C$, and if at the end of the iteration $C$ is in $\cC$, then, for every cluster $C' \in \cC$, and every point $p \in X$, the inequality $\avg(p,C)+\avg(p,C') \geq \frac{t^*}{(1+\epsilon)\min\{|C|,|C'|\}}$ holds.
\end{claim}
\begin{proof}
At the end of a recompute step on cluster $C$, if $C$ is still in $\cC$, then line \ref{:line:fast-epoch:changing-clustering-in-merge} of the algorithm must not have been reached in this iteration, which means that the condition of the if statement in line \ref{:line:fast-epoch:merge-and-split-if-statement} must not have been satisfied.
Therefore, for every cluster $C' \in \cC$, it must be that the estimated averages $\{\avg(p,C)\}_{p \in X}$ computed in this iteration satisfy $\frac{\min\{|C|,|C'|\}}{|C'|}\sum_{p \in C'}\widehat{\avg}(p,C) \geq t^*$.
Since these estimates were computed in this iteration with accuracy $(1+\epsilon)$, and since true values, as well as the sizes of the clusters, did not change during the iteration, the true values must satisfy $\frac{\min\{|C|,|C'|\}}{|C'|}\sum_{p \in C'}\avg(p,C) \geq t^*/(1+\epsilon)$ at the end of the iteration.
Thus, at the end of the iteration,
\[
 \frac{t^*}{(1+\epsilon)\min\{|C|,|C'|\}}
 \leq \frac{1}{|C'|}\sum_{p_1 \in C'}\avg(p,C)
 = \frac{1}{|C'| \cdot |C|}\sum_{p_1 \in C'} \sum_{p_2 \in C}d(p_1,p_2).
\]
So, for every point $p \in X$, by \cref{:obs:triangle-inequality-for-average-distance},
\[
 \avg(p,C) + \avg(p,C') \geq \frac{1}{|C'| \cdot |C|}\sum_{p_1 \in C'} \sum_{p_2 \in C}d(p_1,p_2) \geq \frac{t^*}{(1+\epsilon)\min\{|C|,|C'|\}},
\]
as we needed to prove.
\end{proof}
\begin{corollary}\label{:cor:estimates-show-clusters-far-after-recompute}
At the beginning of every iteration of the loop in line \ref{:line:fast-epoch:main-loop} of Algorithm \ref{:alg:fast-epoch}, if the previous iteration was a recompute step for cluster $C$, and if at the beginning of the current iteration $C$ is in $\cC \setminus \clusterToRecompute$, then, for every cluster $C' \in \cC \setminus \clusterToRecompute$, and every point $p \in X$, the inequality $\widetilde{\avg}(p,C)+\widetilde{\avg}(p,C') \geq (3/8)\frac{t^*}{\min\{\widetilde{size}(C),\widetilde{size}(C')\}}$ holds.
\end{corollary}
\begin{proof}
Let $\tau$ be the current iteration.
By \cref{:cl:clusters-far-after-recompute}, $\avg(p,C)+\avg(p,C') \geq \frac{t^*}{(1+\epsilon)\min\{|C|,|C'|\}}$ must hold at the end of the previous iteration, and thus at the start of the current iteration.
Thus, by invariants 4 and 2 from \cref{:cl:basic-invariants-of-fast-epoch}, at the beginning of the current iteration,
\begin{align*}
    \widetilde{\avg}(p,C)+\widetilde{\avg}(p,C')
    &\geq \frac{t^*}{(1+\epsilon)\min\{|C|,|C'|\}} - error(C) - error(C')\\
    &\geq \frac{t^*}{(1+\epsilon)\min\{|C|,|C'|\}} - t^*/(100\alpha |C|) - t^*/(100\alpha |C'|)\\
    &\geq \frac{t^*}{(1+\epsilon)\min\{|C|,|C'|\}} - t^*/(50 \min\{|C|,|C'|\})\\
    &\geq \frac{3}{4} \cdot \frac{t^*}{\min\{|C|,|C'|\}}.
\end{align*}
So, by \cref{:cor:size-didnt-change-much}, at the beginning of the current iteration, 
\begin{align*}
    \widetilde{\avg}(p,C)+\widetilde{\avg}(p,C') \geq (3/8)\frac{t^*}{\min\{\widetilde{size}(C),\widetilde{size}(C')\}},    
\end{align*}
as we needed to prove.
\end{proof}
We are now ready to prove the lemma.
\begin{proof}[Proof of \cref{:lem:cluster-distances-lemma}]
Let $\tau$ be the current iteration, let $p \in X$ be some point, and let $C$ and $C'$ be two different clusters that belong to $\cC \setminus \clusterToRecompute$ at the beginning of iteration $\tau$.
We need to prove that, at the beginning of iteration $\tau$, the inequalities $\widetilde{\avg}(p,C) + \widetilde{\avg}(p,C') \geq \frac{t^*}{2\min\{|C|,|C'|\}}$ and $\avg(p,C)+\avg(p,C') \geq \frac{t^*}{4\min\{|C|,|C'|\}}$ hold.

Let $\tau' < \tau$ be the last step before $\tau$ that was a recompute step on either of the clusters $C$ and $C'$ (such a step exists by \cref{:cl:values-defined}).
By the definition of $\widetilde{\avg}(p,C)$ and $\widetilde{\avg}(p,C')$ (\cref{:def:old-averages}), this means that the values of $\widetilde{\avg}(p,C)$ and $\widetilde{\avg}(p,C')$ did not change between iteration $\tau'$ and iteration $\tau$, and by looking at the algorithm, it means that $\widehat{size}(C)$ and $\widehat{size}(C')$ did not change between iteration $\tau'$ and iteration $\tau'$.
Furthermore, by \cref{:cor:estimates-show-clusters-far-after-recompute}, at the beginning of iteration $\tau'+1$, the inequality
\[
 \widetilde{\avg}(p,C)+\widetilde{\avg}(p,C') \geq (3/8)\frac{t^*}{\min\{\widetilde{size}(C),\widetilde{size}(C')\}}
\]
held.
So, at the beginning of the current iteration $\tau$, this inequality still holds.
Therefore, at the beginning of iteration $\tau$, by \cref{:cor:size-didnt-change-much},
\[
 \widetilde{\avg}(p,C)+\widetilde{\avg}(p,C') \geq (3/16)\frac{t^*}{\min\{|C|,|C'|\}},
\]
which is the first inequality that we need for the proof of the lemma.
Furthermore, by invariants 4 and 2 of \cref{:cl:basic-invariants-of-fast-epoch}, the above inequality implies
\begin{align*}
    \avg(p,C) + \avg(p,C')
    &\geq \frac{3}{16} \cdot \frac{t^*}{\min\{|C|,|C'|\}} - error(C) - error(C')\\
    &\geq \frac{3}{16} \cdot \frac{t^*}{\min\{|C|,|C'|\}} - t^*/(100\alpha |C'|) - t^*/(100\alpha |C|)\\
    &\geq \frac{2}{16} \cdot \frac{t^*}{\min\{|C|,|C'|\}},
\end{align*}
which is the second inequality that we need for the proof of the lemma.
This concludes the proof of \cref{:lem:cluster-distances-lemma}.
\end{proof}

\subsubsection{Proof of \texorpdfstring{\cref{:lem:steps-are-effective}}{Lemma 9}} \label{:sec:proof-of-lem-steps-are-effective}
We begin with the following observation.
\begin{observation} \label{:obs:swap-decided-by-estimates}
At every swap step made by Algorithm \ref{:alg:fast-epoch}, when the execution is at line \ref{:line:fast-epoch:before-swap}, it must be that $\widetilde{\avg}(p,C(p)) \geq 2\widetilde{\avg}(p,C')$.
\end{observation}
\begin{proof}
By the definition of line \ref{:line:fast-epoch:before-swap} of the algorithm, the chosen point $p$ and cluster $C'$ satisfy $C(p) \neq \{p\}$ and $(|C(p)|/|C(p)\setminus\{p\}|)\widehat{\avg}(p,C(p)) > (\alpha/2) \widehat{\avg}(p,C')$.
Since $C(p) \neq \{p\}$, we have $|C(p)\setminus\{p\}| = |C(p)| - 1 \geq 1$, so $(|C(p)|/|C(p)\setminus\{p\}|) \leq 2$, so the above inequality implies that
\[
 \frac{\widehat{\avg}(p,C(p))}{\widehat{\avg}(p,C')} \geq \alpha/4.
\]
Furthermore, as explained in \cref{:def:old-averages}, $\frac{\widehat{\avg}(p,C(p))}{\widehat{\avg}(p,C')} \leq (1+\epsilon)\frac{\widetilde{\avg}(p,C(p))}{\widetilde{\avg}(p,C')}$, so the previous inequality implies that
\[
 \frac{\widetilde{\avg}(p,C(p))}{\widetilde{\avg}(p,C')} \geq \alpha/(4(1+\epsilon)) \geq 2,
\]
which means that $\widetilde{\avg}(p,C(p)) \geq 2\widetilde{\avg}(p,C')$, as we needed to prove.
\end{proof}
We are now ready to prove \cref{:lem:steps-are-effective}.
\begin{proof}[Proof of \cref{:lem:steps-are-effective}]
Since the step is a swap step, the condition in the if statement at line \ref{:line:fast-epoch:swap-or-recompute-if-statement} must have been met, which means that $\clusterToRecompute$ is empty.
So, since the clusters $C'$ and $C=C(p)$ are in $\cC$, we get that they are in $\cC \setminus ClustersToRecompute$ at the start of this step, so \cref{:lem:cluster-distances-lemma} tells us that $\widetilde{\avg}(p,C) + \widetilde{\avg}(p,C') \geq (3/16)\frac{t^*}{\min\{|C|,|C'|\}}$.
Furthermore, by \cref{:obs:swap-decided-by-estimates}, we have $\widetilde{\avg}(p,C) \geq (2/3)(\widetilde{\avg}(p,C) + \widetilde{\avg}(p,C'))$, so the previous inequality gives use that $\widetilde{\avg}(p,C) \geq (2/16)\frac{t^*}{\min\{|C|,|C'|\}} \geq \frac{t^*}{10\min\{|C|,|C'|\}}$, as we needed to prove.
\end{proof}

\subsubsection{Correctness Analysis of Algorithm \ref{:alg:fast-epoch}} \label{:sec:correctness-analysis-of-fast-epoch}
In this section, we prove the correctness of the algorithm, by proving the following \cref{:cl:return-IP-stable-clustering} and \cref{:cl:return-clustering-with-less-potential}

\begin{claim}\label{:cl:return-IP-stable-clustering}
If the algorithm returns at line \ref{:line:fast-epoch:return-IP-stable-clustering}, then the output is $\alpha$-IP stable.
\end{claim}
\begin{proof}
The proof is by contradiction.
So, assume that when the algorithm returned at line \ref{:line:fast-epoch:return-IP-stable-clustering}, there existed some point $p \in X$ and cluster $C' \in \cC \setminus \{C(p)\}$ such that $C(p) \neq \{p\}$ and $\avg(p,C(p)\setminus\{p\}) > \alpha \avg(p,C')$, and our goal is to show that this leads to a contradiction.

When the algorithm returned, since $C(p) \neq \{p\}$, it must be that $|C(p) \setminus \{p\}| \geq 1$.
So, since $\avg(p,C(p)\setminus\{p\}) = \frac{1}{|C(p)\setminus\{p\}|}\sum_{p' \in C(p)} d(p,p') = \frac{|C(p)|}{|C(p)\setminus\{p\}|}\avg(p,C(p))$, the inequality $\avg(p,C(p)\setminus\{p\}) > \alpha \avg(p,C')$ gives us
\begin{align} \label{:eq:fast-epoch:proof-of-return-IP-stable-clustering:contradiction-assumption}
    \frac{|C(p)|}{|C(p)\setminus\{p\}|}\avg(p,C(p)) > \alpha\avg(p,C')
\end{align}
and
\begin{align} \label{:eq:fast-epoch:proof-of-return-IP-stable-clustering:result-of-contradiction-assumption}
    \avg(p,C(p)) > \avg(p,C').
\end{align}
Just before returning in line \ref{:line:fast-epoch:return-IP-stable-clustering}, the algorithm must have exited the while loop in line \ref{:line:fast-epoch:main-loop}, which means that the condition of that loop was not satisfied.
Therefore, $\clusterToRecompute$ was empty when the algorithm exited the loop, which means that $C(p)$ and $C'$ were both in $\cC \setminus \clusterToRecompute$.
Thus, by \cref{:lem:cluster-distances-lemma}, $\avg(p,C(p))+\avg(p,C') \geq (1/8)\frac{t^*}{\min\{|C(p)|,|C'|\}}$ held when the algorithm exited the loop.
So, since \cref{:eq:fast-epoch:proof-of-return-IP-stable-clustering:result-of-contradiction-assumption} also held when the algorithm exited the loop, we get that $\avg(p,C(p)) \geq (1/16)\frac{t^*}{\min\{|C(p)|,|C'|\}}$ must have held.
Therefore, by invariants 4 and 2 from \cref{:cl:basic-invariants-of-fast-epoch},
\begin{equation}\label{:eq:fast-epoch:proof-of-return-IP-stable-clustering:relating-old-avg-to-C}
\begin{aligned}
    \widetilde{\avg}(p,C(p))
    \geq \avg(p,C(p)) - error(C(p))
    &\geq \avg(p,C(p)) - t^*/(100\alpha|C(p)|)\\
    &\geq \avg(p,C(p)) - \frac{16}{100\alpha}\avg(p,C(p))\\
    &\geq (84/100)\avg(p,C(p)),
\end{aligned}
\end{equation}
and
\begin{equation}\label{:eq:fast-epoch:proof-of-return-IP-stable-clustering:relating-old-avg-to-C'}
\begin{aligned}
    \widetilde{\avg}(p,C')
    \leq \avg(p,C') + error(C')
    \leq \avg(p,C') + t^*/(100\alpha|C'|)
    \leq \avg(p,C') + (16/(100\alpha))\avg(p,C(p))
\end{aligned}
\end{equation}
held while the algorithm exited the loop.
\cref{:eq:fast-epoch:proof-of-return-IP-stable-clustering:relating-old-avg-to-C'} and \cref{:eq:fast-epoch:proof-of-return-IP-stable-clustering:contradiction-assumption} together give us that
\begin{align*}
    \widetilde{\avg}(p,C')
    &< \frac{|C(p)|}{\alpha|C(p) \setminus \{p\}|}\avg(p,C(p)) + (16/(100\alpha))\avg(p,C(p)) \\
    &\leq (116/(100\alpha))\frac{|C(p)|}{|C(p) \setminus \{p\}|}\avg(p,C(p)).
\end{align*}
So, by \cref{:eq:fast-epoch:proof-of-return-IP-stable-clustering:relating-old-avg-to-C},
\[
 \widetilde{\avg}(p,C(p))
 \geq \frac{84}{116} \cdot \avg(p,C(p))
 > \frac{84}{116} \cdot \alpha\frac{|C(p) \setminus \{p\}|}{|C(p)|} \cdot \widetilde{\avg}(p,C').
\]
Therefore, since $\widehat{\avg}(p,C(p))$ and $\widehat{\avg}(p,C')$ are $(1+\epsilon)$-estimates of $\widetilde{\avg}(p,C(p))$ and $\widetilde{\avg}(p,C')$ (see \cref{:def:old-averages}), we get that
\[
 \widehat{\avg}(p,C(p)) > \frac{84}{116} \cdot \alpha \cdot \frac{|C(p) \setminus \{p\}|}{|C(p)|}\cdot \widetilde{\avg}(p,C')
\]
held while the algorithm exited the loop, which means that
\[
 \frac{|C(p)|}{\alpha|C(p) \setminus \{p\}|}\cdot \widehat{\avg}(p,C(p))
 > (\alpha/2)\widetilde{\avg}(p,C').
\]
Since $C(p) \neq \{p\}$ and $C' \in \cC \setminus \{C(p)\}$, and the above inequality implies that the condition of the while loop at line \ref{:line:fast-epoch:main-loop} of the algorithm was satisfied while the algorithm exited that loop, which is a contradiction.

To summarize, we assumed that the algorithm returned at line \ref{:line:fast-epoch:return-IP-stable-clustering} while there existed some point $p \in X$ and cluster $C' \in \cC$ such that $C(p) \neq \{p\}$ and $\avg(p,C(p)\setminus\{p\}) > \alpha \avg(p,C')$, and got a contradiction.
Therefore, if the algorithm returns at line \ref{:line:fast-epoch:return-IP-stable-clustering}, there must not exist such $p$ and $C'$, which exactly means that the returned clustering is $\alpha$-IP stable.
This concludes the proof of \cref{:cl:return-IP-stable-clustering}.
\end{proof}

\begin{claim}\label{:cl:return-clustering-with-less-potential}
If the algorithm returns at line \ref{:line:fast-epoch:return-clustering-with-less-potential}, then the returned cluster $\cC$ has $\Phi(\cC) < \frac{3}{4}\Phi(\cC_{\mathrm{input}})$, where $\cC_{\mathrm{input}}$ denote the clustering that the algorithm received as input.
\end{claim}

\subsubsection{Number of Swap Steps and Merge and Split Steps} \label{:sec:num-swap-and-merge-and-split-steps}
In this section, we bound the total number of swap steps, as well as the total number of merge and split step, that may occur during the execution of Algorithm \ref{:alg:fast-epoch}. For further details, refer to~\cref{:cor:number-of-swap-and-merge-and-split-steps} and~\cref{:def:step-types}.

\begin{claim} \label{:cl:swaps-are-valid}
At every swap step made by Algorithm \ref{:alg:fast-epoch}, when the execution is at line \ref{:line:fast-epoch:before-swap}, it must be that $\avg(p,C(p) \setminus \{p\}) > 4\log(n) \cdot \avg(p,C')$.
\end{claim}
\begin{proof}[Proof of \cref{:cl:swaps-are-valid}]
By the definition of line \ref{:line:fast-epoch:before-swap}, the chosen point $p$ and cluster $C'$ satisfy $C(p) \neq \{p\}$ and $(|C(p)|/|C(p)\setminus\{p\}|)\widehat{\avg}(p,C(p)) > (\alpha/2)\widehat{\avg}(p,C')$.
Since the values $\widehat{\avg}(p,C(p))$ and $\widehat{\avg}(p,C')$ are $(1+\epsilon)$-approximations of $\widetilde{\avg}(p,C(p))$ and $\widetilde{\avg}(p,C')$ (see \cref{:def:old-averages}), the aforementioned inequality implies that
\begin{equation}\label{:eq:fast-epoch:proof-of-swaps-are-valid:estimates-say-point-is-envious}
    (|C(p)|/|C(p)\setminus\{p\}|)\widetilde{\avg}(p,C(p)) > (\alpha/(2+2\epsilon))\widetilde{\avg}(p,C').
\end{equation}
Furthermore, when the execution is at line \ref{:line:fast-epoch:before-swap}, by \cref{:lem:steps-are-effective}, $\avg(p,C(p)) \geq \frac{t^*}{10\min\{|C(p)|,|C'|\}}$.
So, by invariants 4 and 2 of \cref{:cl:basic-invariants-of-fast-epoch},
\begin{equation}\label{:eq:fast-epoch:proof-of-swaps-are-valid:relating-old-avg-to-C}
\begin{aligned}
    \widetilde{\avg}(p,C(p))
    \leq \avg(p,C(p)) + error(C(p))
    &\leq \avg(p,C(p)) + t^*/(100\alpha|C(p)|) \\
    &\leq \avg(p,C(p)) + (1/(10\alpha))\avg(p,C(p))\\
    &\leq (11/10)\avg(p,C(p)),
\end{aligned}
\end{equation}
and
\begin{equation}\label{:eq:fast-epoch:proof-of-swaps-are-valid:relating-old-avg-to-C'}
\begin{aligned}
    \avg(p,C')
    \leq \widetilde{\avg}(p,C') + error(C')
    &\leq \widetilde{\avg}(p,C') + t^*/(100\alpha|C'|) \\
    &\leq \widetilde{\avg}(p,C') + (1/(10\alpha))\avg(p,C(p))\\
    &\leq \widetilde{\avg}(p,C') + (1/(10\alpha))\avg(p,C(p)\setminus\{p\}).
\end{aligned}
\end{equation}
By \cref{:eq:fast-epoch:proof-of-swaps-are-valid:estimates-say-point-is-envious},
\[
 (|C(p)\setminus\{p\}|/|C(p)|)\frac{\widetilde{\avg}(p,C')}{\widetilde{\avg}(p,C(p))} < \frac{2(1+\epsilon)}{\alpha},
\]
so, by \cref{:eq:fast-epoch:proof-of-swaps-are-valid:relating-old-avg-to-C},
\begin{align*}
    (|C(p)\setminus\{p\}|/|C(p)|)\frac{\widetilde{\avg}(p,C')}{\avg(p,C(p))}
    \leq (|C(p)\setminus\{p\}|/|C(p)|)\frac{\widetilde{\avg}(p,C')}{(10/11)\widetilde{\avg}(p,C(p))}
    < (11/10)\frac{2(1+\epsilon)}{\alpha},
\end{align*}
which implies that
\[
 \frac{\widetilde{\avg}(p,C')}{\avg(p,C(p)\setminus\{p\})} < (11/10)\frac{2(1+\epsilon)}{\alpha}.
\]
So, by \cref{:eq:fast-epoch:proof-of-swaps-are-valid:relating-old-avg-to-C'},
\begin{align*}
    \frac{\avg(p,C')}{\avg(p,C(p)\setminus\{p\})}
    \leq \frac{\widetilde{\avg}(p,C')}{\avg(p,C(p)\setminus\{p\})} + (1/(10\alpha))
    < (11/10)\frac{2(1+\epsilon)}{\alpha} + (1/(10\alpha))
    \leq 4/\alpha.
\end{align*}
Since $\alpha = 16\log n$, the above inequality exactly say that $\avg(p,C(p)\setminus\{p\}) > 4\log(n) \cdot \avg(p,C')$, as we needed.
This concludes the proof of \cref{:cl:swaps-are-valid}.
\end{proof}

\begin{corollary} \label{:cor:potential-reduction-by-swap-steps}
Each swap step reduces the potential $\Phi(\cC)$ by more than $\Tilde{\Omega}(\hat{\Phi}(\cC)/(nk))$. 
\end{corollary}
\begin{proof}
By \cref{:thm:existence-of-potential-function}, the potential function $\Phi$ satisfies that \ref{:line:fast-epoch:actual-swap} of the algorithm reduces the potential of the cluster $C$ by at least $\avg(p,C\setminus\{p\})$ and increases the potential of the cluster $C'$ by at most $2\log(n)\cdot\avg(p,C')$.
By \cref{:cl:swaps-are-valid}, this increase to the potential of $C'$ is less than $\frac{1}{2}\avg(p,C\setminus\{p\})$, so the total change to $\Phi(\cC)$ is a reduction by more than $\frac{1}{2}\avg(p,C\setminus\{p\})$.
Furthermore, by the fact that $\avg(p,C\setminus\{p\}) \geq \avg(p,C)$ and by \cref{:lem:steps-are-effective},
\[
 \frac{1}{2}\avg(p,C\setminus\{p\}) \geq \frac{1}{2}\avg(p,C) \geq \frac{t^*}{20\min\{|C(p)|,|C|\}}
 \geq \frac{t^*}{20n}.
\]
In summary, each swap step reduce the potential $\Phi(\cC)$ by at least $\frac{t^*}{20n}$, which is $\Tilde{\Omega}(\hat{\Phi}(\cC)/(nk))$ as we needed. (See the definition of $t^*$ at line \ref{:line:fast-epoch:definition-of-t-start} of the algorithm.)
\end{proof}

\begin{claim} \label{:cl:potential-reduction-by-merge-and-split-steps}
Each merge and split step reduces the potential of the clustering by more than $\Tilde{\Omega}(\hat{\Phi}(\cC)/k)$. I.e., $\Phi(\cC_{\mathrm{after}}) < \Phi(\cC_{\mathrm{before}}) - \Tilde{\Omega}(\hat{\Phi}(\cC_{\mathrm{before}})/k)$.
\end{claim}
\begin{proof}
Fix a merge and split step during the execution of the algorithm, which is defined as a step where the algorithm enters into the if statement at line \ref{:line:fast-epoch:merge-and-split-if-statement}. (See \cref{:def:step-types}.)
Let $\cC$ denote the start of the clustering at the beginning of the step, let $\cC'$ denote the state of the clustering after the algorithm executed line \ref{:line:fast-epoch:changing-clustering-in-merge} during this step, and let $\cC''$ denote the state of the clustering at the end of the step (i.e. after line \ref{:line:fast-epoch:changing-clustering-in-split}).
Our goal is to show that $\Phi(\cC'') < \Phi(\cC) - \Tilde{\Omega}(\hat{\Phi}(\cC)/k)$.
We begin by separately analyzing the merge and the split, in order to relate $\Phi(\cC')$ to $\Phi(\cC)$, and relate $\Phi(\cC'')$ to $\Phi(\cC')$,

\paragraph{Separately Analyzing the Merge and the Split.} Since this is a merge and split step, the condition of the if statement at line \ref{:line:fast-epoch:merge-and-split-if-statement} must be satisfied when the algorithm reaches it. 
So, when the algorithm reached this line, it must be that $\frac{\min\{|C|,|C'|\}}{|C'|}\sum_{p \in C'} \widehat{\avg}(p,C) < t^*$.
Furthermore, this happens just after the estimates $\{\widehat{\avg}(p,C)\}_{p \in C}$ are computed in line \ref{:line:fast-epoch:recompute-averages} of the algorithm.
So, since the estimates produced by the subroutine {\sc CalcAverages} can only err from above, the inequality $\frac{\min\{|C|,|C'|\}}{|C'|}\sum_{p \in C'} \avg(p,C) < t^*$ must hold when the algorithm reaches line \ref{:line:fast-epoch:merge-and-split-if-statement} in this step.
So,
\begin{align*}
    t^*
    > \frac{\min\{|C|,|C'|\}}{|C'|}\sum_{p \in C'} \avg(p,C)
    &= \frac{\min\{|C|,|C'|\}}{|C'||C|}\sum_{p_1 \in C'} \sum_{p_2 \in C} \avg(p_1,p_2) \\
    &= \frac{1}{\max\{|C|,|C'|\}}\sum_{p_1 \in C'} \sum_{p_2 \in C} \avg(p_1,p_2)
\end{align*}
must hold at that time.
So, by \cref{:lem:merge-cost-bound},
\begin{equation}\label{:eq:resulting-merge-potential-change-bound}
    \Phi(\cC') < \Phi(\cC) + \log(n) \cdot t^*.
\end{equation}
Furthermore, by the guarantee of the subroutine $\fastSplit$ (see \cref{:lem:fast-split}),
\begin{equation}\label{:eq:analysis-of-split-reduction}
    \Phi(\cC'') \leq \Phi(\cC') - \Omega(\frac{\Phi(\cC')}{k \log (n)}).
\end{equation}

\paragraph{Putting it Together.}
Since the algorithm did not return in line \ref{:line:fast-epoch:return-clustering-with-less-potential} during this iteration, it must be that $\CalcPotential(\cC,\epsilon)$ returned a value which is greater or equal to $((1+\epsilon)/2)\hat{\Phi}(\cC)$.
Therefore, it must be that $(1+\epsilon)\Phi(\cC) \geq ((1+\epsilon)/2)\hat{\Phi}(\cC)$ (see \cref{:cor:fast-potential-of-clustering}), which means that
\begin{equation}\label{:eq:relating-hat-Phi-C-to-Phi-C}
    \Phi(\cC) \geq \frac{\hat{\Phi}(\cC)}{2}.
\end{equation}
Now, if $\Phi(\cC') < \Phi(\cC)-\hat{\Phi}(\cC)/100$, then by \cref{:eq:analysis-of-split-reduction} we also have $\Phi(\cC'') < \Phi(\cC)-\hat{\Phi}(\cC)/100$, in which case we are done.
So, for the rest of the analysis, we assume that $\Phi(\cC') \geq \Phi(\cC) - \hat{\Phi}(\cC)/100$.
By this assumption and by, \cref{:eq:relating-hat-Phi-C-to-Phi-C}, $\Phi(\cC') \geq \Phi(\cC) - \Phi(\cC)/50 = (49/50)\Phi(\cC)$, which, by \cref{:eq:relating-hat-Phi-C-to-Phi-C}, means
\[
 \Phi(\cC') \geq \frac{49}{100} \cdot \hat{\Phi}(\cC).
\]
Therefore, by \cref{:eq:analysis-of-split-reduction},
\[
 \Phi(\cC'') \leq \Phi(\cC') - \frac{49}{100} \cdot \Omega(\frac{\hat{\Phi}(\cC)}{k \log (n)}),
\]
so, by \cref{:eq:resulting-merge-potential-change-bound},
\[
 \Phi(\cC'') < \Phi(\cC) + \log(n) \cdot t^* - \frac{49}{100} \cdot \Omega(\frac{\Phi(\cC)}{k \log (n)}).
\]
Therefore, if we set the constant in the $\Omega$ notation of the definition of $t^*$ to the same as the one in the $\Omega$ notation of the above inequality (which is the same as the one in the guarantee from \cref{:lem:fast-split}), then we get that
\[
 \Phi(\cC'') < \Phi(\cC) - \frac{24}{100} \cdot \Omega(\frac{\Phi(\cC)}{k \log (n)}),
\]
which is what we needed to prove.
This concludes the proof of \cref{:cl:potential-reduction-by-merge-and-split-steps}.
\end{proof}

\begin{corollary} \label{:cor:number-of-swap-and-merge-and-split-steps}
There are at most $\Tilde{O}(nk)$ swap steps, and at most $\Tilde{O}(k)$ merge-and-split steps.
\end{corollary}
\begin{proof}[Proof of \cref{:cor:number-of-swap-and-merge-and-split-steps}]
At the beginning of the algorithm, it set $\hat{\Phi}(\cC)$ to a value which is at least as large as the initial value of $\Phi(\cC)$. (See line \ref{:line:fast-epoch:computing-initial-estimated-potential} and \cref{:cor:fast-potential-of-clustering}.)
Furthermore, since the only types of steps which affect the clustering are swap steps and merge and split steps (see \cref{:def:step-types}), \cref{:cor:potential-reduction-by-swap-steps} and \cref{:cl:potential-reduction-by-merge-and-split-steps} imply that the potential of the maintained clustering $\cC$ never decreases.
So, since the potential of clustering can never be negative (see \cref{:sec:potential-function}), \cref{:cor:potential-reduction-by-swap-steps} and \cref{:cl:potential-reduction-by-merge-and-split-steps} imply that there can be at most $\Tilde{O}(nk)$ swap steps, and at most $\Tilde{O}(k)$ merge-and-split steps, as we needed to prove.
\end{proof}

\subsubsection{Number of Recompute Steps} \label{:sec:num-compute-steps}
In this section, we bound the total number of recompute steps that may occur during the execution of Algorithm \ref{:alg:fast-epoch}. (See \cref{:cor:number-of-recompute-steps} and \cref{:def:step-types}.)

\begin{claim} \label{:cl:progress-increase-relative-to-error-increase}
whenever lines \ref{:line:fast-epoch:increase-error}-\ref{:line:fast-epoch:increase-numSwaps} are executed, if $x$ is the amount by which line \ref{:line:fast-epoch:increase-progress} increases $progress(C)$, then the amount by which line \ref{:line:fast-epoch:increase-error} increases $error(C)$ is at most $O(x/\widehat{size}(C''))$, and the amount by which line \ref{:line:fast-epoch:increase-numSwaps} increases $numSwaps(C'')$ is at most $O(x \cdot \widehat{size}(C'')/t^*)$.
\end{claim}
\begin{proof}
Firstly, when the algorithm is at line \ref{:line:fast-epoch:before-swap}, by \cref{:lem:steps-are-effective}, $\avg(p,C(p)) \geq \frac{t^*}{10\min\{|C(p)|,|C'|\}}$.
So, by invariants 4 and 2 of \cref{:cl:basic-invariants-of-fast-epoch},
\begin{align*}
    \widetilde{\avg}(p,C(p))
    \geq \frac{t^*}{10\min\{|C(p)|,|C'|\}} - error(C(p))
    &\geq \frac{t^*}{10\min\{|C(p)|,|C'|\}} - t^*/(100\alpha|C(p)|)\\
    &\geq \frac{t^*}{11\min\{|C(p)|,|C'|\}}.
\end{align*}
Therefore, when the algorithm is at line \ref{:line:fast-epoch:before-swap}, by the fact that $\widehat{\avg}(p,C(p)) \geq \widetilde{\avg}(p,C(p))$ (see \cref{:def:old-averages}),
\[
 \widehat{\avg}(p,C(p)) \geq \frac{t^*}{11\min\{|C(p)|,|C'|\}}
\]
which, by invariant 2 of \cref{:cl:basic-invariants-of-fast-epoch}, it implies that
\begin{equation}\label{:eq:proof-of-progress-relative-to-error:relating-error-to-larger-distance}
    \forall C'' \in \{C(p),C'\}, \qquad error(C'') \leq \widehat{\avg}(p,C(p))/(9\alpha)
\end{equation}
and by \cref{:cor:size-didnt-change-much}, it implies that
\begin{equation}\label{:eq:proof-of-progress-relative-to-error:lower-bound-on-prgoress-increase}
    \forall C'' \in \{C(p),C'\}, \qquad 1 \leq 11\widehat{\avg}(p,C(p)) \cdot \widehat{size}(C'')/t^*
\end{equation}
\cref{:eq:proof-of-progress-relative-to-error:relating-error-to-larger-distance} implies that when the algorithm is at line~\ref{:line:fast-epoch:before-swap}, we have
\[
 \frac{1}{1+\epsilon}\cdot \widehat{\avg}(p,C(p))- error(C(p))
 \geq \frac{1}{2}\cdot \widehat{\avg}(p,C(p))
\]
which implies that after line \ref{:line:fast-epoch:set-progress-increase},
\begin{equation}\label{:eq:proof-of-progress-relative-to-error:progress-increase}
    progressIncrease \geq \frac{1}{4}\widehat{\avg}(p,C(p)).
\end{equation}

By \cref{:obs:swap-decided-by-estimates}, when the algorithm is at line \ref{:line:fast-epoch:before-swap}, $\widetilde{\avg}(p,C(p))\geq 2\widetilde{\avg}(p,C')$.
So, since $\widehat{\avg}(p,C(p))$ and $\widehat{\avg}(p,C')$ are $(1+\epsilon)$-approximations of $\widetilde{\avg}(p,C(p))$ and $\widetilde{\avg}(p,C')$ (see \cref{:def:old-averages}), we get that $\widehat{\avg}(p,C(p))\geq 2\widehat{\avg}(p,C')$.
Therefore, at this point, for each $C'' \in \{C(p),C'\}$, by also using invariant 2 of \cref{:cl:basic-invariants-of-fast-epoch},
\[
 (\widehat{\avg}(p,C'') + error(C''))/\widehat{size}(C'') \leq (\widehat{\avg}(p,C) + error(C''))/\widehat{size}(C'')
\]
So, by \cref{:eq:proof-of-progress-relative-to-error:relating-error-to-larger-distance},
\begin{align*}
    (\widehat{\avg}(p,C'') + error(C''))/\widehat{size}(C'')
    \leq (\widehat{\avg}(p,C(p)) + error(C''))/\widehat{size}(C'')
    \leq \frac{10}{9}\widehat{\avg}(p,C(p))/\widehat{size}(C'').
\end{align*}
Therefore, by \cref{:cor:size-didnt-change-much}, when the algorithm increases $error(C'')$ at line \ref{:line:fast-epoch:increase-error}, it increases it by at most $\frac{20}{9}\widehat{\avg}(p,C(p))/\widehat{size}(C'')$.
Furthermore, by \cref{:eq:proof-of-progress-relative-to-error:progress-increase}, line \ref{:line:fast-epoch:increase-progress} increases $error(C'')$ by at least $\frac{1}{4}\widehat{\avg}(p,C(p))$.

In summary, line \ref{:line:fast-epoch:increase-error} increases $error(C'')$ by at most $\frac{20}{9}\widehat{\avg}(p,C(p))/\widehat{size}(C'')$, while line \ref{:line:fast-epoch:increase-progress} increases $progress(C'')$ by at least $\frac{1}{4}\widehat{\avg}(p,C(p))$, and line \ref{:line:fast-epoch:increase-numSwaps} increases $numSwaps(C'')$ by at most $11\widehat{\avg}(p,C(p)) \cdot \widehat{\avg}(p,C(p))/t^*$ (\cref{:eq:proof-of-progress-relative-to-error:lower-bound-on-prgoress-increase}).
This concludes the proof of \cref{:cl:progress-increase-relative-to-error-increase}.
\end{proof}

\begin{claim}\label{:cl:progress-relative-to-error}
During line \ref{:line:fast-epoch:if-statement-error-too-large}, for every cluster $C \in \cC$, $progress(C) = \Omega(error(C) \cdot \widehat{size}(C))$ and also $progress(C) = \Omega(numSwaps(C) \cdot t^* / \widehat{size}(C))$.
\end{claim}
\begin{proof}
Let $\tau$ be the current iteration.
Since we reached line \ref{:line:fast-epoch:if-statement-error-too-large} in this iteration, it must have been a swap step.
By the condition of the if statement in line \ref{:line:fast-epoch:swap-or-recompute-if-statement}, $\clusterToRecompute$ must have been empty at the beginning of the iteration, which means that $C \in \cC \setminus \clusterToRecompute$ held.
Therefore, by \cref{:cl:values-defined}, there was some recompute step on the cluster $C$ before step $\tau$.
Let $\tau'$ be the last such recompute step.
At step $\tau'$, lines \ref{:line:fast-epoch:reset-error} and \ref{:line:fast-epoch:reset-progress} set $error(C)$ and $progress(C)$ and $numSwaps(C)$ to $0$.
After that, until the end of iteration $\tau$, the only times were $error(C)$, $progress(C)$, and $numSwaps(C)$, were changed are by lines \ref{:line:fast-epoch:increase-error}, \ref{:line:fast-epoch:increase-progress}, and \ref{:line:fast-epoch:increase-numSwaps} during swap steps.
So, by \cref{:cl:progress-increase-relative-to-error-increase}, we get that $error(C) = O(progress(C) / \widehat{size}(C))$ and $numSwaps(C) = O(progress(C) \cdot \widehat{size}(C) / t^*)$ hold at the beginning of iteration $\tau$, which exactly means that $progress(C) = \Omega(error(C) \cdot \widehat{size}(C))$ and $progress(C) = \Omega(numSwaps(C) \cdot t^* / \widehat{size}(C))$, as we needed.
\end{proof}

\begin{corollary} \label{:cor:recompute-only-once-large-progress}
Whenever a cluster $C$ is inserted into $\clusterToRecompute$ by line \ref{:line:fast-epoch:recompute-because-error-too-large}, it must have $progress(C) = \Tilde{\Omega}(t^*) = \Tilde{\Omega}(\hat{\Phi}(\cC)/k)$.
\end{corollary}
\begin{proof}
This follows directly from \cref{:cl:progress-relative-to-error} and from the condition of the if statement in line \ref{:line:fast-epoch:if-statement-error-too-large}.
\end{proof}

\begin{claim} \label{:cl:progress-implies-reduction-in-potential}
For every cluster $C''$, whenever $progress(C'')$ is increased by a swap step, this swap step must have reduced $\Phi(\cC)$ by a greater amount.
\end{claim}
\begin{proof}
Let $\Delta$ be the amount by which the current swap step reduces $\Phi(\cC)$.
By \cref{:thm:existence-of-potential-function}, the potential function $\Phi$ satisfies that line \ref{:line:fast-epoch:actual-swap} of the algorithm reduces the potential of the cluster $C(p)$ by at least $\avg(p,C(p)\setminus\{p\})$ and increases the potential of the cluster $C'$ by at most $2\log(n)\cdot\avg(p,C')$.
By \cref{:cl:swaps-are-valid}, this increase to the potential of $C'$ is less than $\frac{1}{2}\avg(p,C(p)\setminus\{p\})$, so the total change to $\Phi(\cC)$ is a reduction by more than $\frac{1}{2}\avg(p,C(p)\setminus\{p\})$.
Therefore, since $\avg(p,C(p)\setminus\{p\}) \geq \avg(p,C(p))$, we get that 
\begin{equation}\label{:eq:reduction-in-potential-is-more-than-half-of-avarage-distance}
    \Delta > \frac{1}{2}\avg(p,C(p)).
\end{equation}

At the beginning of the current swap step, by the fact that $\widehat{\avg}(p,C(p))$ is a $(1+\epsilon)$ approximation of $\widetilde{\avg}(p,C(p))$ (see \cref{:def:old-averages}), together with invariant 4 of \cref{:cl:basic-invariants-of-fast-epoch}, and with \cref{:eq:reduction-in-potential-is-more-than-half-of-avarage-distance},
\begin{align*}
 (\frac{1}{1+\epsilon}\widehat{\avg}(p,C(p)) - error(C(p)))/2
 \leq (\widetilde{\avg}(p,C(p)) - error(C(p)))/2
 \leq (\avg(p,C(p)))/2
 < \Delta.
\end{align*}
Therefore, line \ref{:line:fast-epoch:set-progress-increase} set the variable $ProgressIncrease$ to less than $\Delta$, which means that for each $C''$, the amount by which the current swap step increases $progress(C'')$ is less than $x$ (i.e. less than the amount by which the current swap step reduces $\Phi(\cC)$), as we needed to prove.
\end{proof}

\begin{corollary} \label{:cor:estimates-become-bad-rarely}
Throughout the execution of Algorithm \ref{:alg:fast-epoch}, line \ref{:line:fast-epoch:recompute-because-error-too-large} is executed at most $\Tilde{O}(k)$ times.
\end{corollary}
\begin{proof}
By \cref{:cor:recompute-only-once-large-progress}, whenever line \ref{:line:fast-epoch:recompute-because-error-too-large} is executed on a cluster $C''$, this cluster must have $progress(C'') \geq \Tilde{\Omega}(\hat{\Phi}(\cC)/k)$.
Furthermore, by \cref{:cl:progress-implies-reduction-in-potential}, this implies that since the last time a recompute step happened on cluster $C''$, the cluster $C''$ has been involved in swap steps that reduced $\Phi(\cC)$ by a total greater than $\Tilde{\Omega}(\hat{\Phi}(\cC)/k)$.
Since a recompute step on cluster $C''$ has to happen between any two consecutive times that line \ref{:line:fast-epoch:recompute-because-error-too-large} is executed on a cluster $C''$, this implies that, for any one cluster $C$, if line \ref{:line:fast-epoch:recompute-because-error-too-large} has been executed on this cluster $C$ a total of $x$ times throughout the algorithm, then $C$ has been involved in swap steps that reduced $\Phi(\cC)$ by a total greater than $\Tilde{\Omega}(x\hat{\Phi}(\cC)/k)$. So, since each swap step can involve at most two clusters, we get that the total reduction to $\Phi(\cC)$ throughout the algorithm is greater than $\left(\sum_{C}x_C\right)\cdot\Tilde{\Omega}(\hat{\Phi}(\cC)/k)/2$, where $\left(\sum_{C}x_C\right)$ denotes the total number of times that line \ref{:line:fast-epoch:recompute-because-error-too-large} has been executed throughout the algorithm.
Therefore, since the value of $\hat{\Phi}(\cC)$ is at least as large as the initial value of $\Phi(\cC)$ (see line \ref{:line:fast-epoch:computing-initial-estimated-potential} and \cref{:cor:fast-potential-of-clustering}), we get that the total number of times that line \ref{:line:fast-epoch:recompute-because-error-too-large} has been executed throughout the algorithm is at most $\hat{\Phi}(\cC)/\Tilde{\Omega}(\hat{\Phi}(\cC)/k) = \Tilde{O}(k)$.
\end{proof}

\begin{corollary}\label{:cor:number-of-recompute-steps}
Throughout the execution of Algorithm \ref{:alg:fast-epoch}, there are at most $\Tilde{O}(k)$ recompute steps.
\end{corollary}
\begin{proof}
Each recompute step removes one cluster from $\clusterToRecompute$ at line \ref{:line:fast-epoch:remove-cluster-from-clusters-to-recompute}.
Furthermore, $\clusterToRecompute$ is created with $k$ clusters, and the only lines that insert clusters into $\clusterToRecompute$ are line \ref{:line:fast-epoch:recompute-because-error-too-large}, line \ref{:line:fast-epoch:changing-clusterstorecompute-in-merge}, and line \ref{:line:fast-epoch:changing-clusterstorecompute-in-split}, which each insert a constant number of clusters each time they are executed.
So, since \cref{:cor:estimates-become-bad-rarely} says that line \ref{:line:fast-epoch:recompute-because-error-too-large} is executed at most $\Tilde{O}(k)$ times throughout the execution of the algorithm, and \cref{:cor:number-of-swap-and-merge-and-split-steps} implies that lines \ref{:line:fast-epoch:changing-clusterstorecompute-in-merge} and \ref{:line:fast-epoch:changing-clusterstorecompute-in-split} are executed at most $\Tilde{O}(k)$ times throughout the execution of the algorithm, we get that the total number of recompute step throughout the execution of the algorithm can be at most $k + O(1)\cdot\Tilde{O}(k) = \Tilde{O}(k)$.
\end{proof}

\subsubsection{Running Time Analysis of Algorithm \ref{:alg:fast-epoch}} \label{:sec:runtime-analysis-of-fast-epoch}
In this section, we bound the total running time of Algorithm \ref{:alg:fast-epoch}.
We begin by proving the following two claims, and then complete the analysis in \cref{:cor:final-proof-of-running-time-of-fast-epoch}.

\begin{claim}\label{:cl:running-time-of-swap-step}
Each swap step can be implemented in time $\Tilde{O}(1)$.
\end{claim}
\begin{proof}
To quickly implement the checking of the condition of the loop in line \ref{:line:fast-epoch:main-loop}, and to implement line \ref{:line:fast-epoch:choose-point-to-swap}, the algorithm maintains, for each $p \in X$, a min-heap $\mathrm{Heap}_{p}$ that contains the values of $\widehat{avg}(p,C')$ for all clusters $C' \in \cC \setminus \{C(p)\}$, where each value is accompanied by a pointer to the respective cluster $C'$.
Furthermore, the algorithm maintains a min-heap $\mathrm{MainHeap}$ that contains $\frac{\min_{C' \in \cC \setminus \{C(p)\}} \widehat{\avg}(p,C') }{(|C(p)|/|C(p)\setminus\{p\}|)\widehat{\avg}(p,C(p))}$ for each $p \in X$ such that $C(p) \neq \{p\}$ and $\widehat{\avg}(p,C(p)) \neq 0$, where each value is accompanied by a pointer to the respective point $p$.
These min-heaps while only increasing the running time of the algorithm by a multiplicative $\Tilde{O}(1)$, where the implementation of the updates to $\mathrm{MainHeap}$ uses the other heaps to find the appropriate $\min_{C' \in \cC \setminus \{C(p)\}} \widehat{\avg}(p,C')$ for each point.
Furthermore, using these heaps, it's clear that checking of the condition of the loop in line \ref{:line:fast-epoch:main-loop}, and implementing line \ref{:line:fast-epoch:choose-point-to-swap}, can both be done in time $\Tilde{O}(1)$.
Since the rest of the lines executed during a swap step can all be trivially executed in time $\Tilde{O}(1)$, we get that a swap step runs in time $\Tilde{O}(1)$, as we needed to prove.
\end{proof}

\begin{claim}\label{:cl:running-time-of-recompute-step}
Each recompute step (including merge and split steps) can be implemented in time $\Tilde{O}(n)$
\end{claim}
\begin{proof}
It is already explained in the proof of \cref{:cl:running-time-of-swap-step} how the condition of the while loop in line \ref{:line:fast-epoch:main-loop} can checked in time $\Tilde{O}(1)$.
Furthermore, by \cref{:lem:fast-average-distances}, \cref{:cor:fast-potential-of-clustering}, and \cref{:lem:fast-split}, the procedures $\CalcAverage$, $\CalcPotential$, and $\fastSplit$ each run in time $\Tilde{O}(n)$.
This implies that each recompute step runs in time $\Tilde{O}(n)$, as we needed to prove.
\end{proof}

\begin{corollary}\label{:cor:final-proof-of-running-time-of-fast-epoch}
The total running time of Algorithm \ref{:alg:fast-epoch} is $\Tilde{O}(nk)$.
\end{corollary}
\begin{proof}
\cref{:cl:running-time-of-swap-step} and \cref{:cor:number-of-swap-and-merge-and-split-steps} imply that the total time spent on swap steps is $\Tilde{O}(nk)$, while \cref{:cl:running-time-of-recompute-step} and \cref{:cor:number-of-recompute-steps} imply that the total time spent on recompute steps is $\Tilde{O}(nk)$.
Therefore, the total time that the algorithm spends in the while loop at line \ref{:line:fast-epoch:main-loop} is $\Tilde{O}(nk)$.
Furthermore, since $\CalcPotential$ runs in time $\Tilde{O}(n)$ (see \cref{:cor:fast-potential-of-clustering}, the total amount of time that the algorithm spends outside of this loop is $\Tilde{O}(n)$.
This concludes the proof of \cref{:cor:final-proof-of-running-time-of-fast-epoch}.
\end{proof}

\subsection{Approximating Average Distances Fast}\label{:sec:fast-average-distances}
The goal of this section is to prove \cref{:lem:fast-average-distances}.
To do this, we will need the following claim.
\begin{claim}\label{:cl:find-central-point}
There exists an algorithm $\CalcPotential$ that, given a metric space $(X,d)$ a non-empty cluster $C \subseteq X$, and $\delta > 0$, computes a point $p \in C$ such that $\avg(p,C) \leq \frac{2}{|C|}\sum_{p' \in C} \avg(p',C)$ holds with probability $(1-\delta)$, in time $O(|C|\log(\frac{1}{\delta}))$.
\end{claim}
\begin{proof}
The algorithm samples $t=\ceil{\log(\frac{1}{\delta})}{}$ points $p_1,\ldots,p_t$ from $C$ independently uniformly at random, explicitly computes $\avg(p_i,C)$ for each such point, and outputs the point with the lowest $\avg(p_i,C)$.
Since computing each $\avg(p_i,C)$ takes only $O(|C|)$ time, the algorithm runs in time $O(|C|\log(\frac{1}{\delta}))$.
So, we just need to analyze the failure probability of the algorithm.

\paragraph{Correctness Analysis.}
We need to show that the probability that all $p_i$ simultaneously satisfy $\avg(p_i,C) > \frac{2}{|C|}\sum_{p' \in C} \avg(p',C)$ is at most $\delta$.
Since these points are sampled independently, it is enough to show that each point $p_i$ satisfies $\avg(p_i,C) > \frac{2}{|C|}\sum_{p' \in C} \avg(p',C)$ with probability at most $1/2$.
Indeed, since $\expect{\avg(p_i,C)} =\frac{1}{|C|}\sum_{p' \in C} \avg(p',C)$, where $p_i$ is chosen uniformly from $C$, this probability bound is implied by Markov's inequality.

This concludes the proof of \cref{:cl:find-central-point}.
\end{proof}

\subsubsection{Algorithm and Analysis}
To prove \cref{:lem:fast-average-distances}, we need to present an algorithm (Algorithm \ref{:alg:fast-average-distances}) that, given a metric space $(X,d)$, a non-empty cluster $C \subseteq X$, a set $S \subseteq X$, and $\epsilon > 0$, runs in time $\Tilde{O}(\frac{|C|+|S|}{\epsilon^2})$ and, for each $p \in S$, computes w.h.p a $(1+\epsilon)$-estimate of $\avg(p,C)$.

\begin{algorithm}
\caption{$\CalcAverage$ from \cref{:lem:fast-average-distances}.}\label{:alg:fast-average-distances}
\KwData{$(X,d)$, $S \subseteq X$, non-empty $C \subseteq X$, $0 < \epsilon \leq 1$}
\KwResult{For each $p \in S$, as estimate $\widehat{\avg}(p,C)$ of $\avg(p,C)$}
$p^* \leftarrow \CalcCentral(C)$, $\epsilon' \leftarrow \epsilon/3$, and $t \leftarrow O(1/(\epsilon')^2)$\label{:line:fast-averages:definition-of-t}\\
\If{$d(p^*,p)=0$ for all $p \in C$}{
    \For{$p' \in S$}{
        $\widehat{\avg}(p',C) \leftarrow d(p',p^*)$
    }
    \Return $\{\widehat{\avg}(p',C)\}_{p' \in S}$
}
\For{$i=1$ to $t$}{
    $p_i^{\mathrm{weighted}} \leftarrow$ sample a point $p$ from $C$ with probability proportional to $d(p,p^*)$\\
    $p_i^{\mathrm{uniform}} \leftarrow$ sample a point $p$ uniformly from $C$
}
{\bf compute} $\avg(p^*,C)$\\
\For{$p' \in S$\label{:line:fast-averages:loop-over-points-in-S-to-set-final-estimates}}{
    \For{$i=1$ to $t$}{
        $p_i \leftarrow$ sample $p_i^{\mathrm{weighted}}$ w.p. $\frac{\avg(p^*,C)}{\avg(p^*,C) + d(p',p^*)}$ and $p_i^{\mathrm{uniform}}$ w.p. $\frac{d(p',p^*)}{\avg(p^*,C) + d(p',p^*)}$\label{:line:fast-averages:final-sample}
    }
    $\widehat{\avg}(p',C) \leftarrow \frac{\avg(p^*,C) + d(p',p^*)}{t(1-\epsilon')}\sum_{i=1}^{t} d(p_i,p')/(d(p_i,p^*)+d(p^*,p'))$
}
\Return $\{\widehat{\avg}(p',C)\}_{p' \in S}$
\end{algorithm}

We will begin with the following claim about Algorithm \ref{:alg:fast-average-distances}.

\begin{claim}\label{:cl:distribution-of-samples}
For every $p' \in S$, every $i \in \{1,\ldots,t\}$, and every $p \in C$, the probability that line \ref{:line:fast-averages:final-sample} of the algorithm will set $p_i \leftarrow p$ is exactly $\frac{d(p,p^*)+d(p',p^*)}{\sum_{p'' \in C} (d(p'',p^*)+d(p',p^*))}$.
\end{claim}
\begin{proof}
In the iteration of the loop in line \ref{:line:fast-averages:loop-over-points-in-S-to-set-final-estimates} that selects $p'$, the probability that $p_i$ will be set to $p$ is exactly
\begin{align*}
    &\frac{\avg(p^*,C)}{\avg(p^*,C) + d(p',p^*)} \cdot \Pr[p_i^{\mathrm{weighted}} = p] + \frac{d(p',p^*)}{\avg(p^*,C) + d(p',p^*)} \cdot \Pr[p_i^{\mathrm{uniform}} = p]\\
    = &\frac{\avg(p^*,C)}{\avg(p^*,C) + d(p',p^*)} \cdot \frac{d(p,p^*)}{\sum_{p'' \in C}d(p'',p^*)} + \frac{d(p',p^*)}{\avg(p^*,C) + d(p',p^*)} \cdot \frac{1}{|C|}\\
    = &\frac{\avg(p^*,C)}{\avg(p^*,C) + d(p',p^*)} \cdot \frac{d(p,p^*)}{|C|\avg(p^*,C)} + \frac{d(p',p^*)}{\avg(p^*,C) + d(p',p^*)} \cdot \frac{1}{|C|}\\
    = &\frac{d(p,p^*)}{|C|(\avg(p^*,C) + d(p',p^*))} + \frac{d(p',p^*)}{|C|(\avg(p^*,C) + d(p',p^*))}\\
    = &\frac{d(p,p^*) + d(p',p^*)}{\sum_{p'' \in C} (d(p'',p^*)+d(p',p^*))},
\end{align*}
as needed to prove the claim.
\end{proof}

We are now ready to prove \cref{:lem:fast-average-distances}
\begin{proof}[Proof of \cref{:lem:fast-average-distances}]
It is straightforward to verify that Algorithm \ref{:alg:fast-average-distances} runs in $O(\frac{|C|+|S|}{\epsilon^2})$ time.
So, let $p'$ be an arbitrary point in $S$, and our goal for the rest of the proof is to show that, w.h.p,
\begin{equation}\label{:eq:goal-for-proof-of-fast-average-distances}
    \avg(p,C) \leq \widehat{\avg}(p,C) \leq (1+\epsilon)\avg(p,C).
\end{equation}

By \cref{:cl:find-central-point}, $\avg(p^*,C) \leq \frac{2}{|C|}\sum_{p'' \in C}\avg(p'',C)$. So, by \cref{:lem:average-distance-lemma},
\begin{equation}\label{:eq:proof-of-fast-average-distances:lower-bound-on-average-distance}
    \avg(p^*,C) \leq 4 \cdot \avg(p',C).
\end{equation}
For each $p \in C$, let $x_p = d(p,p')$ and let $\hat{x}_p = d(p,p^*)+d(p',p^*)$.
By the triangle inequality,
\begin{equation}\label{:eq:proof-of-fast-average-distances:bound-on-each-x_p}
    \forall p \in C, \qquad \hat{x}_p \geq x_p
\end{equation}
Furthermore, by the triangle inequality,
\begin{align*}
    \sum_{p \in C} \hat{x}_p
    = \sum_{p \in C} (d(p,p^*)+d(p',p^*))
    \leq \sum_{p \in C} (2d(p,p^*)+d(p',p))
    = 2|C|\avg(p^*,C) + |C|\avg(p',C)
\end{align*}
So, by \cref{:eq:proof-of-fast-average-distances:lower-bound-on-average-distance},
\begin{equation}\label{:eq:proof-of-fast-average-distances:bound-on-sum-of-hat-x_p}
\begin{aligned}
    \sum_{p \in C} \hat{x}_p
    \leq 2|C|\avg(p^*,C) + |C|\avg(p',C)
    \leq 9|C|\avg(p',C)
    = 9\sum_{p \in C}x_p
    = O\left(\sum_{p \in C}x_p\right).
\end{aligned}
\end{equation}
It's not difficult to see that, for our given $p'$, the distributions of the points $p_1,\ldots,p_t$ selected by line \ref{:line:fast-averages:final-sample} are independent.
Furthermore, by \cref{:cl:distribution-of-samples}, each $p_i$ is sampled with probability $\frac{\hat{x}_p}{\sum_{p''\in C} \hat{x}_{p''}}$ to be each $p$.
So, by \cref{:thm:importance-sampling}, since \cref{:eq:proof-of-fast-average-distances:bound-on-each-x_p} and \cref{:eq:proof-of-fast-average-distances:bound-on-sum-of-hat-x_p} hold, if we correctly set the constant in the $O$ notation at line \ref{:line:fast-averages:definition-of-t} of the algorithm, then, w.h.p,
\[
 (1-\epsilon')\sum_{p \in C} x_p \leq \frac{\left(\sum_{p \in C}\hat{x}_p\right)}{t} \cdot \sum_{j=1}^{t}\frac{x_{p_j}}{\hat{x}_{p_j}} \leq (1+\epsilon') \sum_{p \in C} x_p.
\]
Since $\epsilon' = \epsilon/3$ and $\epsilon \leq 1$, we have $\frac{1+\epsilon'}{1-\epsilon'}\leq(1+\epsilon)$, so we get that, w.h.p,
\[
 \sum_{p \in C} x_p \leq \frac{\left(\sum_{p \in C}\hat{x}_p\right)}{t(1-\epsilon')} \cdot \sum_{j=1}^{t}\frac{x_{p_j}}{\hat{x}_{p_j}} \leq (1+\epsilon) \sum_{p \in C} x_p.
\]
So, by the definitions of the $x_p$s and $\hat{x}_p$s, w.h.p,
\[
 \sum_{p \in C} d(p,p') \leq \frac{\sum_{p \in C}(d(p,p^*) + d(p',p^*))}{t(1-\epsilon')} \cdot \sum_{j=1}^{t}  \frac{d(p_j,p')}{d(p_j,p^*) + d(p',p^*)} \leq (1+\epsilon) \sum_{p \in C} d(p,p').
\]
Dividing all sides of the last inequality by $|C|$, we get that, w.h.p,
\[
 \avg(p',C) \leq \frac{\avg(p^*,C) + d(p',p^*)}{t(1-\epsilon')} \cdot \sum_{j=1}^{t}\frac{d(p_j,p')}{d(p_j,p^*) + d(p',p^*)} \leq (1+\epsilon) \avg(p',C),
\]
which implies that
\[
 \avg(p',C) \leq \widehat{\avg}(p',C) \leq (1+\epsilon) \avg(p',C).
\]

To summarize, we saw that Algorithm \ref{:alg:fast-average-distances} runs in time $\Tilde{O}(\frac{|C|+|S|}{\epsilon^2})$ and we saw that for every $p' \in S$, w.h.p, the returned estimate $\widehat{\avg}(p',C)$ satisfies $\avg(p',C) \leq \widehat{\avg}(p',C) \leq (1+\epsilon) \avg(p',C)$.
\end{proof}

\section{Proof of \texorpdfstring{\cref{:thm:when-there-is-a-very-good-clustering}}{Theorem 3}}
In this section, we show that if the input has a ground truth clustering $\cC$ satisfying $\alpha$-IP stability where $\alpha < 1/1000$, then it is possible to efficiently find an $O(\alpha)$-IP stable clustering.
First, we define the following quantity $\beta(\cC)$ for a clustering $\cC$, and show how to relate this quantity to the IP stability of $\cC$. 

\begin{definition}\label{:def:definition-of-beta}
Let $(X,d)$ be a metric space and let $C \subset X$ be a non-empty cluster.
Then, we let $\beta(C)$ denote the minimum value of $\beta$ that satsifies $\max_{p,p' \in C} d(p,p') \leq \beta \min_{p \in C, p' \in X \setminus C} d(p,p')$.
Furthermore, if $C = X$ then we let $\beta(C) \defeq 0$.
Finally, for a clustering $\cC$ of $X$ or of any subset of $X$, we let $\beta(\cC) \defeq \max_{C \in \cC} \beta(C)$.
\end{definition}

\begin{observation}\label{:obs:beta-is-upper-bound-on-IP-stability}
For every metric space $(X,d)$, every clustering $\cC$ of $(X,d)$ must be $\beta(\cC)$-IP stable
\end{observation}
\begin{proof}
Let $\cC$ be a clustering in a metric space $(X,d)$, let $p \in X$ be a point such that $C(p)\neq\{p\}$, and let $C' \in \cC$ be a cluster other than $C(p)$.
We need to show that $\avg(p,C(p)\setminus\{p\}) \leq \beta(\cC) \cdot \avg(p,C')$:
By the definition of $\avg$, for every set $S$, $\min_{p' \in S} d(p,p') \leq \avg(p,S)\leq \max_{p' \in S}d(p,p')$. Using this fact and the definition of $\beta(\cC)$,
\begin{align*}
    \avg(p,C(p)\setminus\{p\})
    \leq \max_{p' \in C(p) \setminus \{p\}} d(p,p')
    \leq \max_{p',p'' \in C(p)} d(p'',p')
    &\leq \beta(\cC) \cdot \min_{p'' \in C(p), p'\in C'} d(p'',p')\\
    &\leq \beta(\cC) \cdot \min_{p'\in C'} d(p,p')\\
    &\leq \beta(\cC) \cdot \avg(p,C').
\end{align*}
\end{proof}

\begin{lemma}[Lemma~1 in~\citep{ahmadi2022individual}; 
from~\cite{daniely2012clustering}]\label{lem:separation}
   Let $\cC$ be a clustering such that for every $C\neq C'\in \cC$ and $p\in C$, $\avg(p, C) \le \alpha \cdot \avg(p, C')$, where $\alpha < 1$. Then, 
    \begin{enumerate}
        \item For every $p\in C$, $p'\in C'$, 
        \begin{align}\label{eq:all-dist-equal}
            (1-\alpha) \cdot \avg(p, C') \le d(p, p') \le \frac{1 + \alpha^2}{1-\alpha} \cdot \avg(p, C'),
        \end{align}
        \item For every $p, q\in C$ and $C'\neq C$,
        \begin{align}\label{eq:iner-dist-smaller}
            d(p,q) \le \frac{2\alpha}{1-\alpha} \cdot \avg(p, C').
        \end{align}
    \end{enumerate}
\end{lemma}
\begin{corollary}\label{cor:inter-to-intra-distance}
    Let $\cC$ be a clustering such that for every $C\neq C'\in \cC$ and $p\in C$, $\avg(p, C) \le \alpha \cdot \avg(p, C')$, where $\alpha < 1$. Then, for every $C\neq C' \in \cC$,
    \begin{align*}
        \max_{p,q\in C} d(p,q) \le \beta \cdot \min_{p\in C, p'\in C'} d(p,p'),
    \end{align*}
    where $\beta = \frac{2\alpha (1+\alpha^2)}{(1-\alpha)^5}$.
\end{corollary}
\begin{proof}
    Let $p,q$ denote the farthest pair of points in $C$. Moreover, let $x\in C, x'\in C'$ denote the closest pair of points in $C$ and $C'$.  
    \begin{align*}
        d(p,q) 
        &\le \frac{2\alpha}{1-\alpha} \avg(p, C') && \rhd\text{By~\eqref{eq:iner-dist-smaller}} \\
        &\le \frac{2\alpha}{1-\alpha} \cdot \frac{1}{1-\alpha} \cdot d(p,x') \\
        &\le \frac{2\alpha}{(1-\alpha)^2} \cdot \frac{1 + \alpha^2}{(1-\alpha)^2} \cdot \avg(x', C) && \rhd\text{By~\eqref{eq:all-dist-equal}} \\
        &\le \frac{2\alpha}{(1-\alpha)^2} \cdot \frac{1+\alpha^2}{(1-\alpha)^3} \cdot d(x,x') && \rhd\text{By~\eqref{eq:all-dist-equal}}
    \end{align*}
\end{proof}

Finally, we use the following theorem in order to find the clustering of minimal $\beta(C)$. The proof of \cref{:thm:dynamic-programming-to-find-clustering-with-minimum-beta} is deferred to \cref{:sec:dynamic-programming-to-find-clustering-with-minimum-beta}.

\begin{theorem}\label{:thm:dynamic-programming-to-find-clustering-with-minimum-beta}
There exists a deterministic algorithm that, given a metric space $(X,d)$ and a desired number of clusters $k$ that admit a $k$-clustering $\cC^*$ of $(X,d)$ with $\beta(\cC^*) < 1$, computes a $k$-clustering $\cC$ of $(X,d)$ with $\beta(\cC) \leq \beta(\cC^*)$.
\end{theorem}

Now, we are ready to prove \cref{:thm:when-there-is-a-very-good-clustering}.

\begin{proof}[Proof of \cref{:thm:when-there-is-a-very-good-clustering}]
Given a metric space $(X,d)$ and a desired number of clusters $k$, such that there exists a $k$-clustering $\cC^*$ of $(X,d)$ that is $\alpha^*$-IP stable for $\alpha^* < 0.001$.
Then, by \cref{cor:inter-to-intra-distance}, $\beta(\cC^*) \leq \frac{2\alpha^*(1+(\alpha^*)^2)}{(1-\alpha^*)^5} \leq 3\alpha^*<1$.
So, the algorithm from \cref{:thm:dynamic-programming-to-find-clustering-with-minimum-beta} on $(X,d)$ and $k$ outputs a $k$-clustering $\cC$ with $\beta(\cC) \leq \beta(\cC^*) \leq 3\alpha^*$.
Thus, by \cref{:obs:beta-is-upper-bound-on-IP-stability}, the clustering $\cC$ is $\alpha$-IP stable for $\alpha=\beta(\cC) \leq 3\alpha^*$.
%
\end{proof}

\subsection{Proof of \texorpdfstring{\cref{:thm:dynamic-programming-to-find-clustering-with-minimum-beta}}{Theorem 9}}\label{:sec:dynamic-programming-to-find-clustering-with-minimum-beta}
\begin{algorithm}[H]
\caption{$\createTree$}\label{:alg:create-tree}
\KwData{$(X,d)$, and an MST $T$ of $(X,d)$.}
\KwResult{A binary tree $\tau$ with $n$ leaves (corresponding to $X$), where each node $u$ of $\tau$ represents a subset $f(u)$ of $X$ s.t. $f(u) = f(\text{left-child}(u)) \cup f(\text{right-child}(u))$.}
$root \leftarrow$ new node\\
$f(root) \leftarrow X$\\
\If{$|X| = 1$}{
    left-child$(root) \leftarrow$ null\\
    right-child$(root) \leftarrow$ null\\
    \Return $root$\\
}
$e_{\mathrm{max}} \leftarrow$ the edge of $T$ with maximum length\\
$C_1,C_2 \leftarrow$ the two connected components of $T \setminus \{e_{\mathrm{max}}\}$\\
$(X_1,d),(X_2,d)$ the restrictions of $(X,d)$ to $C_1$ and $C_2$\\
$T_1,T_2 \leftarrow$ the restrictions of $T$ to $C_1$ and $C_2$\\
left-child$(root) \leftarrow \createTree((X_1, d), T_1)$\\
right-child$(root) \leftarrow \createTree((X_2, d), T_2)$ \\
\Return $root$
\end{algorithm}

\begin{definition}\label{:def:clustering-induced-by-tree}
Let $(X,d)$ be a metric space and let $\tau$ be the output of Algorithm \ref{:alg:create-tree} on this metric space.
Then, we say that a clustering $\cC$ of $X$ is \emph{induced} by the tree $\tau$ if, for every cluster $C \in \cC$, there exists a node $u$ of $\tau$ with $f(u)=C$.
\end{definition}

\begin{claim}\label{:cl:best-clustering-is-induced-by-tree}
For every metric space $(X,d)$ and every clustering $\cC$ of $(X,d)$ with $\beta(\cC) < 1$, $\cC$ must be induced by the tree $\tau$ resulting from Algorithm \ref{:alg:create-tree} on $(X,d)$ and a minimum spanning tree of it.
\end{claim}

\begin{lemma}\label{:lem:dynamic-programming-to-find-best-clustering-induced-by-tree}
There exists a deterministic polynomial time algorithm that, given a metric space $(X,d)$, a desired number of clusters $k$, and a tree $\tau$ returned by Algorithm \ref{:alg:create-tree} on $(X,d)$, computes a $k$-clustering $\cC$ with $\beta(\cC) = \min_{\cC' \in \cS_{\tau}}\beta(\cC')$, where $\cS_{\tau}$ denote the set of $k$-clusterings induced by $\tau$. 
\end{lemma}

\begin{proof}[Proof of \cref{:thm:dynamic-programming-to-find-clustering-with-minimum-beta}] 
Let $(X,d)$ be a metric space that admit a $k$-clustering $\cC^*$ with $\beta(\cC^*) < 1$.
To find a $k$-clustering $\cC$ with $\beta(\cC) \leq \beta(\cC^*)$, first we compute a minimum spanning tree $T$ of $(X,d)$, then run Algorithm~\ref{:alg:create-tree} on $(X,d)$ and $T$ to produce a tree $\tau$, and finally run the algorithm from \cref{:lem:dynamic-programming-to-find-best-clustering-induced-by-tree} on the metric space $(X,d)$ and the tree $\tau$ with parameter $k$.

By \cref{:cl:best-clustering-is-induced-by-tree}, the clustering $\cC^*$ must be induced by the tree $\tau$.
Therefore, the algorithm from \cref{:lem:dynamic-programming-to-find-best-clustering-induced-by-tree} must return a clustering $\cC$ with $\beta(\cC) \leq \beta(\cC^*)$, as we need.
Since computing a minimum spanning tree of $(X,d)$ can be done in polynomial time, and since Algorithm \ref{:alg:create-tree} and the algorithm from \cref{:lem:dynamic-programming-to-find-best-clustering-induced-by-tree} both run in polynomial time, this suffices to prove \cref{:thm:dynamic-programming-to-find-clustering-with-minimum-beta}.
\end{proof}

\subsubsection{Proof of \texorpdfstring{\cref{:cl:best-clustering-is-induced-by-tree}}{Claim 17}}\label{:sec:best-clustering-is-induced-by-tree}

First we need to prove the following claim.
\begin{claim}\label{clm:tree-property}
    Let $\tau$ be the output of Algorithm \ref{:alg:create-tree} on $(X,d)$ and a minimum spanning tree of $(X,d)$ as inputs, where $(X,d)$ admits an $k$-clustering $\cC$ with $\beta(\cC)<1$.
    Then, every node in $\tau$ represents either a subset of clusters in $\cC$, or a subset of a cluster in $\cC$.
\end{claim}
\begin{proof}
    We show that for every node $u\in \tau$, if there exists $p\in C,q\in C'$ where $C\neq C'$, such that $\{p,q\} \subseteq f(u)$, then $f(u)$ contains all points in $C\cup C'$. 
    
    Suppose this property does not hold for all nodes in $\tau$. consider a node $v \in \tau$ closest to the root that does not satisfy this property. Note that the root itself trivially satisfies the property. Let $v_p$ denote the parent node of $v$ in $\tau$. Without loss of generality, suppose that $v$ does not contain all points in $C$ and its intersection with $C'$ is non-empty. Since $v$ is the closest node to the root that does not satisfy the described property, $v_p$ contains all points in $C \cup C'$. Let $T_{v_p}$ be the restriction of the MST $T$ on the points contained in $v_p$.

    First, we prove that $T_{v_p}[C]$, the induced subgraph of $T_{v_p}$ on $C$ is a connected component. Otherwise, $T_{v_p}[C]$ has at least two connected components $C_1, C_2$ such that both $T_{v_p}[C_1]$ and $T_{v_p}[C_2]$ are trees and they both have distinct edges leaving $C_1$ and $C_2$ in $T_{v_p}$. However, it contradicts the minimality of $T_{v_p}$, and consequently the minimality of $T$. This follows, because by the condition $\beta(\cC)<1$, all edges within $C$ are strictly smaller than any edge with one end-point in $C$ and one end-point outside of $C$. So, if we add an arbitrary edge between $C_1, C_2$ to $T$, then we can remove one of the edges leaving $C$; hence, the cost of $T$ after this modification strictly decrease which is a contradiction.  

    Now, if by removing an edge in $T_v$ the points in $C$ are split into at least two non-empty sets, then the edge has to be between two points in $C$. However, this is not possible because $T_{v_p}$ is a tree spanning both $C$ and $C'$ and by the condition $\beta(\cC)<1$ any edge between $C$ and $C'$ is strictly large than any edge corresponding to two points in $C$. So, all nodes in $\tau$ satisfy the desired property. 
\end{proof}

We are now ready to prove \cref{:cl:best-clustering-is-induced-by-tree}.

\begin{proof}[Proof of~\cref{:cl:best-clustering-is-induced-by-tree}]
The proof is by contradiction.
Lets assume that $\cC$ is not induced by the tree $\tau$.
Then, there must exist a cluster $C \in \cC$ such that no node $v$ in $\tau$ represents the set $C\subseteq X$.
Let $v^C$ be the highest node in $\tau$, i.e., closest to the root, that represents a set $f(v^C) \subseteq C$. Note that such a node exists because, for every point $p \in X$, there exists a leaf node in $\tau$ that represents the set $\{p\}$.
By the choice of $C$, $f(v^C) \neq C$, so $f(v^C)$ is strictly contained in $C$,
\begin{equation}\label{:eq:proof-of-best-clustering-is-induced-by-tree:C-strictly-contained-in-C}
    f(v^C) \subset C,
\end{equation}
which means that $f(v^C) \neq X$, and thus $v^C$ has a parent $v^C_{\mathrm{parent}}$ and a sibling $v^C_{\mathrm{sibling}}$.
By the choice of $v^C$, since $v^C_{\mathrm{parent}}$ is closer to the root than $v^C$, there must exist some cluster $C' \in \cC$ other than $C$ such that $f(v^C_{\mathrm{parent}}) \cap C' \neq \emptyset$.
By the design of Algorithm \ref{:alg:create-tree} and the choice of $v^C$, $f(v^C)$ is non-empty and $f(v^C) \subseteq f(v^C_{\mathrm{parent}}) \cap C$, which means that $f(v^C_{\mathrm{parent}}) \cap C \neq \emptyset$.
By \cref{clm:tree-property}, since $f(v^C_{\mathrm{parent}}) \cap C' \neq \emptyset$ and $f(v^C_{\mathrm{parent}}) \cap C \neq \emptyset$, we must have
\begin{equation}\label{:eq:proof-of-best-clustering-is-induced-by-tree:parent-node-contains-both-clusters}
    C \cup C' \subseteq f(v^C_{\mathrm{parent}}).
\end{equation}
Since $f(v^C)$ is strictly contained in $C$ (\cref{:eq:proof-of-best-clustering-is-induced-by-tree:C-strictly-contained-in-C}) and since $C'$ is non-empty, there must be at least one point in $C$ and at least one point in $C'$ that are not in $f(v^C)$.
By \cref{:eq:proof-of-best-clustering-is-induced-by-tree:parent-node-contains-both-clusters} and by the fact that $f(v^C_{\mathrm{parent}}) = f(v^C) \cup f(v^C_{\mathrm{sibling}})$, these two points must both belong to $f(v^C_{\mathrm{sibling}})$; hence, by~\cref{clm:tree-property}, $C \cup C' \subseteq f(v^C_{\mathrm{sibling}})$.
So, $f(v^C) \subseteq C \subseteq f(v^C_{\mathrm{sibling}})$, which contradicts the structure of $\tau$. 

Thus, the clustering $\cC$ must be induced by $\tau$.
This concludes the proof of \cref{:cl:best-clustering-is-induced-by-tree}.
\end{proof}

\subsubsection{Proof of \texorpdfstring{\cref{:lem:dynamic-programming-to-find-best-clustering-induced-by-tree}}{Lemma 11}}\label{:sec:dynamic-programming-to-find-best-clustering-induced-by-tree}

\begin{definition}\label{:def:clustering-restricted-to-subtree}
Given a metric space $(X,d)$ and the tree $\tau$ from Algorithm~\ref{:alg:create-tree}. 
For every node $u$ of $\tau$, we say that a clustering $\cC_u$ of $f(u)$ is induced by the subtree of $\tau$ rooted at $u$ if, for every cluster $C \in \cC_u$, there exists a node $v$ in the subtree of $\tau$ rooted at $u$ such that $C = f(v)$.
We note that $\beta(\cC_u)$ still refers to the quantity $\max_{C \in \cC_u} \beta(C)$, where each $\beta(C)$ is defined with respect to the whole space $X$.
\end{definition}

\begin{proof}[Proof of \cref{:lem:dynamic-programming-to-find-best-clustering-induced-by-tree}]

We follow a dynamic programming approach. For each tuple $(u, k')$ where $u$ is a node in $\tau$ and $1\le k'\le k$, $DP(u,k')$ denotes the $k'$-clustering $\cC_u$ of $f(u)$ induced by the subtree of $\tau$ rooted at $u$ minimizing $\beta(\cC_u)$, or Null if such a $k'$-clustering does not exist.
Note that for any node $u$, $DP(u,1)=\{f(u)\}$, i.e., there is a trivial $1$-clustering of $f(u)$  which is a single cluster containing $f(u)$.
Furthermore, for every leaf node $\ell$ and every $k' > 1$, $DP(\ell, k')=\mathrm{Null}$ because there is no way to partition $f(\ell)$ into $k' > 1$ disjoint non-empty clusters, as $|f(\ell)|=1$.
For every node $u \in \tau$, let $u_R$ and $u_L$ respectively denote the right and left child of $u$.
Then, the update rule in the dynamic programming is as follows:
\begin{itemize}
    \item $DP(u, 1) = \{f(u)\}$, where $f(u)$ is the set of points contained in $u$.
    \item $DP(u, k'>1) = \argmin_{\cC_{u,k'} \in \mathrm{Candidates}_{u,k'}} \beta(\cC_{u,k'})$ if $\mathrm{Candidates}_{u,k'}$ is not empty, where
    \begin{align*}
       \mathrm{Candidates}_{u,k'} 
        &\defeq \{DP(u_R,i)\cup DP(u_L,k'-i) \mid i \in [k'-1] \\ 
        &\text{ s.t. $DP(u_R,i)$ and $DP(u_L,k'-i)$ are not Null}\}, 
    \end{align*} 
    \item $DP(u, k'>1) = \mathrm{Null}$ if $\mathrm{Candidates}_{u,k'}$ is an empty set.
\end{itemize}

The dynamic programming has $|\tau|k=O(nk)$ cells and updating each of them requires at most $O(k n^2)$ time, the total runtime of the algorithm is $O(n^3 k^2)$.
Furthermore, by definition, the value of $DP(\mathrm{root}(\tau), k)$ is exactly the desired clustering that we need to compute.

Regarding the correctness of the algorithm, it follows by an inductive argument and the definition of clustering induced by $\tau$ together with~\cref{:cl:best-clustering-is-induced-by-tree}. 
\end{proof}

\section{Proof of \texorpdfstring{\cref{:thm:constant-median-IP}}{Theorem 4}}
We first present a general framework in \cref{:sec:general-framework} for extending the techniques of \cref{:sec:basic-merge-and-split-algorithm} to other types of IP stability.
Then, throughout most of this section, we show how to apply this framework in order to construct an algorithm that, given $(X,d)$ and a desired number of clusters $k$, produces an $(O(1),\sqrt{\median})$-IP stable $k$-clustering.
To finish of, we use the following theorem, concluding that such a clustering must also be $(O(1),\median)$-IP stable.

\begin{theorem}\label{:thm:constant-sqrt-median-IP}
There exists a deterministic polynomial-time algorithm that, given a metric space $(X,d)$ and a desired number of clusters $2 \leq k \leq |X|$, computes an $(O(1), \sqrt{\median})$-IP stable $k$-clustering of $(X,d)$, where $\sqrt{\median}$ denotes the square root of the median function from \cref{:thm:constant-median-IP}.
\end{theorem}
The proof of \cref{:thm:constant-sqrt-median-IP} is deferred to \cref{:sec:concluding-sqrt-median-IP}.

\begin{lemma}\label{:lem:sqrt-median-to-median}
Given a metric space $(X,d)$, every $(\alpha,\sqrt{\median})$-IP stable cluster of $(X,d)$ must also be $(\alpha^2,\median)$-IP stable.
\end{lemma}
\begin{proof}
Given a metric space $(X,d)$, if a clustering $\cC$ is $(\alpha,\sqrt{\median})$-IP stable, then, for every point $p$ with $C(p)\neq\{p\}$ and every cluster $C' \in \cC \setminus \{C\}$,
$ \sqrt{\median(p,C(p)\setminus\{p\})} \leq \alpha \cdot \sqrt{\median(p,C')}$. So, $\median(p,C(p)\setminus\{p\}) \leq \alpha^2 \cdot \median(p,C')$. Hence, $\cC$ is a $(\alpha^2,\median)$-IP stable clustering.
\end{proof}

Then, \cref{:thm:constant-median-IP} follows from \cref{:thm:constant-sqrt-median-IP} and \cref{:lem:sqrt-median-to-median}.
\begin{proof}[Proof of \cref{:thm:constant-median-IP}]
We use the algorithm from \cref{:thm:constant-sqrt-median-IP}.
By \cref{:lem:sqrt-median-to-median}, the clustering that this algorithm returns must be $(O(1),\median)$-IP stable.
Therefore, this algorithm satisfies all the properties promised in \cref{:thm:constant-median-IP}.
\end{proof}

\subsection{Framework for General \texorpdfstring{$f$}{f}-IP Stability}\label{:sec:general-framework}
In this subsection, we show how the techniques from \cref{:sec:basic-merge-and-split-algorithm} can be extended to a framework for constructing algorithms that compute $f$-IP stable clusterings for other functions $f$.
\begin{theorem}\label{:thm:framework-for-general-f}
For every metric space $(X,d)$, let $f:X \times 2^X \to \bbR_{\geq 0}$ and $\Phi_f:(2^X \setminus \{\emptyset\}) \to \bbR_{\geq 0}$ and $\alpha > 0$. Furthermore, suppose the following properties hold:
\begin{enumerate}[leftmargin=*]
    \item For every non-empty subset $S \subset X$ and every point $p \in X \setminus S$,
    \[
     f(p,S) \leq \Phi_f(S \cup \{p\}) - \Phi_f(S) \leq \alpha \cdot f(p,S).
    \]
    \item There exists a deterministic polynomial-time algorithm that computes $f(p,S)$, given any $S \subset X$, and $p \in X \setminus S$.
    \item There exists a deterministic polynomial-time algorithm that computes a $\poly(n)$-approximation of $\Phi_f(C)$, given any $C \subset X$.
    \item There exists a deterministic polynomial-time algorithm {\sc Split}$_f$ that, given a clustering $\cC$ of $X$, finds a cluster $C^* \in \cC$ and a partition $(C^*_1,C^*_2)$ of $C^*$ such that $\Phi(C^*_1) + \Phi(C^*_2) \leq \Phi(C^*) - \frac{\Phi(\cC)}{\poly(n)}$.
    \item For every two disjoint non-empty clusters $C,C' \subseteq X$, and every $p \in C$,
    \begin{align*}
        \Phi_f(C \cup C') \leq \Phi_f(C) + \Phi_f(C') + \poly(n)\cdot(f(p,C\setminus\{p\}) + f(p,C'))
    \end{align*}
    \item There exists a deterministic polynomial time algorithm {\sc Initialize}$_f$ that, given a metric space $(X,d)$ and a desired number of clusters $k$, finds a $k$-clustering $\cC$ that $\poly(n)$-approximately minimizes the potential function $\Phi_f(\cC)$; more precisely, $\Phi_f(\cC) \leq \poly(n) \cdot \Phi_f(\cC^*)$, where $\cC^*$ is a $k$-clustering minimizing $\Phi_f$, and $\Phi_f(\cC)\defeq \sum_{C \in \cC}\Phi_f(C)$.
\end{enumerate}
Then, for every constant $c>1$, there exists a deterministic polynomial-time algorithm that, given a desired number of clusters $k$, computes a $(c\alpha,f)$-IP stable $k$-clustering of $(X,d)$.
\end{theorem}

Since the proof of \cref{:thm:framework-for-general-f} is analogous to the arguments in \cref{:sec:basic-merge-and-split-algorithm}, we only provide the following sketch.

\begin{proof}[Proof Sketch for \cref{:thm:framework-for-general-f}]
The algorithm is analogous to Algorithm \ref{:alg:basic-merge-and-split}.
It begins by initializing a clustering using the algorithm from property 6 in \cref{:thm:framework-for-general-f}.
Then, as long as the current clustering is not $(c\alpha,f)$-IP stable (which can be checked using the algorithm from property 2), it picks a point $p$ and cluster that violate the constraint of $(c\alpha,f)$-IP stability, and produces an estimate of $\Phi_{f}(\cC)$ using the algorithm from property 3 of \cref{:thm:framework-for-general-f}.
If the estimate of $\Phi_{f}(\cC)$, the value of $f(p,C \setminus \{p\}) + f(p,C')$, and properties 4 and 5 of \cref{:thm:framework-for-general-f} are together enough to guarantee that a merge-and-split step would decrease the overall potential of the clustering by a multiplicative $(1-\frac{1}{\poly(n)})$ factor, then it executes a merge-and-split step, merging clustering $C(p)$ and $C'$, then using the algorithm from property 4 to split some cluster $C^*$.
Otherwise, it must be that $f(p,C \setminus \{p\}) + f(p,C') \geq \frac{\Phi_{f}(\cC)}{\poly(n)}$, so, since $p$ and $C'$ violate $(c\alpha,f)$-IP stability, it must be that $f(p,C \setminus \{p\}) \geq \frac{1}{2}(f(p,C \setminus \{p\}) + f(p,C')) \geq \frac{\Phi_{f}(\cC)}{\poly(n)}$, so, by property 1 of \cref{:thm:framework-for-general-f}, swapping the point $p$ from cluster $C(p)$ to cluster $C'$ must reduce the overall potential by 
\begin{align*}
 f(p,C \setminus \{p\}) -\alpha f(p,C') > f(p,C \setminus \{p\}) - \frac{\alpha}{c\alpha}f(p,C \setminus \{p\}) = \Omega(f(p,C \setminus \{p\})) \geq \frac{\Phi_{f}(\cC)}{\poly(n)},   
\end{align*}
and this is what the algorithm does.

In summary, using all the properties from \cref{:thm:framework-for-general-f}, the algorithm is able to initialize a $k$-clustering $\cC$ in polynomial time such that $\Phi_{f}(\cC)$ is a $\poly(n)$-approximation of $\Phi_{f}(\cC^*)$, where $\cC^*$ is a minimizer of $\Phi_{f}$, and then continue in iterations as long as the current maintained clustering is not $(c\alpha,f)$-IP stable, where each iteration runs in polynomial time and reduces the potential $\Phi_{f}$ of the maintained clustering $\cC$ by a multiplicative factor $<(1-\frac{1}{\poly(n)})$.
Thus, since the potential of the maintained clustering can never go below $\Phi_{f}(\cC^*)$, the number of iterations made by the algorithm is bounded by $\poly(n)$, so the overall runtime of the algorithm is polynomial.
\end{proof}

\subsection{Potential Function and Preliminaries for \texorpdfstring{$\sqrt{\median}$}{sqrt(median)}-IP}

In this section, we introduce the potential function $\Phi_{\sqrt{\median}}$. Just as the potential function discussed in \cref{:sec:potential-function} played a crucial role in \cref{:sec:basic-merge-and-split-algorithm}, this new function will be integral to the framework presented in \cref{:sec:general-framework}.
Furthermore, we prove some preliminary observations for analyzing this potential function.

\begin{definition}[Potential function for median-IP] \label{:def:potential-for-sqrt-median}
Given a metric space $(X,d)$, let $G=(X,E,w)$ be the complete clique graph on $X$, together with edge-lengths $w(p_1,p_2) \defeq \sqrt{d(p_1,p_2)}$ for each edge $(p_1,p_2)$ of the clique.
Then, we define the potential function $\Phi_{\sqrt{\median}}:(2^X \setminus \{\emptyset\}) \to \bbR_{\geq 0}$ such that, for every non-empty subset $C \subseteq X$,
\[
 \Phi_{\sqrt{\median}}(C) \defeq \begin{cases}
       \frac{1}{2 - \sqrt{2}} \cdot \max_{\text{TSP tour $\Pi$ of $G[C]$}} w(\Pi) &\quad\text{if }|C| \geq 2\\
       0 &\quad\text{if }|C|=1 
     \end{cases}
\]
where $G[C]$ is the subgraph of $G$ induced on the set $C$ of vertices.\footnote{A TSP tour $\Pi$ of $G[C]$ is a Hamiltonian cycle in $G[C]$, and its length is $w(\Pi) = \sum_{e \in E(\Pi)} w(e) = \sum_{(p_1,p_2)\in E(\Pi)\} \sqrt{d(p_1,p_2)}}$.
If $|C|=2$, i.e. if $C = \{p_1,p_2\}$ for some $p_1,p_2 \in X$, then we consider $\Pi$ to be taking the edge $(p_1,p_2)$ twice, which means that $w(\Pi) = 2w(p_1,p_2)=2\sqrt{d(p_1,p_2)}$.
}
\end{definition}

To use this potential function in the framework of \cref{:sec:general-framework}, we need to prove the following theorem.
The proof of the theorem is in \cref{:sec:proof-of-potential-function-for-sqrt-median}.

\begin{theorem}\label{:thm:potential-for-sqrt-median}
For every metric space $(X,d)$, every non-empty set $S \subset X$ and every point $p \in X \setminus S$, the potential function from \cref{:def:potential-for-sqrt-median} satisfies
\[
 \sqrt{\median(p,S)} \leq \Phi_{\sqrt{\median}}(S \cup \{p\}) - \Phi_{\sqrt{\median}}(S) \leq O(1) \cdot \sqrt{\median(p,S)},
\]
where $\median(p,S)$ is defined as in \cref{:thm:constant-median-IP}.
\end{theorem}

In order to analyze the potential function defined in this section, the following observations will be useful.

\begin{observation}\label{:obs:edge-exists}
Let $G=(V,E)$ be a complete clique graph with at least two vertices, and let $\Pi$ be a TSP tour of $G$. If a partition $(A,B)$ of $V$ satisfies $|B|>|A|$, then there must be an edge $(p,p') \in E(\Pi)$ such that $p,p' \in B$.
\end{observation}

\begin{observation}\label{:obs:triangle-inequality-in-sqrt-metric}
Given a metric space $(X,d)$ and the graph $G$ as defined in \cref{:def:potential-for-sqrt-median}, for every three points $p,p',p'' \in X$,
$w(p,p') \geq w(p',p'') - w(p,p'')$.
\end{observation}

\begin{observation}\label{:obs:bound-on-cluster-median-potential}
For every metric space $(X,d)$ and every non-empty subset $C \subseteq X$,
\[
 \sqrt{\diameter(C)} \leq \Phi_{\sqrt{\median}}(C) \leq O(|C| \cdot \sqrt{\diameter(C)}).
\]
\end{observation}

\begin{corollary}\label{:cor:bound-on-median-potential-of-clustering}
For every metric space $(X,d)$ and every non-empty subset $\cC$,
\[
 \max_{C \in \cC} \sqrt{\diameter(C)} \leq \sum_{C \in \cC} \Phi_{\sqrt{\median}}(C) \leq O(n \cdot \max_{C \in \cC} \sqrt{\diameter(C)}).
\]
\end{corollary}

\subsubsection{Proofs of Observations}
In this section, we prove \cref{:obs:edge-exists} and \cref{:obs:triangle-inequality-in-sqrt-metric}.

\begin{proof}[Proof of \cref{:obs:edge-exists}]
If the clique graph $G=(V,E)$ has exactly two vertices, then the TSP tour $\Pi$ must take the edge between these two vertices, and the condition $|B|>|A|$ implies that $B=V$, which means that this edge has both its endpoints in $B$.
So, for the rest of the proof, we can assume that $|V|\neq 2$.
Then, by the condition $|V|\geq 2$, we have $|V|> 2$.
Therefore, each vertex $v \in B$ must have at least two incident edges of $E(\Pi)$.
So, since $|B| > |A|$ and since each vertex of $A$ can have at most two incident edges of $E(\Pi)$,
\[
 \sum_{v \in B} \mathrm{deg}_{E(\Pi)}(v) \geq 2|B|>2|A|\geq \sum_{v \in A} \mathrm{deg}_{E(\Pi)}(v).
\]
Since $(A,B)$ is a partition of $V$, the above inequality implies that there must exist some edge of $E(\Pi)$ whose endpoints are both in $B$.
\end{proof}

\begin{proof}[Proof of \cref{:obs:triangle-inequality-in-sqrt-metric}]
By the definition of the lengths $w$ and by the triangle inequality 
\[
 w(p',p'') = \sqrt{d(p',p'')} \leq \sqrt{d(p',p) + d(p,p'')}
\]
so, by the inequality $\sqrt{a+b} \leq \sqrt{a}+\sqrt{b}$, and by again using the definitions of the lengths $w$,
\begin{align*}
    w(p',p'')
    \leq \sqrt{d(p',p) + d(p,p'')}
    \leq \sqrt{d(p',p)} + \sqrt{d(p,p'')}
     = w(p',p) + w(p,p'').
\end{align*}
\end{proof}

\begin{proof}[Proof of \cref{:obs:bound-on-cluster-median-potential}]
Firstly, if $|C|=1$, then $\Phi_{\sqrt{\median}}(C)$ and $\sqrt{\diameter(C)}$ are both zero, so the inequalities hold.
Thus, for the rest of the proof, we assume $|C|\geq 2$.
So, let $p,p' \in C$ be two point in $C$ such that $d(p,p')=\diameter(C)$.
Then, there exists a TSP tour of $C$ that uses the edge $(p,p')$, which means that
\[
 w(p,p')
 \leq \max_{\text{TSP tour $\Pi$ of $G[C]$}} w(\Pi)
 \leq \Phi_{\sqrt{\median}}(C).
\]
So, since we picked the point $p,p'$ such that $d(p,p')=\diameter(C)$,
\begin{align*}
    \sqrt{\diameter(C)} = \sqrt{d(p,p')} = w(p,p') \leq \Phi_{\sqrt{\median}}(C).
\end{align*}
Furthermore, for every TSP tour $\Pi$ of $G[C]$, since $\Pi$ contains exactly $|C|$ edges, and the length of each edge $e \in E(\Pi)$ is $w(e)=\sqrt{d(e)}\leq\sqrt{\diameter(C)}$, we get that every TSP tour $\Pi$ of $G[C]$ has $w(\Pi) \leq |C| \cdot \sqrt{\diameter(C)}$, which exactly means that
\begin{align*}
    \Phi_{\sqrt{\median}}(C)
    \leq \frac{1}{2-\sqrt{2}}\cdot|C| \cdot \sqrt{\diameter(C)}
    = O(|C| \cdot \sqrt{\diameter(C)}).
\end{align*}
Thus, we proved both sides of the desired inequality from \cref{:obs:bound-on-cluster-median-potential}, which concludes the proof of the observation.
\end{proof}

\begin{proof}[Proof of \cref{:cor:bound-on-median-potential-of-clustering}]
By \cref{:obs:bound-on-cluster-median-potential}, each cluster $C \in \cC$ satisfies
\[
 \Phi_{\sqrt{\median}}(C) \leq O(|C| \cdot \sqrt{\diameter(C')})
 \leq O(|C| \cdot \max_{C' \in \cC} \sqrt{\diameter(C')}).
\]
So,
\begin{align*}
    \sum_{C \in \cC} \Phi_{\sqrt{\median}}(C)
    \leq \sum_{C \in \cC} O(|C| \cdot \max_{C' \in \cC} \sqrt{\diameter(C')})
    &\leq \left(\sum_{C \in \cC} |C|\right) \cdot O(\max_{C' \in \cC} \sqrt{\diameter(C')})\\
    &= O(n \cdot \max_{C' \in \cC} \sqrt{\diameter(C')})\\
    &= O(n \cdot \max_{C \in \cC} \sqrt{\diameter(C)})
\end{align*}
Furthermore, letting $C_{\max}$ be the cluster with largest diameter in $\cC$, then using \cref{:obs:bound-on-cluster-median-potential},
\begin{align*}
    \max_{C \in \cC} \sqrt{\diameter(C)}
    = \sqrt{\diameter(C_{\max})}
    \leq \Phi_{\sqrt{\median}}(C_{\max})
    \leq \sum_{C \in \cC} \Phi_{\sqrt{\median}}(C).
\end{align*}
Thus, we proved both sides of the desired inequality from \cref{:cor:bound-on-median-potential-of-clustering}, concluding the proof of the corollary.
\end{proof}

\subsection{Proof of \texorpdfstring{\cref{:thm:potential-for-sqrt-median}}{Theorem 12}} \label{:sec:proof-of-potential-function-for-sqrt-median}
We being by proving the following claim.

\begin{claim}\label{:cl:getting-rid-of-p}
Given a metric space $(X,d)$, a set $S \subset X$ with $|S|\geq2$, a point $p \in X \setminus S$, and a TSP tour $\Pi_{S \cup \{p\}}$ of the graph $G[S \cup \{p\}]$ as defined in \cref{:def:potential-for-sqrt-median} and \cref{:thm:potential-for-sqrt-median}.
If there exists a point $p' \in S$ such that $d(p,p') \leq \median(p,S)$ and the edge $(p',p)$ belongs to the TSP tour $\Pi_{S \cup \{p\}}$, then there exists a TSP tour $\Pi_{S}$ of $G[S]$ such that $w(\Pi_{S}) \geq w(\Pi_{S \cup \{p\}}) - O(1) \cdot \sqrt{\median(p,S)}$.
\end{claim}
\begin{proof}
Since $|S| \geq 2$, there must be some point $p'' \in S$ other than $p'$ such that $(p,p'') \in E(\Pi_{S \cup \{p\}})$.
To construct the tour $\Pi_{S}$ to be like $\Pi_{S \cup \{p\}}$ except the two edges $(p',p)$ and $(p,p'')$ are replaced by the single edge $(p',p'')$, thus skipping over the point $p$.
Therefore,
\[
 w(\Pi_{S}) - w(\Pi_{S \cup \{p\}})
 = w(p',p'') - w(p',p) - w(p,p'').
\]
So, by \cref{:obs:triangle-inequality-in-sqrt-metric},
\[
 w(\Pi_{S}) - w(\Pi_{S \cup \{p\}})
 \leq -2w(p',p)
 = -2\sqrt{d(p',p)}.
\]
Therefore, since we are given that $d(p,p') \leq \median(p,S)$,
\begin{align*}
    w(\Pi_{S})
    \geq w(\Pi_{S \cup \{p\}}) - 2\sqrt{d(p',p)} 
    \geq w(\Pi_{S \cup \{p\}}) - 2\sqrt{\median(p,S)},
\end{align*}
as we needed to prove.
\end{proof}

We are now ready to prove \cref{:thm:potential-for-sqrt-median}.

\begin{proof}[Proof of \cref{:thm:potential-for-sqrt-median}]
Let $(X,d)$, $S$, and $p$ be as in the theorem, and let $G$ be defined as in \cref{:def:potential-for-sqrt-median}. We will first deal with the case of $|S|=1$, and then separately prove each of the two inequalities in \cref{:thm:potential-for-sqrt-median} when $|S| \geq 2$.

\paragraph{Case $|S|=1$.}
In this case there exists some $p' \in X$ such that $S=\{p'\}$, so $\Phi_{\sqrt{\median}}(S)=0$ and $\Phi_{\sqrt{\median}}(S \cup \{p\})=\frac{1}{2-\sqrt{2}}\cdot 2\sqrt{d(p,p')}$, while $\median(p,S)=d(p,p')$.
This immediately implies that $\sqrt{\median(p,S)} \leq \Phi_{\sqrt{\median}}(S \cup \{p\}) - \Phi_{\sqrt{\median}}(S) \leq O(1) \cdot \sqrt{\median(p,S)}$, as we needed.

\paragraph{Left-Hand Side Inequality.}
In this paragraph, our goal is to prove the left-hand side inequality from \cref{:thm:potential-for-sqrt-median}, under the additional assumption $|S|\geq 2$.
In other words, we need to show that $\Phi_{\sqrt{\median}}(S) + \sqrt{\median(p,S)} \leq \Phi_{\sqrt{\median}}(S \cup \{p\})$.
So, let $\Pi_{S}$ be the maximum-length TSP tour in $G[S]$, which means that $\Phi_{\sqrt{\median}}(S) = \frac{1}{2 - \sqrt{2}}\cdot w(\Pi_{S})$.
We need to construct a TSP tour $\Pi_{S \cup \{p\}}$ of $G[S \cup \{p\}]$ with $\frac{1}{2 - \sqrt{2}}\cdot w(\Pi_{S \cup \{p\}}) \geq \frac{1}{2 - \sqrt{2}}\cdot w(\Pi_{S}) + \sqrt{\median(p,S)}$.
Let $A \defeq \{p' \in S \mid d(p',p) < \median(p,S)\}$ and let $B \defeq \{p' \in S \mid d(p',p) \geq \median(p,S)\}$.
By the definition of $\median(p,S)$, $|B| > |A|$. (See \cref{:thm:constant-median-IP}.)
Since $(A,B)$ is a partition of $S$ with $|S|\geq 2$ and $|B|>|A|$, by \cref{:obs:edge-exists}, there must be some edge $(p_1,p_2) \in E(\Pi_{S})$ such that $p_1,p_2 \in B$.
We build the TSP tour $\Pi_{S \cup \{p\}}$ in $G[S \cup \{p\}]$ by starting with $\Pi_{S}$, then removing the edge $(p_1,p_2)$ and instead adding the edges $(p_1,p)$ and $(p,p_2)$.
It's not hard to see that
\begin{align*}
    w(\Pi_{S \cup \{p\}}) - w(\Pi_{S})
    = w(p_1,p) + w(p,p_2) - w(p_1,p_2)
    = \sqrt{d(p_1,p)} + \sqrt{d(p,p_2)} - \sqrt{d(p_1,p_2)}.
\end{align*}
Thus, by the triangle inequality,
\[
 w(\Pi_{S \cup \{p\}}) - w(\Pi_{S})
 \geq \sqrt{d(p_1,p)} + \sqrt{d(p,p_2)} - \sqrt{d(p_1,p) + d(p,p_2)}.
\]
So, using the inequality $\sqrt{a}+\sqrt{b} - \sqrt{a+b} - (2-\sqrt{2})\sqrt{\min\{a,b\}} \geq 0$,\footnote{It's very easy to see that this inequality holds for all $a,b\geq 0$. The desired values of $a$ and $b$ can always be achieved by first starting with setting $a=b$ and then increasing only one of $a$,$b$. Setting $a=b$ causes exact equality to hold, then increasing one of them can only increase the left-hand side expression, since $\min\{a,b\}$ does not change.}
\[
 w(\Pi_{S \cup \{p\}}) - w(\Pi_{S})
 \geq (2-\sqrt{2})\sqrt{\min\{d(p_1,p),d(p,p_2)\}}.
\]
Since $p_1,p_2 \in B$, they satisfy $\min\{d(p_1,p),d(p,p_2)\}\geq \median(p,S)$, so
\[
  w(\Pi_{S \cup \{p\}}) - w(\Pi_{S})
 \geq (2-\sqrt{2})\sqrt{\median(p,S)},
\]
which exactly gives the inequality
\[
 \frac{1}{2 - \sqrt{2}}\cdot w(\Pi_{S \cup \{p\}})
 \geq \frac{1}{2 - \sqrt{2}}\cdot w(\Pi_{S}) + \sqrt{\median(p,S)},
\]
as we needed.
This concludes the proof of $\sqrt{\median(p,S)} \leq \Phi_{\sqrt{\median}}(S \cup \{p\}) - \Phi_{\sqrt{\median}}(S)$.

\paragraph{Right-Hand Side Inequality}
In this paragraph, our goal is to prove the right-hand side inequality from \cref{:thm:potential-for-sqrt-median}, under the additional assumption $|S|\geq 2$.
In other words, we need to show that $\Phi_{\sqrt{\median}}(S \cup \{p\}) - O(1)\cdot\sqrt{\median(p,S)} \leq \Phi_{\sqrt{\median}}(S)$.
So, let $\Pi_{S \cup \{p\}}$ be the maximum-length TSP tour in $G[S \cup \{p\}]$, which means that $\Phi_{\sqrt{\median}}(S \cup \{p\}) = \frac{1}{2 - \sqrt{2}}\cdot w(\Pi_{S \cup \{p\}})$.
We need to show that there exists a TSP tour $\Pi_{S}$ of $G[S]$ with $\frac{1}{2 - \sqrt{2}}\cdot w(\Pi_{S}) \geq \frac{1}{2 - \sqrt{2}}\cdot w(\Pi_{S \cup \{p\}}) - O(1)\cdot \sqrt{\median(p,S)}$, which is an equivalent condition to $w(\Pi_{S}) \geq w(\Pi_{S \cup \{p\}}) - O(1)\cdot \sqrt{\median(p,S)}$.
To do this, we will construct a TSP tour $\Pi_{S \cup \{p\}}'$ of $G[S \cup \{p\}]$ that has the property required by \cref{:cl:getting-rid-of-p} and satisfies $w(\Pi_{S \cup \{p\}}') \geq w(\Pi_{S \cup \{p\}}) - O(1)\cdot \sqrt{\median(p,S)}$.

Firstly, if the tour $\Pi_{S \cup \{p\}}$ itself has the property required by \cref{:cl:getting-rid-of-p}, then the claim implies the existence of a tour $\Pi_{S}$ as we need, and then we are done.
So, for the rest of the proof, we assume that the tour $\Pi_{S \cup \{p\}}$ doesn't satisfy the property from \cref{:cl:getting-rid-of-p}, meaning that there is no point $p' \in S$ such that both $(p,p') \in E(\Pi_{S \cup \{p\}})$ and $d(p,p') \leq \median(p,S)$ hold.
Since $|S|\geq 2$, this implies that there are two distinct points $p_1,p_2 \in S$ such that $(p_1,p),(p,p_2) \in E(\Pi_{S \cup \{p\}})$ and $d(p,p_1),d(p,p_2) > \median(p,S)$.
Let $A \defeq \{p' \in S \mid d(p',p) \leq \median(p,S)\}$ and let $B \defeq \{p' \in S \mid d(p',p) > \median(p,S)\}$. (Notice that unlike in the previous section of the proof, the inequality in the definition of $B$ is strict and the one in the definition of $A$ is not.)
Therefore, by the definition of $\median(p,S)$, $|A| \geq |B|$, which means that $|A \cup \{p\}| > |B|$.
Thus, by \cref{:obs:edge-exists}, there must be some edge $(p_1',p') \in \Pi_{S \cup \{p\}}$ such that $p_1',p' \in A \cup \{p\}$.
Since $|S \cup \{p\}| \geq 3$, there must be some point $p_2'$ other than $p_1'$ such that $(p',p_2') \in E(\Pi_{S \cup \{p\}})$.
Furthermore, by the assumption $d(p,p_1),d(p,p_2) > \median(p,S)$, we have $p_1,p_2 \in B$.
Therefore, both incident edges of $p$ in $E(\Pi_{S \cup \{p\}})$ lead from it to a point $B$, while each of $p_1',p',p_2'$ has an incident edge of $E(\Pi_{S \cup \{p\}})$ leading from them to a point outside of $B$, so $p \notin \{p_1',p',p_2'\}$.
This means that the edges $(p_1,p),(p,p_2)$ are both distinct from the edges $(p_1',p'),(p',p_2')$.
Now, we construct the TSP tour $\Pi_{S \cup \{p\}}'$ of $G[S \cup \{p\}]$ to be like $\Pi_{S \cup \{p\}}$ except that the edges $(p_1,p),(p,p_2),(p_1',p'),(p',p_2')$ are replaced by $(p_1,p'),(p',p_2),(p_1',p),(p,p_2')$.
Then, the tour $\Pi_{S \cup \{p\}}'$ satisfies the property from \cref{:cl:getting-rid-of-p} because $(p_1',p) \in E(\Pi_{S \cup \{p\}}')$ and because $p_1' \in A \cup \{p\}$ means that $d(p,p_1') \leq \median(p,S)$.
Furthermore,
\begin{align*}
    &w(\Pi_{S \cup \{p\}}') - w(\Pi_{S \cup \{p\}}) \\
    = &w(p_1,p')+w(p',p_2)+w(p_1',p)+w(p,p_2') - w(p_1,p)-w(p,p_2)-w(p_1',p')-w(p',p_2')\\
    = &(w(p_1,p')-w(p_1,p)) + (w(p',p_2)-w(p,p_2)) + (w(p_1',p)-w(p_1',p')) + (w(p,p_2')-w(p,p_2')).
\end{align*}
Thus, by \cref{:obs:triangle-inequality-in-sqrt-metric},
\begin{align*}
    w(\Pi_{S \cup \{p\}}') - w(\Pi_{S \cup \{p\}})
    \geq (-w(p',p)) + (-w(p',p)) + (-w(p',p)) + (-w(p',p))
    &= -4w(p',p) \\
    &= -4\sqrt{d(p',p)}.
\end{align*}
So, since $p' \in A \cup \{p\}$,
\begin{align*}
    w(\Pi_{S \cup \{p\}}')
    \geq w(\Pi_{S \cup \{p\}}) - 4\sqrt{d(p',p)}
    &\geq w(\Pi_{S \cup \{p\}}) - 4\sqrt{\median(p,S)} \\
    &= w(\Pi_{S \cup \{p\}}) - O(\sqrt{\median(p,S)}).
\end{align*}

In summary, we proved that there exists a TSP tour $\Pi_{S \cup \{p\}}'$ of $G[S \cup \{p\}]$ that satisfies the conditions of \cref{:cl:getting-rid-of-p} and also satisfies $w(\Pi_{S \cup \{p\}}') \geq w(\Pi_{S \cup \{p\}}) - O(\sqrt{\median(p,S)})$.
So, by \cref{:cl:getting-rid-of-p}, there exists a TSP tour $\Pi_{S}$ of $G[S]$ with
\begin{align*}
    w(\Pi_{S})
    \geq w(\Pi_{S \cup \{p\}}') - O(\sqrt{\median(p,S)})
    \geq w(\Pi_{S \cup \{p\}}) - O(\sqrt{\median(p,S)}).
\end{align*}
This implies that the TSP tour $w(\Pi_{S})$ satisfies
\begin{align*}
    \frac{1}{2-\sqrt{2}}w(\Pi_{S})
    &\geq \frac{1}{2-\sqrt{2}}w(\Pi_{S \cup \{p\}}') - O(\sqrt{\median(p,S)}) \\
    &= \Phi_{\sqrt{\median}}(S \cup \{p\}) - O(\sqrt{\median(p,S)}),
\end{align*}
and thus that $\Phi_{\sqrt{\median}}(S) \geq \Phi_{\sqrt{\median}}(S \cup \{p\}) - O(\sqrt{\median(p,S)})$, as needed.
\end{proof}

\subsection{Split Procedure for \texorpdfstring{$\sqrt{\median}$}{sqrt(median)}-IP}

In this section, we prove the following lemma.

\begin{lemma}\label{:lem:split-procedure-for-median-IP}
There exists a deterministic polynomial-time algorithm that, given a metric space $(X,d)$ and a clustering $\cC$ of $X$, finds a cluster $C^* \in \cC$ and a partition $(C^*_1,C^*_2)$ of $C^*$ with $\Phi_{\sqrt{\median}}(C^*_1) + \Phi_{\sqrt{\median}}(C^*_1) \leq \Phi_{\sqrt{\median}}(C^*) - \left(\sum_{C \in \cC} \Phi_{\sqrt{\median}}(C)\right)/\poly(n)$.
\end{lemma}

To do this, we first need the following claim.

\begin{claim}\label{:cl:exists-large-median}
Let $(X,d)$ be a metric space, let $C \subseteq X$, and let $p,p' \in C$.
Then,
\[
 \median(p,C \setminus \{p\}) + \median(p',C \setminus \{p'\}) \geq d(p,p')
\]
\end{claim}
\begin{proof}
Let $A \defeq \{p'' \in C \mid d(p'',p) \leq \median(p,C\setminus\{p\})\}$ and $B \defeq \{p'' \in C \mid d(p'',p) > \median(p,C\setminus\{p\})\}$ and $A' \defeq \{p'' \in C' \mid d(p'',p) \leq \median(p,C\setminus\{p'\})\}$ and $B' \defeq \{p'' \in C' \mid d(p'',p) > \median(p,C\setminus\{p'\})\}$.
Then, $|A|\geq |B|$ and $|A'| \geq |B'|$, which means that $|A \cup \{p\}| > |B|$ and $|A' \cup \{p'\}| > B'$.
Since $(A \cup \{p\},B)$ and $(A' \cup \{p'\},B')$ are both partitions of $C$, we get that $|A \cup \{p\}| > |C|/2$ and $|A' \cup \{p'\}|>|C|/2$ and thus that there must be some point $p'' \in (A \cup \{p\}) \cap (A' \cup \{p'\})$.
Therefore, $d(p,p'')\leq \median(p,C\setminus\{p\})$ and $d(p',p'')\leq \median(p',C\setminus\{p'\})$.
So, by the triangle inequality,
\[
 d(p,p') \leq d(p,p'') + d(p',p'')
 \leq \median(p,C\setminus\{p\}) + \median(p',C\setminus\{p'\}).
\]
as we needed to prove.
\end{proof}
We are now ready to prove the lemma. The algorithm from the lemma is Algorithm \ref{:alg:split-for-median-IP}.

\begin{algorithm}
\caption{\split: The cluster splitting procedure from \cref{:lem:split-procedure-for-median-IP}.}\label{:alg:split-for-median-IP}
\KwData{$(X,d)$ where $|X|=n$, and a clustering $\cC$.}
\KwResult{a cluster $C^* \in \cC$ and a partition of $C^*$.}
$C^* \leftarrow \argmax_{C \in \cC} \max_{p,p' \in C} d(p,p')$\\
$p,p' \leftarrow \argmax_{p,p' \in C^*} d(p,p')$\\
$p'' \leftarrow \argmax_{p'' \in \{p,p'\}}\median(p'',C^*\setminus\{p''\})$\\
\Return{$C^*$ and the partition $(C^* \setminus \{p''\},\{p''\})$}
\end{algorithm}

\begin{proof}[Proof or \cref{:lem:split-procedure-for-median-IP}]
Firstly, By \cref{:def:potential-for-sqrt-median}, $\Phi_{\sqrt{\median}}(\{p''\})=0$.
Secondly, by \cref{:thm:potential-for-sqrt-median},
\[
 \Phi_{\sqrt{\median}}(C^* \setminus \{p''\}) \leq \Phi_{\sqrt{\median}}(C^*) - \sqrt{\median(p'',C^*\setminus\{p''\})}.
\]
Thirdly, by \cref{:cl:exists-large-median} and by the choice of $p''$,
\[
 \median(p'',C^*\setminus\{p''\}) \geq \frac{1}{2}(\median(p,C^*\setminus\{p\}) + \median(p',C^*\setminus\{p'\}))
 \geq \frac{1}{2}d(p,p').
\]
Putting these three together, then using the way that $p$ and $p'$ were chosen,
\begin{equation}\label{:eq:decrease-by-median-IP-split}
\begin{aligned}
    \Phi_{\sqrt{\median}}(C^* \setminus \{p''\}) + \Phi_{\sqrt{\median}}(\{p''\})
    &\leq \Phi_{\sqrt{\median}}(C^*) - \sqrt{\median(p'',C^*\setminus\{p''\})}\\
    &\leq \Phi_{\sqrt{\median}}(C^*) - \sqrt{\frac{1}{2}d(p,p')}\\
    &= \Phi_{\sqrt{\median}}(C^*) - \max_{C' \in \cC} \max_{p,p' \in C'} \sqrt{\frac{1}{2}d(p,p')}\\
    &= \Phi_{\sqrt{\median}}(C^*) - \max_{C' \in \cC} \max_{p,p' \in C'} \frac{1}{\sqrt{2}}w(p,p')
\end{aligned}
\end{equation}
Furthermore, by \cref{:cor:bound-on-median-potential-of-clustering},
\begin{align*}
    \sum_{C \in \cC} \Phi_{\sqrt{\median}}(C)
    \leq O(n \cdot \max_{C' \in \cC} \sqrt{\median(C')})
    = O(n \cdot \max_{C' \in \cC} \max_{p,p' \in C'} w(p,p')).
\end{align*}
Combining the last inequality and \cref{:eq:decrease-by-median-IP-split},
\[
 \Phi_{\sqrt{\median}}(C^* \setminus \{p''\}) + \Phi_{\sqrt{\median}}(\{p''\})
 \leq \Phi_{\sqrt{\median}}(C^*) - \frac{1}{O(n)}\left(\sum_{C \in \cC} \Phi_{\sqrt{\median}}(C)\right),
\]
as needed.
\end{proof}

\subsection{Bounding Potential Increase of Cluster Merges for \texorpdfstring{$\sqrt{\median}$}{sqrt(median)}-IP}

The goal of this section is to prove the following lemma.

\begin{lemma} \label{:lem:sqrt-median-merge-cost-bound}
For every metric space $(X,d)$ with $|X|=n$, every two disjoint non-empty clusters $C,C' \subseteq X$, and every point $p \in X$, if $\Phi_{\sqrt{\median}}$ is the potential function from \cref{:def:potential-for-sqrt-median}, then
\begin{align*}
    \Phi_{\sqrt{\median}}(C \cup C') \leq \Phi_{\sqrt{\median}}(C) + \Phi_{\sqrt{\median}}(C') + \poly(n)\cdot(\sqrt{\median(p,C)} + \sqrt{\median(p,C')}).
\end{align*}
\end{lemma}

To prove the above lemma, we will need the following claims.

\begin{claim}\label{:cl:upper-bounding-length-of-TSP-tour}
Given a metric space $(X,d)$ with $|X|=n$, a subset $S \subseteq X$ with $|S| \geq 2$, a point $p \in X$, and a radius $r > 0$, and a subset $A \subseteq S$ such that every point $p' \in A$ satisfies $w(p',p) \leq r$.
Then, every TSP tour $\Pi_{S}$ of $G[S]$, as defined in \cref{:def:potential-for-sqrt-median}, satisfies
\[
 w(\Pi_{S}) \leq 2\left(\sum_{p' \in S \setminus A} w(p',p)\right) + \poly(n) \cdot r
\]
\end{claim}
\begin{proof}
By \cref{:obs:triangle-inequality-in-sqrt-metric},
\begin{align*}
    w(\Pi_S)
    = \sum_{(p',p'') \in E(\Pi_S)} w(p',p'')
    &\leq \sum_{(p',p'') \in E(\Pi_S)} (w(p',p) + w(p,p'')) \\
    &= 2\sum_{p' \in S} w(p,p')\\
    &= 2\left(\sum_{p' \in S \setminus A} w(p,p')\right) + 2\left(\sum_{p' \in A} w(p,p')\right)\\
    &\leq 2\left(\sum_{p' \in S \setminus A} w(p,p')\right) + 2\left(\sum_{p' \in A} r\right)\\
    &= 2\left(\sum_{p' \in S \setminus A} w(p,p')\right) + \poly(n) \cdot r
\end{align*}
\end{proof}

\begin{claim}\label{:cl:exists-TSP-tour-with-large-length}
For every metric space $(X,d)$ with $|X|=n$, every subset $S \subseteq X$ with $|S| \geq 2$, and every a point $p \in X$, there exists a TSP tour $\Pi_{S}$ of $G[S]$, as defined in \cref{:def:potential-for-sqrt-median}, such that
\[
 w(\Pi_{S}) \geq 2\left(\sum_{p' \in S \mid d(p',p) > \median(p,S)} w(p',p)\right) - \poly(n) \cdot \sqrt{\median(p,S)}
\]
\end{claim}
\begin{proof}
Let $A \defeq \{p' \in S \mid d(p',p) \leq \median(p,S)\}$ and $B \defeq \{p' \in S \mid d(p',p) > \median(p,S)\}$.
Then, $|A| \geq |B|$, so let $A' \subseteq A$ be some subset of $A$ with $|A'|=|B|$.
We construct the tour $\Pi_S$ by starting with a path that alternates between $A'$ and $B$ until visiting all vertices of these sets, and then extending that path to visit the vertices of $A \setminus A'$, and finally coming back to the starting vertex.
Then, for every $p' \in B$, the two edges of $\Pi_S$ that are incident to $p'$ must both have their other endpoint in $A$, so, by \cref{:obs:triangle-inequality-in-sqrt-metric} and by the definition of $A$, they must each have length at least $w(p',p)-\sqrt{\median(p,S)}$.
Therefore,
\begin{align*}
    \sum_{e \in \Pi_S} w(e)
    &\geq 2\sum_{p' \in B} (w(p',p)-\sqrt{\median(p,S)}) \\
    &= 2\sum_{p' \in S \mid d(p',p) > \median(p,S)} (w(p',p)-\sqrt{\median(p,S)})\\
    &= 2\left(\sum_{p' \in S \mid d(p',p) > \median(p,S)} w(p',p)\right) - \poly(n) \cdot \sqrt{\median(p,S)}
\end{align*}
\end{proof}

Now, we are ready to prove \cref{:lem:sqrt-median-merge-cost-bound}.

\begin{proof}[Proof of \cref{:lem:sqrt-median-merge-cost-bound}]
Firstly, if $|C|=1$ or $|C'|=1$, then the lemma follows from \cref{:thm:potential-for-sqrt-median}.
So, for the rest of the proof, we assume that $|C| \geq 2$ and $|C'| \geq 2$.
Let $A \defeq \{p' \in C \mid d(p',p) \leq \median(p,C)\}$ and $B \defeq \{p' \in C \mid d(p',p) > \median(p,C)\}$ and $A' \defeq \{p' \in C' \mid d(p',p) \leq \median(p,C')\}$ and $B' \defeq \{p' \in C' \mid d(p',p) > \median(p,C')\}$.
By \cref{:cl:upper-bounding-length-of-TSP-tour}, since all points $p' \in A \cup A'$ satisfy $w(p',p) \leq \sqrt{\median(p,C)} + \sqrt{\median(p,C')}$,
\[
 \Phi_{\sqrt{\median}}(C \cup C') \leq \frac{1}{2-\sqrt{2}} \cdot 2\left(\sum_{p' \in B \cup B'} w(p',p)\right) + \poly(n) \cdot (\sqrt{\median(p,C)} + \sqrt{\median(p,C')}).
\]
Furthermore, by \cref{:cl:exists-TSP-tour-with-large-length},
\[
 \Phi_{\sqrt{\median}}(C) \geq \frac{1}{2-\sqrt{2}} \cdot 2\left(\sum_{p' \in B} w(p',p)\right) - \poly(n) \cdot \sqrt{\median(p,C)}
\]
and
\[
 \Phi_{\sqrt{\median}}(C') \geq \frac{1}{2-\sqrt{2}} \cdot 2\left(\sum_{p' \in B'} w(p',p)\right) - \poly(n) \cdot \sqrt{\median(p,C')}.
\]
Subtracting the last two inequalities from the first one,
\[
 \Phi_{\sqrt{\median}}(C \cup C') - \Phi_{\sqrt{\median}}(C) - \Phi_{\sqrt{\median}}(C') \leq \poly(n) \cdot (\sqrt{\median(p,C)} + \sqrt{\median(p,C')}),
\]
as we needed to prove.
\end{proof}

\subsection{Initialization for \texorpdfstring{$\sqrt{\median}$}{sqrt(median)}-IP}
In this section, we show that initializing using the greedy $k$-centers algorithm also works for the potential function $\Phi_{\sqrt{\median}}$ from \cref{:def:potential-for-sqrt-median}.
In other words, we prove the following lemma, which is an analogue to \cref{:lem:k-center-initialization}.

\begin{lemma}\label{:lem:initialization-for-median-IP}
    Given a metric space $(X,d)$ with $n$ points and a $k$-clustering $\cC$. If $\cC$ is a $2$-approximate optimal solution to the $k$-centers problem, then it $\poly(n)$-approximately minimizes the potential function $\Phi_{\sqrt{\median}}$; more precisely, $\Phi_{\sqrt{\median}}(\cC) \le \poly(n) \cdot \Phi_{\sqrt{\median}}(\cC^*)$, where $\cC^*$ is a $k$-clustering minimizing $\Phi_{\sqrt{\median}}$.
\end{lemma}

\begin{proof}
Since $\cC$ is a $2$-approximate $k$-centers solution, and since radius is a $2$-approximation of diameter,
\[
 \max_{C \in \cC} \diameter(C) \leq 4 \cdot \max_{C \in \cC^*} diamter(C).
\]
So, by \cref{:cor:bound-on-median-potential-of-clustering},
\begin{align*}
    \Phi_{\sqrt{\median}}(\cC)
    \leq O(n \cdot \max_{C \in \cC} diamter(C))
    \leq O(n \cdot \max_{C \in \cC^*} diamter(C))
    \leq O(n \cdot \Phi_{\sqrt{\median}}(\cC^*)).
\end{align*}
\end{proof}

\subsection{Putting it All Together}\label{:sec:concluding-sqrt-median-IP}

In this section, we prove \cref{:thm:constant-sqrt-median-IP} using the framework of \cref{:sec:general-framework}.
\begin{proof}[Proof of \cref{:thm:constant-sqrt-median-IP}]
In order to prove the theorem, we just need to show that for every metric space $(X,d)$, the function $\sqrt{\median}:X \times 2^X \to \bbR_{\geq 0}$ and the potential function $\Phi_{\sqrt{\median}}$ from \cref{:def:potential-for-sqrt-median} satisfy all the conditions of \cref{:thm:framework-for-general-f} for $\alpha = O(1)$.
We will now do this by going over each of the conditions:
\begin{enumerate}
    \item The first condition from \cref{:thm:framework-for-general-f} is satisfied by \cref{:thm:potential-for-sqrt-median}.
    \item It is clear that $\sqrt{\median}(p,S)$ can be computed in polynomial time given $(X,d)$, $S \subseteq X$, and $p$, so the second condition holds.
    \item Using \cref{:cor:bound-on-median-potential-of-clustering}, it is clear that a $\poly(n)$ approximation of $\Phi_{\sqrt{\median}}(\cC)$ can be computed in polynomial time, by simply computing $\max_{C \in \cC} \diameter(C)$. So, the third condition holds. We note that an $O(1)$-approximation can be achieved via an approximate maximum matching algorithm, but we will not prove this.
    \item By \cref{:lem:split-procedure-for-median-IP}, the fourth condition holds.
    \item By \cref{:lem:sqrt-median-merge-cost-bound}, the fifth condition holds.
    \item By \cref{:lem:initialization-for-median-IP} and the classic $2$-approximate $k$-centers algorithm, the sixth condition holds.
\end{enumerate}
In summary, we saw that all the conditions of \cref{:thm:framework-for-general-f} are satisfied for $f=\sqrt{\median}$ with $\alpha=O(1)$. So, by \cref{:thm:framework-for-general-f}, there exists a deterministic polynomial time algorithm that, given a metric space $(X,d)$ and a desired number of clusters $k$, computes a $(O(1),\sqrt{\median})$-IP stable $k$-clustering of $(X,d)$.
\end{proof}

\section{Proof of \texorpdfstring{\cref{:thm:existence-of-max-IP}}{Theorem 5}}
In this section, we prove \cref{:thm:existence-of-max-IP}, showing that every metric space $(X,d)$ and number of clusters $2 \leq k \leq |X|$ admit a $(1,\max)$-IP stable $k$-clustering, where $\max(p,S)\defeq \max_{p' \in S}d(p,p')$.
We show that the natural local search algorithm for this problem, as detailed in Algorithm \ref{:alg:local-search-for-max-IP}, terminates. However, we note that the algorithm may not terminate within polynomial time.

\begin{algorithm}
\caption{\maxIPLS: Local Search For Max-IP.}\label{:alg:local-search-for-max-IP}
\KwData{$(X,d)$, $2 \leq k \leq n =|X|$.}
\KwResult{A $(1,\max)$-IP stable $k$-clustering of $(X,d)$.}
\textbf{initialize} an arbitrary $k$-clustering $\mathcal{C}$ of $(X,d)$ \\
\While{exists $p$ and $C'$ s.t. $C(p)\neq\{p\}$ and $\max(p,C(p) \setminus \{p\}) > \max(p,C')$ \label{:line:local-search-for-max-IP:loop-condition}}{
    \textbf{move} $p$ to cluster $C'$\\
}
\textbf{return} the current clustering
\end{algorithm}

First, we need the following definition and claim. We will use these to prove \cref{:thm:existence-of-max-IP}, and then we will prove \cref{:cl:max-IP-steps-reduce-potential}.

\begin{definition}[Potential function for max-IP]\label{:def:potential-for-max-IP}
Consider the complete clique graph $G=(X,E)$ on $X$ where each edge $(p,p') \in E$ has length $w(p,p') = d(p,p')$.
Label the edges of $G$ so that $w(e_1)\geq w(e_2) \geq \ldots \geq w(e_{|E|})$.
For every clustering $\cC$ of $X$, let $s_{\cC} \in \{0,1\}^{|E|}=\{0,1\}^{\binom{|X|}{2}}$ be the string whose $i$th character is a $1$ if both endpoints of $e_i$ are in the same cluster of $\cC$ and is $0$ otherwise.
\end{definition}
\begin{claim}\label{:cl:max-IP-steps-reduce-potential}
At each step of Algorithm \ref{:alg:local-search-for-max-IP}, the string $s_{\cC}$ corresponding to the maintained clustering $\cC$, as in \cref{:def:potential-for-max-IP}, goes down in lexicographic order.
\end{claim}
\begin{proof}
Given a specific step of Algorithm \ref{:alg:local-search-for-max-IP}, let $\cC$ be the state of the clustering at the beginning of the step, let $p$ and $C'$ be the point and cluster selected in this step, let $C$ be the cluster containing $p$ at the beginning of the step, and let $\cC'$ be the state of the clustering at the end of the step.
Let
\[
 I_{\mathrm{reduced}} \defeq \{i \mid \text{the $i$th character of $s_{\cC}$ is $1$ and the $i$th character of $s_{\cC'}$ is $0$}\}
\]
and
\[
 I_{\mathrm{increased}} \defeq \{i \mid \text{the $i$th character of $s_{\cC}$ is $0$ and the $i$th character of $s_{\cC'}$ is $1$}\}.
\]
In order to show that $s_{\cC'}$ is smaller than $s_{\cC}$ in lexicographic order, we need to show that $I_{\mathrm{reduced}}$ is non-empty and that $\min\{i \in I_{\mathrm{reduced}}\} < \min\{i \in I_{\mathrm{increased}}\}$.
Since the difference between $\cC$ and $\cC'$ is that $p$ was moved from $C$ to $C'$,
\[
 I_{\mathrm{reduced}} = \{i \mid \text{$e_i$ has one endpoint in $C\setminus \{p\}$ and the other endpoint is $p$}\}
\]
and
\[
 I_{\mathrm{increased}} = \{i \mid \text{$e_i$ has one endpoint in $C'$ and the other endpoint is $p$}\}.
\]
Since $C \neq \{p\}$, the set $I_{\mathrm{reduced}}$ cannot be empty.
Furthermore, since the edges are ordered in non-increasing order of their length, and since $\max(p,C\setminus\{p\}) > \max(p,C')$, we must have $\min\{i \in I_{\mathrm{reduced}}\} < \min\{i \in I_{\mathrm{increased}}\}$, as we needed.

In summary, we showed that for each step of Algorithm \ref{:alg:local-search-for-max-IP}, there is at least one $i$ such that the step turns the $i$th character of the string $s_{\cC}$ from $1$ to $0$, and that the minimum $i$ for which this happens is smaller than the minimum $i$ for which the step turns the $i$th character of the string $s_{\cC}$ from $0$ to $1$.
Thus, each step of the algorithm reduces the string $s_{\cC}$ in lexicographic order.
This concludes the proof of \cref{:cl:max-IP-steps-reduce-potential}
\end{proof}

Now, we are ready to prove \cref{:thm:existence-of-max-IP}.
\begin{proof}[Proof of \cref{:thm:existence-of-max-IP}]
Consider an execution of Algorithm \ref{:alg:local-search-for-max-IP} on $(X,d)$ with parameter $k$.
\cref{:cl:max-IP-steps-reduce-potential} implies that the clustering $\cC$ maintained by the algorithm never returns to a state that it was in before.
Therefore, since there are a finite number of possible clusterings of $(X,d)$, Algorithm \ref{:alg:local-search-for-max-IP} must terminate in finite time.
When the algorithm terminates, the condition of the while loop in line \ref{:line:local-search-for-max-IP:loop-condition} must not be satisfied, which exactly means that the current clustering $\cC$ maintained by the algorithm is $(1,\max)$-IP stable.
Furthermore, it is not hard to see that the number of clusters in $\cC$ always remains $k$.
Thus, the algorithm terminates and outputs a $(1,\max)$-IP stable $k$-clustering of $(X,d)$.
Therefore, such a clustering exists.
\end{proof}
\end{document}

%% file: aistats-defs.tex
\usepackage{amsmath}
\usepackage{mathtools}
\usepackage{amssymb}
\usepackage{amsthm}
\usepackage{enumitem}
\usepackage{thmtools,thm-restate}
\usepackage{graphicx}
\usepackage{xcolor}
\usepackage{hyperref}
\hypersetup{colorlinks=true,allcolors=blue}
\usepackage[capitalise,noabbrev]{cleveref}
\usepackage{caption}
\usepackage{subcaption}
\usepackage{array}
\usepackage[linesnumbered,ruled,norelsize]{algorithm2e}
\SetKwComment{Comment}{/* }{ */}
\usepackage[round]{natbib}

\newtheorem{theorem}{Theorem}

\newtheorem{lemma}{Lemma}

\newtheorem{claim}{Claim}

\newtheorem{corollary}{Corollary}

\newtheorem{observation}{Observation}

\theoremstyle{definition}
\newtheorem{definition}{Definition}

\DeclareMathOperator*{\argmax}{arg\,max}
\DeclareMathOperator*{\argmin}{arg\,min}

\def\cC{{\cal C}}

\def\cS{{\cal S}}

\def\bbN{{\mathbb N}}

\def\bbR{{\mathbb R}}

\def\createTree{{\textsc{CreateTree}}}
\def\split{{\textsc{Split}}}
\def\maxIPLS{{\textsc{MaxIP-LS}}}
\def\fastLS{{\textsc{Fast-LS}}}
\def\CalcPotential{{\textsc{CalcPotential}}}
\def\CalcAverage{{\textsc{CalcAverage}}}
\def\CalcCentral{{\textsc{CalcCentralPoint}}}

\def\split{{\textsc{Split}}}
\def\fastSplit{{\textsc{FastSplit}}}

\DeclarePairedDelimiter\ceil{\lceil}{\rceil}

\def\poly{\mathrm{poly}}
\def\polylog{\mathrm{polylog}}

\newcommand{\expect}[2][]{\mathbb{E}_{#1}\left[{#2}\right]}
\newcommand*{\defeq}{\stackrel{\text{def}}{=}}
\newcommand{\quotes}[1]{{``{#1}"}}

\DeclareMathOperator{\avg}{avg}
\DeclareMathOperator{\median}{median}
\DeclareMathOperator{\diameter}{diameter}
\newcommand{\clusterToRecompute}{\ensuremath{\mathcal{C}_R}}